\documentclass[12pt,onecolumn]{IEEEtran}

\usepackage{amsmath,amsfonts,amssymb,enumerate}
\usepackage{color}
\usepackage{lscape}
\usepackage{latexsym}
\usepackage{multirow}
\usepackage{rotating}
\usepackage{threeparttable}
\usepackage{graphicx}
\usepackage{algpseudocode,algorithm}
\usepackage[margin=1in]{geometry}

\newtheorem{theorem}{\bf Theorem}
\newtheorem{lemma}[theorem]{\bf Lemma}
\newtheorem{prop}[theorem]{\bf Proposition}

\newenvironment{proof}{\noindent{\em Proof:}}{\quad \hfill$\Box$\vspace{2ex}}

\newtheorem{definition}[theorem]{\bf Definition}

\begin{document}

\title{Blind One-Bit Compressive Sampling}
\author{Lixin Shen and Bruce W. Suter

\thanks{This research is supported in
part by an award from National Research Council via the Air Force Office of Scientific Research and by the US National Science Foundation under grant DMS-1115523.}

\thanks{L. Shen (Corresponding author) is with Department of Mathematics, Syracuse University, Syracuse, NY 13244. (Email: {lshen03@syr.edu}). }
\thanks{Bruce W. Suter is with Air Force Research Laboratory. AFRL/RITB, Rome, NY 13441-4505 (Email: {bruce.suter@rl.af.mil}).}
\thanks{Copyright (c) 2012 IEEE. Personal use of this material is permitted. However, permission to use this material for any other purposes must be obtained from the IEEE by sending a request to pubs-permissions@ieee.org.}
}

\maketitle
\begin{abstract}
The problem of 1-bit compressive sampling is addressed in this paper. We introduce an optimization model for reconstruction of sparse signals from 1-bit measurements. The model targets a solution that has the least $\ell_0$-norm among all signals satisfying consistency constraints stemming from the 1-bit measurements. An algorithm for solving the model is developed. Convergence analysis of the algorithm is presented. Our approach is to obtain a sequence of optimization problems by successively approximating the $\ell_0$-norm and to solve resulting problems by exploiting the proximity operator. We examine the performance of our proposed algorithm and compare it with the binary iterative hard thresholding (BIHT) \cite{Jacques-Laska-Boufounos-Baraniuk:11} a state-of-the-art algorithm for 1-bit compressive sampling reconstruction. Unlike the BIHT, our model and algorithm does not require a prior knowledge on the sparsity of the signal. This makes our proposed work a promising practical approach for signal acquisition.
\end{abstract}


\section{Introduction}
Compressive sampling is a recent advance in signal acquisition \cite{Candes-Romberg-Tao:CPAM:06,Candes-Tao:IEEE-TIT:06}. It provides a method to reconstruct a sparse signal $x \in \mathbb{R}^n$ from linear measurements
\begin{equation}\label{model:linear}
y = \Phi x,
\end{equation}
where $\Phi$ is a given $m \times n$ measurement matrix with $m < n$ and $y \in \mathbb{R}^m$ is the measurement vector acquired. The objective of compressive sampling is to deliver an approximation to $x$ from $y$ and $\Phi$. It has been demonstrated that the sparse signal $x$ can be recovered exactly from $y$ if $\Phi$ has Gaussian i.i.d. entries and satisfies the restricted isometry property \cite{Candes-Tao:IEEE-TIT:06}. Moreover, this sparse signal can be identified as a vector that has the smallest $\ell_0$-norm among all vectors yielding the same measurement vector $y$ under the measurement matrix $\Phi$.

However, the success of the reconstruction of this sparse signal is based on the assumption that the measurements have infinite bit precision. In realistic settings, the measurements are never exact and must be discretized prior to further signal analysis. In practice, these measurements are  quantized, a mapping from a continuous real value to a discrete value over some finite range. As usual, quantization inevitably introduces errors in measurements. The problem of estimating a sparse signal from a set of quantized measurements has been addressed  in recent literature. Surprisedly, it has been demonstrated theoretically and numerically that 1-bit per measurement is enough to retain information for sparse signal reconstruction. As pointed out in \cite{Boufounos-Baraniuk:08,Jacques-Laska-Boufounos-Baraniuk:11}, quantization to 1-bit measurements is appealing in practical applications. First, 1-bit quantizers are  extremely inexpensive hardware devices that test values above or below zeros, enabling simple, efficient, and fast quantization. Second, 1-bit quantizers are robust to a number of non-linear distortions applied to measurements. Third, 1-bit quantizers do not suffer from dynamic range issues. Due to these attractive properties of 1-bit quantizers, in this paper we will develop efficient algorithms for reconstruction of sparse signals from 1-bit measurements.

The 1-bit compressive sampling framework originally introduced in \cite{Boufounos-Baraniuk:08} is briefly described as follows. Formally, it can be written as
\begin{equation}\label{model:1bit}
y = A(x):=\mathrm{sign}(\Phi x),
\end{equation}
where the function $\mathrm{sign}(\cdot)$ denotes the sign of the variable, element-wise, and zero values are assigned to be $+1$. Thus, the measurement operator $A$, called a 1-bit scalar quantizer, is a mapping from $\mathbb{R}^n$ to the Boolean cube $\{-1,1\}^m$. Note that the scale of the signal has been lost during the quantization process. We search for a sparse signal $x^\star$ in the unit ball of  $\mathbb{R}^m$ such that the sparse signal $x^\star$ is consistent with our knowledge about the signal and measurement process, i.e., $A(x^\star)=A(x)$.

The problem of reconstructing a sparse signal from its 1-bit measurements is generally non-convex, and therefore it is a challenge to develop an algorithm that can find a desired solution. Nevertheless, since this problem was introduced in \cite{Boufounos-Baraniuk:08} in 2008, there are several algorithms that have been developed for attacking it  \cite{Boufounos-Baraniuk:08,Laska-Wen-Yin-Baraniuk:IEEESP:11,Plan-Vershynin:11,Yan-Yang-Osher:IEEESP:12}. Among those existing 1-bit compressive sampling algorithms, the binary iterative hard thresholding (BIHT) \cite{Jacques-Laska-Boufounos-Baraniuk:11} exhibits its superior performance in both reconstruction error and as well as consistency via numerical simulations over the algorithms in \cite{Boufounos-Baraniuk:08,Laska-Wen-Yin-Baraniuk:IEEESP:11}. When there are a lot of sign flips in the measurements, a method based on adaptive outlier pursuit for 1-bit compressive sampling was proposed in \cite{Yan-Yang-Osher:IEEESP:12}. The algorithms in  \cite{Jacques-Laska-Boufounos-Baraniuk:11,Yan-Yang-Osher:IEEESP:12} require the sparsity of the desired signal to be given in advance. This requirement, however, is hardly satisfied in practice. By keeping only the sign of the measurements, the magnitude of the signal is lost. The models associated with the aforementioned algorithms seek sparse vectors $x$ satisfying consistency constraints \eqref{model:1bit} in the unit sphere. As a result, these models are essentially non-convex and non-smooth.  In \cite{Plan-Vershynin:11}, a convex minimization problem is formulated for reconstruction of sparse signals from 1-bit measurements and is solved by linear programming. The details of the above algorithms will be briefly reviewed in the next section.

In this paper, we introduce a new $\ell_0$ minimization model over a convex set determined by consistency constraints for 1-bit compressive sampling recovery and develop an algorithm for solving the proposed model. Our model does not require prior knowledge on the sparsity of the signal, therefore, is referred to as the blind 1-bit compressive sampling model. Our approach for dealing with our proposed model is to obtain a sequence of optimization problems by successively approximating the $\ell_0$-norm and to solve resulting problems by exploiting the proximity operator \cite{moreau:BSMF:65}. Convergence analysis of our algorithm is presented.

This paper is organized as follows. In Section~\ref{sec:background}  we review and comment current 1-bit compressive sampling models and then introduce our own model by assimilating advantages of existing models. Heuristics for solving the proposed model are discussed in Section~\ref{sec:characterization}. Convergence analysis of the algorithm for the model is studied in Section~\ref{sec:convergence}. A numerical implementable algorithm for the model is presented in Section~\ref{sec:complete-algorithm}.  The performance of our algorithm is demonstrated and compared with the BIHT in Section~\ref{sec:experiments}. We present our conclusion in Section~\ref{sec:conclusion}.

\section{Models for One-Bit Compressive Sampling}\label{sec:background}
In this section, we begin with reviewing existing models for reconstruction of sparse signals from 1-bit measurements. After analyzing these models, we propose our own model that assimilates the advantages of the existing ones.

Using matrix notation, the 1-bit measurements in \eqref{model:1bit} can be equivalently expressed as
\begin{equation}\label{model:1bit-matrix}
Y\Phi x \ge 0,
\end{equation}
where $Y := \mathrm{diag}(y)$ is an $m \times m$ diagonal matrix whose $i$th diagonal element is the $i$th entry of $y$. The expression $Y\Phi x \ge 0$ in \eqref{model:1bit-matrix} means that all entries of  the vector $Y\Phi x$ are no less than $0$. Hence, we can treat the 1-bit measurements as sign constraints that should be enforced in the construction of the signal $x$ of interest. In what follows, equation~\eqref{model:1bit-matrix} is referred to as sign constraint or consistency condition, interchangeably.

The optimization model for reconstruction of a sparse signal from 1-bit measurements in \cite{Boufounos-Baraniuk:08} is
\begin{equation}\label{model:boufounos}
\min \|x\|_1 \quad \mbox{s.t.} \quad Y\Phi x \ge 0 \quad \mbox{and} \quad \|x\|_2=1,
\end{equation}
where $\|\cdot\|_1$ and $\|\cdot\|_2$ denote the $\ell_1$-norm and the $\ell_2$-norm of a vector, respectively. In model~\eqref{model:boufounos}, the $\ell_1$-norm objective function is used to favor sparse solutions, the sign constraint $Y\Phi x \ge 0$ is used to impose the consistency between the 1-bit measurements and the solution, the constraint $\|x\|_2=1$ ensures a nontrivial solution lying on the unit $\ell_2$ sphere.

Instead of solving model~\eqref{model:boufounos} directly, a relaxed version of model~\eqref{model:boufounos}
\begin{equation}\label{model:boufounos-relaxed}
\min \left\{\lambda \|x\|_1 +  \sum_{i=1}^m h((Y\Phi x )_i)\right\} \quad \mbox{s.t.} \quad  \|x\|_2=1
\end{equation}
was proposed  in \cite{Boufounos-Baraniuk:08} and solved by employing a variation of the fixed point continuation algorithm in \cite{Hale-Yin-Zhang:SAIMOPT-08}. Here $\lambda$ is a regularization parameter and $h$ is chosen to be the one-sided $\ell_1$ (or $\ell_2$) function, defined at $z \in \mathbb{R}$ as follows
\begin{equation}\label{one-sided-2}
h(z): = \left\{
          \begin{array}{ll}
            |z| \;\;( \mbox{or}\; \frac{1}{2}z^2), & \hbox{if $z<0$;} \\
            0, & \hbox{otherwise.}
          \end{array}
        \right.
\end{equation}
We remark that the one-sided $\ell_2$ function was adopted in \cite{Boufounos-Baraniuk:08} due to its convexity and smoothness properties that are required by a fixed point continuation algorithm.

In \cite{Laska-Wen-Yin-Baraniuk:IEEESP:11} a restricted-step-shrinkage algorithm was proposed for solving model~\eqref{model:boufounos}. This algorithm is similar in sprit to trust-region methods for nonconvex optimization on the unit sphere and has a provable convergence guarantees.

Binary iterative hard thresholding (BIHT) algorithms were recently introduced for reconstruction of sparse signals from 1-bit measurements in \cite{Jacques-Laska-Boufounos-Baraniuk:11}. The BIHT algorithms are developed for solving the following constrained optimization model
\begin{equation}\label{model:BIHT}
\min \sum_{i=1}^m h((Y\Phi x )_i) \quad \mbox{s.t.} \quad \|x\|_0 \le s \quad \mbox{and} \quad \|x\|_2=1,
\end{equation}
where $h$ is defined by equation \eqref{one-sided-2}, $s$ is a positive integer, and the $\ell_0$-norm $\|x\|_0$ counts the number of non-zero entries in $x$. Minimizing the objective function of model \eqref{model:BIHT} enforces the consistency condition \eqref{model:1bit-matrix}. The BIHT algorithms for model \eqref{model:BIHT} are a simple modification of the iterative thresholding algorithm proposed in \cite{Blumensath-Davies:ACHA:09}. It was shown numerically that the BIHT algorithms perform significantly better than the other aforementioned algorithms in \cite{Boufounos-Baraniuk:08,Laska-Wen-Yin-Baraniuk:IEEESP:11} in terms of both reconstruction error as well as consistency. Numerical experiments in \cite{Jacques-Laska-Boufounos-Baraniuk:11} further show that the BIHT algorithm with $h$ being the one-sided $\ell_1$ function performs better in low noise scenarios while the BIHT algorithm with $h$ being the one-sided $\ell_2$ function perform better in high noise scenarios.  Recently, a robust method for recovering signals from 1-bit measurements using adaptive outlier pursuit was proposed for the measurements having noise (i.e., sign flips) in \cite{Yan-Yang-Osher:IEEESP:12}.

The algorithms reviewed above for 1-bit compressive sampling are developed for optimization problems having convex objective functions and non-convex constraints. In \cite{Plan-Vershynin:11} a convex optimization program for reconstruction of sparse signals from 1-bit measurements was introduced as follows:
\begin{equation}\label{model:plan}
\min \|x\|_1 \quad \mbox{s.t.} \quad Y\Phi x \ge 0 \quad \mbox{and} \quad \|\Phi x\|_1=p,
\end{equation}
where $p$ is any fixed positive number. The first constraint $Y\Phi x \ge 0$ requires that a solution to model \eqref{model:plan} should be consistent with the 1-bit measurements. If a vector $x$ satisfies the first constraint, so is $ax$ for all $0< a <1$. Hence, an algorithm for minimizing the $\ell_1$-norm by only requiring consistency with the measurements will yield the solution $x$ being zero. The second constraint $\|\Phi x\|_1=p$ is then used to prevent model \eqref{model:plan} from returning a zero solution, thus, resolves the amplitude ambiguity. By taking the first constraint into consideration, we know that $\|\Phi x\|_1 = \langle y, \Phi x \rangle$, therefore, the second constraint becomes $\langle \Phi^\top y, x \rangle = p$. This confirms that both objective function and constraints of model \eqref{model:plan} are convex. It was further pointed out in \cite{Plan-Vershynin:11} that model \eqref{model:plan} can be cast as a linear program. As comparing model \eqref{model:plan} with model~\eqref{model:boufounos}, both the constraint $\|x\|_2=1$ in model~\eqref{model:boufounos} and the constraint $\|\Phi x\|_1=p$ in model~\eqref{model:plan}, the only difference between both models, enforce a non-trivial solution. However, as we have already seen, model~\eqref{model:plan} with the constraint $\|\Phi x\|_1=p$ can be solved by a computationally tractable algorithm.

Let us further comment on models \eqref{model:BIHT} and \eqref{model:plan}. First, the sparsity constraint in model \eqref{model:BIHT} is impractical since the sparsity of the underlying signal is unknown in general. Therefore, instead of imposing this sparse constraint, we consider to minimize an optimization model having the $\ell_0$-norm as its objective function. Second, although model~\eqref{model:plan} can be tackled by efficient linear programming solvers and the solution of model~\eqref{model:plan} preserves the effective sparsity of the underlying signal (see \cite{Plan-Vershynin:11}), the solution is not necessarily sparse in general as shown in our numerical experiments (see Section~\ref{sec:experiments}). Motivated by the aforementioned models and the associated algorithms,  we plan in this paper to reconstruct sparse signals from 1-bit measurements via solving the following constrained optimization model
\begin{equation}\label{model:SSZ}
\min \|x\|_0 \quad \mbox{s.t.} \quad Y\Phi x \ge 0 \quad \mbox{and} \quad \|\Phi x\|_1=p,
\end{equation}
where $p$ is again a arbitrary positive number. This model has the $\ell_0$-norm as its objective function and inequality $Y\Phi x \ge 0$ and equality $\|\Phi x\|_1=p$ as its convex constraints.

We remark that the actual value of $p$ is not important as long as it is positive. More precisely, suppose that $\mathcal{S}$ and $\mathcal{S}^\diamond$ are two sets collecting all solutions of model~\eqref{model:SSZ} with $p=1$ and  $p=p^\diamond>0$, respectively. If $x \in \mathcal{S}$, that is, $Y\Phi x \ge 0$ and $\|\Phi x\|_1=1$,  then, by denoting $x^\diamond:={p^\diamond}x$, it can be verified that  $\|x^\diamond\|_0=\|x\|_0$, $Y\Phi x^\diamond \ge 0$, and $\|\Phi x^\diamond \|_1=p^\diamond$. That indicates  $x^\diamond \in \mathcal{S}^\diamond$. Therefore, we have that ${p^\diamond}\mathcal{S} \subset \mathcal{S}^\diamond$. Conversely, we can show that $\mathcal{S}^\diamond \subset {p^\diamond}\mathcal{S}$ by reverting above steps. Hence, ${p^\diamond}\mathcal{S} = \mathcal{S}^\diamond$.  Without loss of generality, the positive number $p$ is always assumed to be $1$ in the rest part of the paper.

To close this section,  we compare  model \eqref{model:BIHT} and our proposed model~\eqref{model:SSZ} in the following result.
\begin{prop}\label{prop:models-relation}
Let $y \in \mathbb{R}^m$ be the 1-bit measurements from an $m\times n$ measurement matrix $\Phi$ via equation~\eqref{model:1bit} and let $s$ be a positive integer. Assume that the vector $x \in \mathbb{R}^n$ is a solution to model~\eqref{model:SSZ}. Then  model~\eqref{model:BIHT} has the unit vector $\frac{x}{\|x\|_2}$ as its solution if $\|x\|_0 \le s$; otherwise,  model~\eqref{model:BIHT} can not have a solution satisfying the consistency constraint if $\|x\|_0 > s$.
\end{prop}
\begin{proof}\ \
Since the vector $x$ is a solution to model~\eqref{model:SSZ}, then $x$ satisfies the consistency constraint $Y\Phi x \ge 0$. Hence, it, together with definition of $h$ in \eqref{one-sided-2}, implies that
$$
\sum_{i=1}^m h\left(\left(Y\Phi \frac{x}{\|x\|_2} \right)_i\right)=0.
$$
We further note that $\left\|\frac{x}{\|x\|_2}\right\|_0 = \|x\|_0$ and $\left\|\frac{x}{\|x\|_2}\right\|_2 =1$. Hence, the vector $\frac{x}{\|x\|_2}$ is a solution of model~\eqref{model:BIHT} if $\|x\|_0 \le s$.

On the other hand, if $\|x\|_0 > s$ then all solutions to model~\eqref{model:BIHT} do not satisfy the consistency constraint. Suppose this statement is \emph{false}. That is,  there exists a solution of model~\eqref{model:BIHT}, say $x^\sharp$, such that  $Y\Phi x^\sharp \ge 0$, $\|x^\sharp\|_0 \le s$, and $\|x^\sharp\|_2=1$ hold. Set $x^\diamond: = \frac{x^\sharp}{\|\Phi x^\sharp\|_1}$. Then $\|x^\diamond\|_0=\|x^\sharp\|_0 \le s$, $Y\Phi x^\diamond \ge 0$, and $\|\Phi x^\diamond\|_1=1$. Since $\|x^\diamond\|_0 < \|x\|_0$, it turns out that $x$ is not a solution of model~\eqref{model:SSZ}. This contracts our assumption on the vector $x$. This completes the proof of the result.
\end{proof}

From Proposition~\ref{prop:models-relation}, we can see that the sparsity $s$ for model~\eqref{model:BIHT} is critical. If $s$ is set too large, a solution to model~\eqref{model:BIHT} may not be the sparsest solution satisfying the consistency constraint; if $s$ is set too small, solutions to model~\eqref{model:BIHT} cannot satisfy the consistency constraint. In contrast, our model~\eqref{model:SSZ} does not require the sparsity constraint used in model~\eqref{model:BIHT} and delivers the sparsest solution satisfying the consistency constraint. Therefore, these properties make our model more attractive for 1-bit compressive sampling than the BIHT. Since sparsity of the underlying signal is not specified in advance in model \eqref{model:SSZ}, we refer it to as \emph{blind 1-bit compressive sampling model}.

\section{An Algorithm for the Blind 1-Bit Compressive Sampling}\label{sec:characterization}
In this section, we will develop algorithms for the proposed model~\eqref{model:SSZ}. We first reformulate model~\eqref{model:SSZ} as an unconstrained optimization problem via the indicator function of a closed convex set in $\mathbb{R}^{m+1}$. It turns out that the objective function of this unconstrained optimization problem is the sum of the $\ell_0$-norm and the indicator function composing with a matrix associated with the 1-bit measurements. Instead of directly solving the  unconstrained optimization problem we use some smooth concave functions to approximate the $\ell_0$-norm and then linearize the concave functions. The resulting model can be viewed as an optimization problem of minimizing a weighted $\ell_1$-norm over the closed convex set. The solution of this resulting model is served as a new point at which the concave functions will be linearized. This process is repeatedly performed until a certain stopping criteria is met. Several concrete examples for approximating the $\ell_0$-norm are provided at the end of this section.

We begin with introducing our notation and recalling some background from convex analysis. For the $d$-dimensional Euclidean space $\mathbb{R}^d$, the class of all lower semicontinuous convex functions $f: \mathbb{R}^d \rightarrow (-\infty, +\infty]$ such that $\mathrm{dom} f:=\{x \in \mathbb{R}^d: f(x) <+\infty\} \neq \emptyset$ is denoted by $\Gamma_0(\mathbb{R}^d)$. The indicator function of a closed convex set $C$ in  $\mathbb{R}^d$  is defined, at $u \in \mathbb{R}^d$, as
$$
\iota_C(u): =\left\{
               \begin{array}{ll}
                 0, & \hbox{if $u\in C$;} \\
                 +\infty, & \hbox{otherwise.}
               \end{array}
             \right.
$$
Clearly, $\iota_C$ is in $\Gamma_0(\mathbb{R}^d)$ for any closed nonempty convex set $C$.

Next, we reformulate model~\eqref{model:SSZ} as an unconstrained optimization problem. To this end, from the $m \times n$ matrix $\Phi$ and the $m$-dimensional vector $y$ in equation~\eqref{model:1bit}, we define an $(m+1)\times n$ matrix
\begin{equation}\label{eq:B}
B := \begin{bmatrix} \mathrm{diag}(y) \\ y^\top \end{bmatrix} \Phi
\end{equation}
and a subset of $\mathbb{R}^{m+1}$
\begin{equation}\label{eq:C}
\mathcal{C}:=\{z:  z_{m+1}=1 \; \mbox{and}\; z_i\ge 0, \; i=1,2,\ldots, m\},
\end{equation}
respectively. Then a vector $x$ satisfies the two constraints of model~\eqref{model:SSZ} if and only if the vector $Bx$ lies in the set $\mathcal{C}$. Hence, model~\eqref{model:SSZ} can be rewritten as
\begin{equation}\label{model:matrix}
\min\{\|x\|_0 + \iota_\mathcal{C}(Bx): x \in \mathbb{R}^n \}.
\end{equation}

Problem~\eqref{model:matrix} is known to be NP-complete due to the non-convexity of the $\ell_0$-norm. Thus, there is a need for an algorithm that can pick the sparsest vector $x$ satisfying the relation $Bx \in \mathcal{C}$. To attack this $\ell_0$-norm optimization problem, a common approach that appeared in recent literature is to approximate the $\ell_0$-norm by its computationally feasible approximations. In the context of compressed sensing, we review several popular choices for defining the $\ell_0$-norm as the limit of a sequence. More precisely, for a positive number $\epsilon \in (0,1)$, we consider separable concave functions of the form
\begin{equation}\label{def:F}
F_\epsilon (x) : = \sum_{i=1}^n f_\epsilon (|x_i|), \quad  x \in \mathbb{R}^n,
\end{equation}
where $f_\epsilon: \mathbb{R}_+ \rightarrow \mathbb{R}$ is strictly increasing, concave, and twice continuously differentiable such that
\begin{equation}\label{def:F-L0}
\lim_{\epsilon \rightarrow 0+} F_\epsilon (x)  = \|x\|_0, \quad  \mbox{for all} \quad x  \in \mathbb{R}^n.
\end{equation}
Since the function $f_\epsilon$ is concave and smooth on $\mathbb{R}_+:=[0, \infty)$, it can be majorized by a simple function formed by its first-order Taylor series expansion at a arbitrary point. Write $\mathcal{F}_\epsilon(x, v):=F_\epsilon (v) + \langle \nabla F_\epsilon (|v|), |x|-|v| \rangle$.  Therefore, at any point $v \in \mathbb{R}^n$ the following inequality holds
\begin{equation}\label{eq:L-expansion}
F_\epsilon (x) < \mathcal{F}_\epsilon(x, v)
\end{equation}
for all $x \in \mathbb{R}^n$ with $|x| \neq |v|$. Here, for a vector $u$, we use $|u|$ to denote a vector such that each element of $|u|$ is the absolute value of the corresponding element of $u$. Clearly, when $v$ is close enough to $x$, $\mathcal{F}_\epsilon(x, v)$ the expression on the right-hand side of \eqref{eq:L-expansion} provides a reasonable approximation to the one on its left-hand side. Therefore, it is considered as a computationally feasible approximation to the $\ell_0$-norm of $x$. With such an approximation, a simplified problem is solved and its solution is used to formulate another simplified problem which is closer to the ideal problem \eqref{model:matrix}. This process is then repeated until the solutions to the simplified problems become stationary or meet a termination criteria. This procedure is summarized in Algorithm~\ref{alg:matrix}.

\begin{algorithm}\caption{(Iterative scheme for model~\eqref{model:matrix})}\label{alg:matrix}
 \begin{algorithmic}[htb]

   \State Initialization: choose $\epsilon \in (0,1)$ and let $x^{(0)} \in \mathbb{R}^n$ be an initial point.
   \Repeat {($k\ge 0$)}
   \State Step 1: Compute $x^{(k+1)}$:
$$
x^{(k+1)} \in \mathrm{argmin}\left\{\mathcal{F}_\epsilon(x, |x^{(k)}|) + \iota_\mathcal{C}(Bx): x \in \mathbb{R}^n \right\}.
$$
   \Until{a given stopping criteria is met}
 \end{algorithmic}
\end{algorithm}

The terms $F_\epsilon (|x^{(k)}|)$ and $\langle \nabla F_\epsilon (|x^{(k)}|), |x^{(k)}| \rangle$ appeared in the optimization problem in Algorithm~\ref{alg:matrix} can be ignored because they are irrelevant to the optimization problem. Hence the expression for $x^{(k+1)}$ in Algorithm~\ref{alg:matrix} can be simplified as
\begin{equation}\label{model:subproblem}
x^{(k+1)} \in \mathrm{argmin}\left\{ \langle \nabla F_\epsilon (|x^{(k)}|), |x| \rangle + \iota_\mathcal{C}(Bx): x \in \mathbb{R}^n \right\}.
\end{equation}
Since $f_\epsilon$ is strictly concave and increasing on $\mathbb{R}_+$, $f'_\epsilon$ is positive on  $\mathbb{R}_+$. Hence,
$\langle \nabla F_\epsilon (|x^{(k)}|), |x| \rangle = \sum_{i=1}^n f'_\epsilon(|x^{(k)}_i|) |x_i|$   can be viewed as the weighted $\ell_1$-norm of $x$ having $f'_\epsilon(|x^{(k)}_i|)$ as its $i$th weight. Thus, the objective function of the above optimization problem is convex. Details for finding a solution to the problem will be presented in the next section.

In the rest of this section,  we list several possible choices of the functions in \eqref{def:F} including but not limited to the Mangasarian function  in \cite{Mangasarain:Optimization:99} and the Log-Det function in \cite{Fazel-Hindi-Boyd:03}.

The Mangasarian function  is given as follows:
\begin{equation}\label{eq:Mangasarain}
F_\epsilon (x)  = \sum_{i=1}^n \left(1- e^{-|x_i|/\epsilon}\right),
\end{equation}
where $x \in \mathbb{R}^n$. This function is used to approximate  the $\ell_0$-norm to obtain minimum-support solutions (that is, solutions with as many components equal to zero as possible).  The usefulness of the Mangasarian function was demonstrated in finding sparse solutions of underdetermined linear systems (see \cite{Jokar-Pfetsch:SISC:08}).


The Log-Det function is defined as
\begin{equation}\label{eq:L-bak}
F_\epsilon (x)  = \sum_{i=1}^n \frac{\log(|x_i|/\epsilon+1)}{\log(1/\epsilon)},
\end{equation}
where $x \in \mathbb{R}^n$. Notice that  $\|x\|_0$ is equal to the rank of the diagonal matrix $\mathrm{diag}(x)$. The function $F_\epsilon (x)$ is equal to $(\log(1/\epsilon))^{-1}\log(\mathrm{det}(\mathrm{diag}(x)+\epsilon I)) + n$, the logarithm of the determinant of the matrix $\mathrm{diag}(x)+\epsilon I$. Hence, it was named as the $\mbox{Log-Det}$ heuristic and  used for minimizing the rank of a positive semidefinite matrix over a convex set in \cite{Fazel-Hindi-Boyd:03}. Constant terms can be ignored since they will not affect the solution of the optimization problem~\eqref{model:subproblem}. Hence the Log-Det function in \eqref{eq:L-bak} can be replaced by
\begin{equation}\label{eq:L}
F_\epsilon (x)  = \sum_{i=1}^n \log(|x_i|+\epsilon).
\end{equation}

The function $F_\epsilon$ for the above three choices are plotted in Figure~\ref{fig:plot-functions} for $n=1$ and $\epsilon$ being $\frac{1}{4}$, $\frac{1}{8}$, $\frac{1}{16}$, and $\frac{1}{32}$. We can see that for a fixed $\epsilon \in (0, 1)$ the Mangasarian function is the one which is the most closest to the $\ell_0$-norm.

We point it out that the Mangasarian function is bounded by $1$, therefore, is non-coercive while the Log-Det function is coercive. This makes difference in convergence analysis of the associated Algorithm~\ref{alg:matrix} that will be presented in the next section. In what follows, the function $F_\epsilon$ is the Mangasarian function or the Log-Det function. We specify it only when it is noted.

\begin{figure}[h]
\centering
\begin{tabular}{cc}
\scalebox{0.37}{\includegraphics{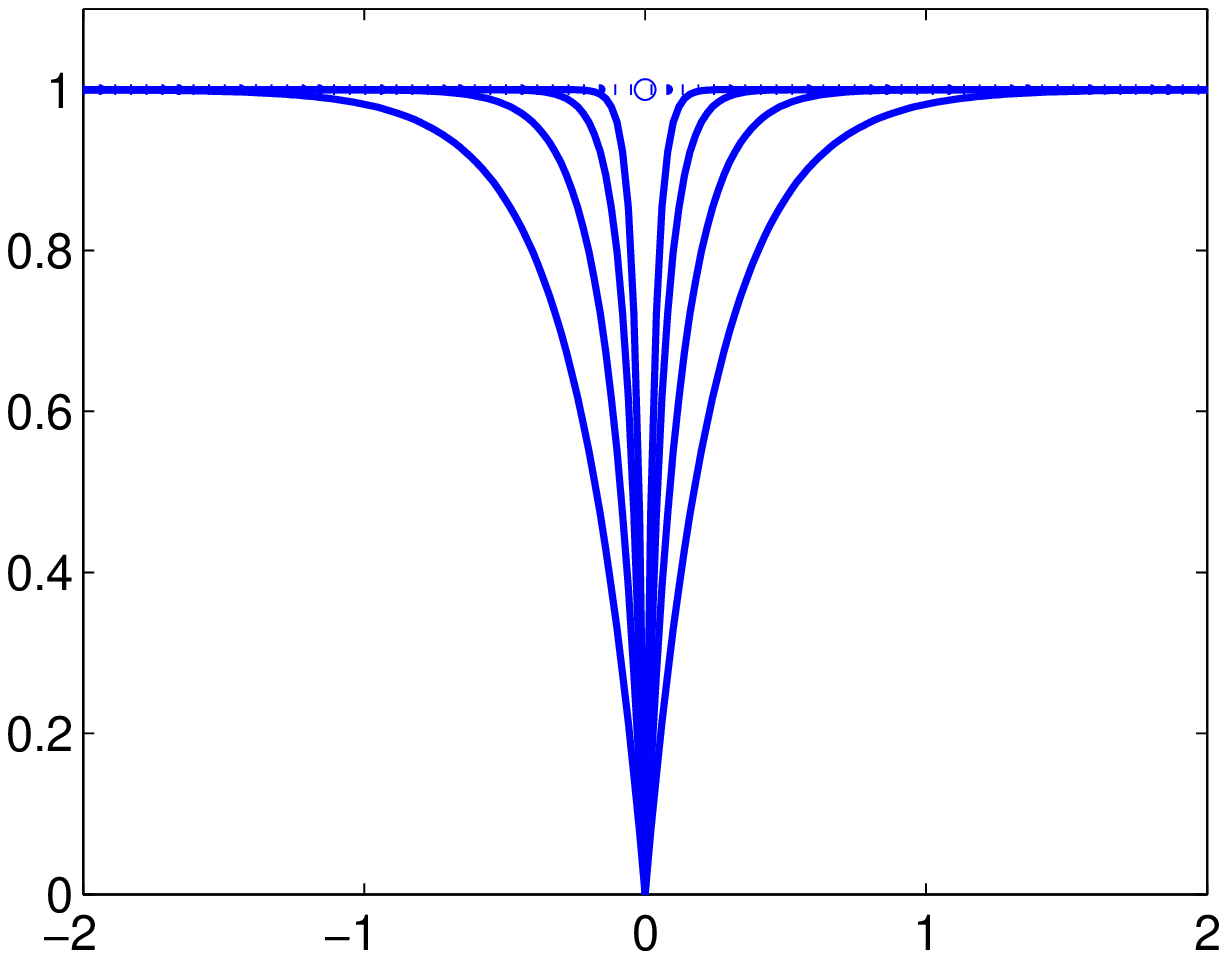}}&
\scalebox{0.37}{\includegraphics{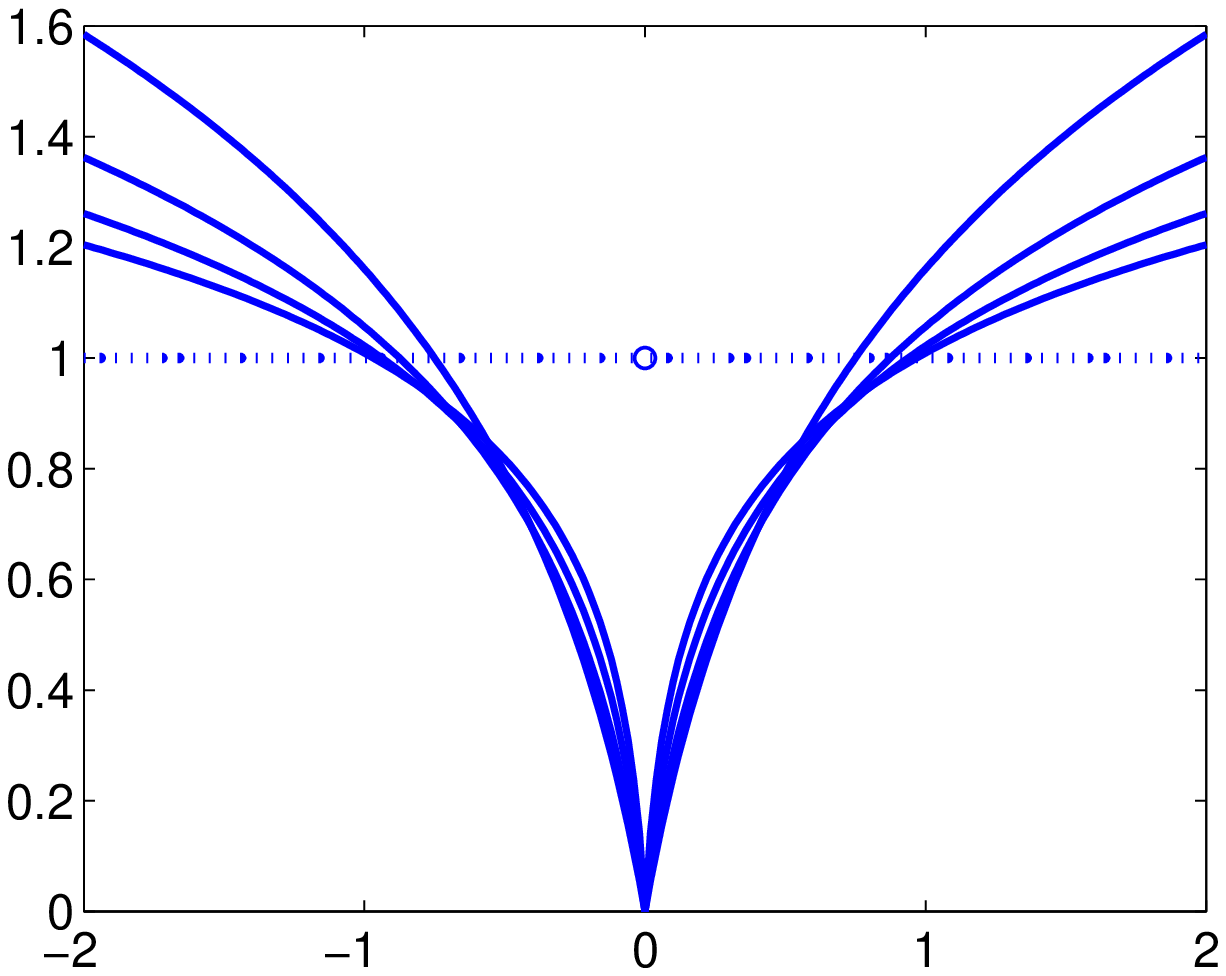}}\\
(a) Mangasarian &(b) Log-Det

\end{tabular}
\caption{Plots of $F_{\frac{1}{4}}$, $F_{\frac{1}{8}}$, $F_{\frac{1}{16}}$, $F_{\frac{1}{32}}$ with $n=1$ for (a) the Mangasarian function; (b) the Log-Det function.}
\label{fig:plot-functions}
\end{figure}
\section{Convergence Analysis}\label{sec:convergence}

In this section, we shall give convergence analysis for Algorithm~\ref{alg:matrix}. We begin with presenting the following result.

\begin{theorem}\label{thm:property}
Given $\epsilon\in(0, 1)$, $x^{(0)}\in\mathbb{R}^n$, and the set $\mathcal{C}$ defined by \eqref{eq:C}, let the sequence $\{x^{(k)}: k\in\mathbb{N}\}$ be generated by Algorithm~\ref{alg:matrix}, where $\mathbb{N}$ is the set of all natural numbers. Then the following three statements hold:
\begin{itemize}
\item [(i)] The sequence $\{F_\epsilon(x^{(k)}): k\in\mathbb{N}\}$ converges when  $F_\epsilon$ is corresponding to  the Mangasarian function \eqref{eq:Mangasarain} or the Log-Det function \eqref{eq:L};

\item [(ii)] The sequence $\{x^{(k)}:  k\in\mathbb{N}\}$ is bounded when  $F_\epsilon$ is the Log-Det function;
\item [(iii)] $\sum_{k=1}^{+\infty}\left\||x^{(k+1)}|-|x^{(k)}|\right\|_2^2$ is convergent when  the sequence $\{x^{(k)}:  k\in\mathbb{N}\}$ is bounded.
\end{itemize}
\end{theorem}
\begin{proof}\ \
We first prove Item (i). The key step for proving it is to show that the sequence $\{F_\epsilon(x^{(k)}): k\in\mathbb{N}\}$ is decreasing and bounded below. The boundedness of the sequence is due to the fact that $F_\epsilon(0) \le F_\epsilon(x^{(k)})$. From Step 1 of Algorithm~\ref{alg:matrix} or equation~\eqref{model:subproblem}, one can immediately have that
\begin{equation*}\label{eq:temp1-thm}
\iota_\mathcal{C}(Bx^{(k+1)})=0
\end{equation*}
and
\begin{equation}\label{eq:temp1-thm}
\langle \nabla F_\epsilon (|x^{(k)}|), |x^{(k+1)}| \rangle \le \langle \nabla F_\epsilon (|x^{(k)}|), |x^{(k)}| \rangle.
\end{equation}
By identifying $x^{(k)}$ and $x^{(k+1)}$, respectively, as $v$ and $x$ in \eqref{eq:L-expansion} and using the inequality in \eqref{eq:temp1-thm}, we get $F_\epsilon(x^{(k+1)}) \le F_\epsilon(x^{(k)})$. Hence, the sequence $\{F_\epsilon(x^{(k)}): k\in\mathbb{N}\}$ is decreasing and bounded below. Item (i) follows immediately.

When $F_\epsilon$ is chosen as the Log-Det function, the coerciveness of $F_\epsilon$ together with Item (i) implies that  the sequence $\{x^{(k)}:  k\in\mathbb{N}\}$ must be bounded, that is, Item (ii) holds.

Finally, we prove Item (iii).  Denote $ w^{(k)}:=|x^{(k+1)}|-|x^{(k)}|$. From the second-order Taylor expansion of the function $F_\epsilon$ at $x^{(k)}$ we have that
\begin{equation}\label{eq:temp2-thm}
F_\epsilon(x^{(k+1)})=\mathcal{F}_\epsilon(x^{(k+1)}, x^{(k)})+\frac{1}{2} (w^{(k)})^\top \nabla^2  F_\epsilon(v) w^{(k)},
\end{equation}
where $v$ is some point in the line segment linking the points $|x^{(k+1)}|$ and $|x^{(k)}|$ and $\nabla^2  F_\epsilon(v)$ is the Hessian matrix of $F_\epsilon$ at the point $v$.

By \eqref{eq:temp1-thm}, the first term on the right-hand of equation \eqref{eq:temp2-thm} is less than $F_\epsilon(x^{(k)})$. By equation \eqref{eq:L},  $\nabla^2  F_\epsilon(v)$ for $v$ lying in the first octant of $\mathbb{R}^n$ is a diagonal matrix and is equal to
$-\frac{1}{\epsilon^2}\mathrm{diag}(e^{-\frac{v_1}{\epsilon}}, e^{-\frac{v_2}{\epsilon}}, \ldots, e^{-\frac{v_n}{\epsilon}})$ or $-\mathrm{diag}((v_1+\epsilon)^{-2}, (v_2+\epsilon)^{-2}, \ldots (v_n+\epsilon)^{-2})$ which corresponds to $F_\epsilon$ being the Mangasarian or the Log-Det function. Hence, the matrix $\nabla^2  F_\epsilon(v)$ is negative definite. Since the sequence $\{x^{(k)}:  k\in\mathbb{N}\}$ is bounded, there exists a constant $\rho>0$ such that
$$
(w^{(k)})^\top \nabla^2  F_\epsilon(v) w^{(k)} \le -\rho \|w^{(k)}\|_2^2.
$$
Putting all above results together into \eqref{eq:temp2-thm}, we have that
$$
F_\epsilon(x^{(k+1)})\le F_\epsilon(x^{(k)}) -\frac{\rho}{2} \left\||x^{(k+1)}|-|x^{(k)}|\right\|_2^2.
$$
Summing the above inequality from $k=1$ to $+\infty$ and using Item (i) we get the proof of Item (iii).
\end{proof}

From Item (iii) of Theorem~\ref{thm:property}, we have $\left\||x^{(k+1)}|-|x^{(k)}|\right\|_2 \rightarrow 0$ as $k \rightarrow \infty$.

To further study properties of the sequence $\{x^{(k)}: k\in\mathbb{N}\}$ generated by Algorithm~\ref{alg:matrix}, the matrix $B^\top$ is required to have the range space property (RSP) which is originally introduced in \cite{Zhao-Li:SIOP:12}. With this property and motivated by the work in \cite{Zhao-Li:SIOP:12} we prove that Algorithm~\ref{alg:matrix} can yield a sparse solution for model~\eqref{model:matrix}.

Prior to presenting the definition of the RSP, we introduce the notation to be used throughout the rest of this paper. Given a set $S \subset \{1, 2, \ldots, n\}$, the symbol $|S|$ denotes the cardinality of $S$, and $S^c:=\{1, 2, \ldots, n\}\setminus S$ is the complement of $S$. Recall that for a vector $u$, by abuse of notation, we also use $|u|$ to denote the vector whose elements are the absolute values of the corresponding elements of $u$. For a given matrix $A$ having $n$ columns, a vector $u$ in $\mathbb{R}^n$, and a set $S \subset \{1, 2, \ldots, n\}$, we use the notation $A_S$ to denote the submatrix extracted from $A$ with column indices in $S$, and $u_S$ the subvector extracted from $u$ with component indices in $S$.

\begin{definition}[Range Space Property (RSP)]\ \
Let $A$ be an $m\times n$ matrix. Its transpose $A^\top$ is said to satisfy the {\emph {range space property (RSP)}} of order $K$ with a constant $\rho>0$ if for all sets $S\subseteq\{1,\dots, n\}$ with $|S|\geq K$ and for all $\xi$ in the range space of $A^{\top}$ the following inequality holds
$$
\|\xi_{S^c}\|_1\le \rho\|\xi_S\|_1.
$$
\end{definition}

We remark that if the transpose of an $m\times n$ matrix $B$ has the RSP of order $K$ with a constant $\rho>0$, then for every non-empty set $S\subseteq\{1,\dots, n\}$, the transpose of the matrix $B_S$, denoted by $B_S^\top$,  has the RSP of order $K$ with constant $\rho$ as well.

The next result shows that if the transpose of the matrix $B$ in Algorithm~\ref{alg:matrix} possesses the RSP, then  Algorithm~\ref{alg:matrix} can lead to a sparse solution for model~\eqref{model:matrix}. To this end, we define a mapping $\sigma: \mathbb{R}^d \rightarrow \mathbb{R}^d$ such that the $i$th component of the vector $\sigma(u)$ is the $i$th largest component of $|u|$.

\begin{prop}\label{lema:smallestToZero}
Let $B$ be the $(m+1)\times n$ matrix be defined by \eqref{eq:B} and let $\{x^{(k)}: k \in \mathbb{N}\}$ be the sequence generated by Algorithm~\ref{alg:matrix}. Assume that the matrix $B^\top$  has the RSP of order $K$ with $\rho>0$ satisfying $(1+\rho)K<n$. Suppose that the sequence $\{x^{(k)}: k \in \mathbb{N}\}$ is bounded. Then $(\sigma(x^{(k)}))_n$ the $n$th largest component of $x^{(k)}$ converges to $0$.
\end{prop}
\begin{proof}\ \ Suppose this proposition is \emph{false}. Then there exist a constant $\gamma>0$ and a subsequence $\{x^{(k_j)}:j\in\mathbb{N}\}$ such that $(\sigma(x^{(k_j)}))_n\geq 2\gamma>0$ for all $j \in \mathbb{N}$. From Item (iii) of Theorem \ref{thm:property} we have that
\begin{equation}\label{eq:sigma>gamma}
(\sigma(x^{(k_j+1)}))_n\geq \gamma
\end{equation}
for all sufficient large $j$. For simplicity, we set $y^{(k_j)}:=\nabla F_{\epsilon}(|x^{(k_j)}|)$. Hence, by inequality \eqref{eq:sigma>gamma} and $F_\epsilon$, we know that
\begin{equation} \label{eq:three-ineqs}
|x^{(k_j)}| >0 \quad |x^{(k_j+1)}|>0, \quad \mbox{and} \quad y^{(k_j)}>0
\end{equation}
for all sufficient large $j$. In what follows, we assume that the integer $j$ is large enough such that the above inequalities in \eqref{eq:three-ineqs} hold.

Since the vector $x^{(k_j+1)}$ is obtained through Step 1 of Algorithm~\ref{alg:matrix}, i.e., equation \eqref{model:subproblem}, then by Fermat's rule and the chain rule of subdifferential we have that
$$
0=\mathrm{diag}(y^{(k_j)})\partial {\|\cdot\|_1}(\mathrm{diag}(y^{(k_j)})x^{(k_j+1)})+B^\top b^{(k_j+1)},
$$
where $b^{(k_j+1)}\in\partial \iota_C(B x^{(k_j+1)})$. By \eqref{eq:three-ineqs}, we get
$$
\partial {\|\cdot\|_1}(\mathrm{diag}(y^{(k_j)})x^{(k_j+1)})=\{ \mathrm{sgn}(x^{(k_j+1)})\},
$$
where $\mathrm{sgn}(\cdot)$ denotes the sign of the variable element-wise. Thus
$$
y^{(k_j)}=|\xi^{(k_j+1)}|,
$$
where $\xi^{(k_j+1)}=B^\top b^{(k_j+1)}$ is in the range of $B^\top$.

Let $S$ be the set of indices corresponding to the $K$ smallest components of $|\xi^{(k_j+1)}|$. Hence,
$$
\sum_{i=1}^{n-K}(\sigma(y^{(k_j)}))_i = \|\xi^{(k_j+1)}_{S^c}\|_1
$$
and
$$
\sum_{i=n-K+1}^{n}(\sigma(y^{(k_j)}))_i = \|\xi^{(k_j+1)}_{S}\|_1.
$$
Since $B^\top$ has the RSP of order $K$ with the constant $\rho$, we have that $\|\xi^{(k_j+1)}_{S^c}\|_1\le \rho\|\xi^{(k_j+1)}_S\|_1$. Therefore,
\begin{equation}\label{eq:rsp1}
\sum_{i=1}^{n-K}(\sigma(y^{(k_j)}))_i \le \rho \sum_{i=n-K+1}^{n}(\sigma(y^{(k_j)}))_i.
\end{equation}
However, by the definition of $\sigma$, we have that
$$
\sum_{i=1}^{n-K}(\sigma(y^{(k_j)}))_i \ge (n-K)(\sigma(y^{(k_j)}))_{n-K+1}
$$
and
$$
\sum_{i=n-K+1}^{n}(\sigma(y^{(k_j)}))_i \le K (\sigma(y^{(k_j)}))_{n-K+1}.
$$
These inequalities together with the condition $(1+\rho)K<n$ lead to
$$
\sum_{i=1}^{n-K}(\sigma(y^{(k_j)}))_i > \rho \sum_{i=n-K+1}^{n}(\sigma(y^{(k_j)}))_i,
$$
which contradicts to \eqref{eq:rsp1}. This completes the proof of the proposition.
\end{proof}

From Proposition~\ref{lema:smallestToZero}, we conclude that a sparse solution is guaranteed via Algorithm~\ref{alg:matrix} if the transpose of $B$ satisfies the RSP. Next, we answer how sparse this solution will be. To this end, we introduce some notation and develop a technical lemma. For a vector $x \in \mathbb{R}^d$, we denote by $\tau(x)$ the set of the indices of non-zero elements of $x$, i.e., $\tau(x):=\{i: x_i \neq 0\}$. For a sequence $\{x^{(k)}: k \in \mathbb{N}\}$, a positive number $\mu$, and an integer $k$, we define $I_\mu(x^{(k)}):=\{i: |x_i^{(k)}|\ge \mu\}$.

\begin{lemma}\label{lema:<n}
Let $B$ be the $(m+1)\times n$ matrix defined by \eqref{eq:B}, let $F_\epsilon$ be the Log-Det function defined by \eqref{eq:L}, and let $\{x^{(k)}: k \in \mathbb{N}\}$ be the sequence generated by Algorithm~\ref{alg:matrix}. Assume that the matrix $B^\top$  has the RSP of order $K$ with $\rho>0$ satisfying $(1+\rho)K<n$. If there exist $\mu > \rho \epsilon n$   such that $|I_\mu(x^{(k)})|\geq K$ for all sufficient large $k$,  then there exists a $ k'' \in\mathbb{N}$ such that $\|x^{(k)}\|_0<n$ and $\tau (x^{(k+1)})\subseteq\tau(x^{(k'')})$ for all $k> k''$.
\end{lemma}
\begin{proof}\ \
Set $y^{(k)}:=\nabla F_{\epsilon}(|x^{(k)}|)$. Since $x^{(k+1)}$ is a solution to the optimization problem~\eqref{model:subproblem}, then by Fermat's rule and the chain rule of subdifferential we have that
$$
0\in\mathrm{diag}(y^{(k)})\partial {\|\cdot\|_1}(\mathrm{diag}(y^{(k)})x^{(k+1)})+B^\top b^{(k+1)},
$$
where $b^{(k+1)}\in\partial \iota_C(B x^{(k+1)})$. Hence, if $x_i^{(k+1)}\neq 0$, we have that  $y^{(k)}_i = |(B^\top b^{(k+1)})_i|$.

For $i\in I_\mu(x^{(k)})$, we have that $|x_i^{(k)}| \ge \mu$ and $y^{(k)}_i=f'_{\epsilon}(|x^{(k)}_i|) \le f'_{\epsilon}(\mu)$  for all $k\in\mathbb{N}$, where $f_{\epsilon} = \log (\cdot +\epsilon)$.  Furthermore, there exist a $k'$ such that $|x_i^{k+1}|>0$ for $i\in I_\mu(x^{(k)})$ and $k\ge k'$ due to Item (iii) in Theorem \ref{thm:property}.  Thus, we have for all $k\ge k'$

\begin{eqnarray*}
\sum_{i\in I_\mu(x^{(k)})}{|(B^\top b^{(k+1)})_i|} &=& \sum_{i\in I_\mu(x^{(k)})} y^{(k)}_i \\
&\le& \sum_{i\in I_\mu(x^{(k)})} f'_{\epsilon}(\mu)\le W^*,
\end{eqnarray*}
where $W^* = n \lim_{\epsilon \rightarrow 0+} f_\epsilon'(\mu)=\frac{n}{\mu}$ is a positive number dependent on $\mu$.

Now, we are ready to prove $\|x^{(k)}\|_0<n$ for all $k> k''$. By Proposition~\ref{lema:smallestToZero}, we have that $(\sigma(x^{(k)}))_n \rightarrow 0$ when $k\rightarrow +\infty$. Therefore, there exists an integer $k''>k'$ such that $|I_\mu(x^{(k)})|\ge K$ and  $0\le \sigma(x^{(k)}))_n<\min\{\frac{\mu}{\rho n}-\epsilon, \mu\}$ for all $k\ge k''$.  Let $i_0$ be the index such that $|x_{i_0}^{(k'')}|=(\sigma(x^{(k'')}))_n$.  We will show that $x_{i_0}^{(k''+1)}=0$. If this statement is not true, that is, $x_{i_0}^{(k''+1)}$ is not zero, then
\begin{equation}\label{eq:contr1}
|(B^\top b^{(k''+1)})_{i_0}|=f'_{\epsilon}(|x^{(k'')}_{i_0}|)=\frac{1}{|x^{(k'')}_{i_0}|+\epsilon}>\rho W^*.
\end{equation}
However, since $i_0$ is not in the set $I_\mu(x^{(k'')})$ and $B^\top$ satisfies the RSP, we have that
\begin{eqnarray*}
|(B^\top b^{(k''+1)})_{i_0}|&\le& \sum_{i\notin I_\mu(x^{(k'')})}|(B^\top b^{(k''+1)})_{i}| \\
&\leq& \rho \sum_{i\in I_\mu(x^{(k'')})}|(B^\top b^{(k''+1)})_{i}| \le \rho W^*,
\end{eqnarray*}
which contradicts to \eqref{eq:contr1}. Hence, we have that $x_{i_0}^{(k''+1)}=0$ and $|\tau(x^{(k''+1)})|<n$. By replacing $k''$ by $k''+1$ and repeating this process we can obtain $x_{i_0}^{(k''+\ell)}=0$ for all $\ell\in\mathbb{N}$. Therefore, $\|x\|_0<n$ for all $k>k''$. This process can be also applied to other components satisfying $x_i^{(k''+1)}=0$. Thus there exists a $k'' \in\mathbb{N}$ such that $\tau(x^{(k)})\subseteq\tau(x^{(k'')})$ for all $k\ge k''$.
\end{proof}

With Lemma~\ref{lema:<n}, the next result shows that when the transpose of $B$ satisfies the RSP there exists a cluster point of the sequence generated by Algorithm~\ref{alg:matrix} that is sparse and satisfies the consistency condition.

\begin{theorem}
Let $B$ be the $(m+1)\times n$ matrix defined by \eqref{eq:B}, let $F_\epsilon$ be the Log-Det function defined by \eqref{eq:L}, and let $\{x^{(k)}: k \in \mathbb{N}\}$ be the sequence generated by Algorithm~\ref{alg:matrix}. Assume that the matrix $B^\top$  has the RSP of order $K$ with $\rho>0$ satisfying $(1+\rho)K<n$. Then there is a subsequence $\{x^{(k_j)}: j\in\mathbb{N}\}$ that converges to a $\lfloor (1+\rho)K\rfloor$-sparse solution, that is $(\sigma(x^{(k_j)}))_{\lfloor (1+\rho)K+1\rfloor}\rightarrow 0$ as $j\rightarrow +\infty$ and $\epsilon\rightarrow 0$.
\end{theorem}
\begin{proof}\ \
Suppose the theorem is \emph{false}. Then there exist $\mu^*$, for any  $0<\epsilon^*<\frac{\mu^*}{\rho n}$, there exist a $\epsilon \in (0, \epsilon^*)$ and $k'$ such that
$(\sigma(x^{(k)}))_{\lfloor (1+\rho)K+1\rfloor}\ge \mu^*$ for all $k\ge k'$. It implies that for all $k\ge k'$
\begin{equation}\label{eq:rho_k}
|I_{\mu^*}(x^{(k)})|\ge \lfloor (1+\rho)K+1\rfloor>(1+\rho)K> K.
\end{equation}
By Lemma \ref{lema:<n}, there exist a $k''\ge k'$ such that $\|x^{(k)}\|_0<n$ and $\tau(x^{(k+1)})\subseteq\tau (x^{(k'')})$ for all $k\ge k''$. Let $S=\tau(x^{(k'')})$. Thus $x^{(k)}_{S^c}=0$ for all $k\ge k''$. Therefore, the optimization problem \eqref{model:subproblem} for updating  $x^{(k+1)}$ can be reduced to the following one
\begin{equation}\label{eq:subp}
x_S^{k+1} \in \mathrm{arg}\min\{\langle (\nabla F_{\epsilon}(|x^{(k)}|))_S, u \rangle +\iota((B_S) u): u\in\mathbb{R}^{|S|}\}.
\end{equation}

If $|\tau(x^{(k'')})|>|I_{\mu^*}(x^{(k'')})|$, from \eqref{eq:rho_k} we have $(1+\rho)K<|S|$. Thus from Lemma \ref{lema:<n} and $B^\top_S$ having RSP with the same parameters, there exist a $k'''>k''$ such that  $\tau (x^{(k)})<\tau (x^{(k'')})$ for all $k\ge k'''$. Therefore, by induction, there must exist a $\tilde k$ such that for all $k\ge \tilde k$
 $$
 \tau (x^{(k)})=I_{\mu^*}(x^{(k)}), \; \tau (x^k)\subseteq \tau (x^{(\tilde k)}).
 $$
It means that for all $k\ge \tilde k$ all the nonzero components of $x^{(k)}$ are bounded below by $\mu^*$.
  Therefore, for any $k\ge \tilde k$, the updating equation \eqref{model:subproblem} is reduced by \eqref{eq:subp} with $S=I_{\mu^*}(x^{(k)})$. From Lemma \ref{lema:smallestToZero} we get $[\sigma(x^{(k)})]_{|S|}\rightarrow 0$ which contradicts with $|x_{|S|}^k|\ge \mu^*$. Therefore, we get this theorem.
\end{proof}


\section{An Implementation of Algorithm~\ref{alg:matrix}}\label{sec:complete-algorithm}
In this section, we describe in detail an implementation of Algorithm~\ref{alg:matrix} and show how to select the parameters of the associated  algorithm.

Solving problem~\eqref{model:subproblem} is the main issue for Algorithm~\ref{alg:matrix}. A general model related to \eqref{model:subproblem} is
\begin{equation}\label{model:general}
\min\{\|\Gamma x\|_1 + \varphi (Bx): x \in \mathbb{R}^n\},
\end{equation}
where $\Gamma$ is a diagonal matrix with positive diagonal elements and $\varphi$ is in $\Gamma_0(\mathbb{R}^{m+1})$. In particular, if we choose $\Gamma = \nabla F_\epsilon (|x^{(k)}|)$ and $\varphi = \iota_\mathcal{C}$, where $x^{(k)}$ is a vector in $\mathbb{R}^n$, $\epsilon$ is a positive number, $\mathcal{C}$ is given by \eqref{eq:C}, and $F_\epsilon$ is a function given by \eqref{def:F}, then model~\eqref{model:general} reduces to the optimization problem in Algorithm~\ref{alg:matrix}.

We solve model~\eqref{model:general} by using recently developed first-order primal-dual algorithm (see, e.g., \cite{Chambolle-Pock:JMIV11,Li-Micchelli-Shen-Xu:IP-12,Zhang-Burger-Osher:JSC:2011}). To present this algorithm, we need two concepts  in convex analysis,  namely,  the proximity operator and conjugate function. The proximity operator was introduced in \cite{moreau:RASPS:62}. For a function $f \in \Gamma_0(\mathbb{R}^d)$, the proximity operator of $f$ with parameter $\lambda$, denoted by $\mathrm{prox}_{\lambda f}$, is a mapping from $\mathbb{R}^d$ to itself, defined for a given point $x \in \mathbb{R}^d$ by
\begin{equation*}\label{eq:prox}
\mathrm{prox}_{\lambda f} (x):=\mathop{\mathrm{argmin}} \left\{\frac{1}{2\lambda} \|u-x\|^2_2 + f(u): u \in \mathbb{R}^d \right\}.
\end{equation*}
The conjugate of $f\in \Gamma_0(\mathbb{R}^d)$ is the function $f^* \in
\Gamma_0(\mathbb{R}^d)$ defined at $z \in \mathbb{R}^d$ by
$$
f^*(z):= \sup\{\langle x, z \rangle-f(x): x\in \mathbb{R}^d\}.
$$
With these notation, the first-order primal-dual (PD) method for solving \eqref{model:general} is summarized in Algorithm~\ref{Alg:PD} (referred to as PD-subroutine).

\begin{algorithm}\caption{ PD-subroutine (The first-order primal-dual algorithm for solving \eqref{model:general})}\label{Alg:PD}
 \begin{algorithmic}[htb]
   \State \textbf{Input}: the $(m+1) \times n$ matrix $B$ defined by \eqref{eq:B}; two positive numbers $\alpha$ and $\beta$ satisfying the relation $\alpha \beta < \frac{1}{\|B\|^2}$; the $n\times n$ diagonal matrix $\Gamma$ with all diagonal elements positive; and the function $\varphi \in \Gamma_0(\mathbb{R}^n)$.

   \State \textbf{Initialization}: $i=0$ and an initial guess $(u^{-1}, u^{0}, x^{0}) \in \mathbb{R}^{m+1} \times \mathbb{R}^{m+1} \times \mathbb{R}^{n}$
   \Repeat {($i\ge 0$)}
   \State Step 1: Compute $x^{i+1}$:
$$
x^{i+1}=\mathrm{prox}_{\alpha\|\cdot\|_1 \circ \Gamma} \left(x^{i}-\alpha B^\top (2u^{i}-u^{i-1})\right)
$$
   \State Step 2: Compute $u^{i+1}$:
$$
u^{i+1}=\mathrm{prox}_{\beta \varphi^*} (u^{i}+\beta B x^{i+1})
$$
    \State Step 3: Set $i:=i+1$.

   \Until{a given stopping criteria is met and the corresponding vectors $u^i$, $u^{i+1}$, and $x^{i+1}$ are denoted by $u^{cur}$, $u^{new}$, and $x^{new}$, respectively.}
    \State \textbf{Output}: $(u^{cur}, u^{new}, x^{new})=\mathrm{PD}(\alpha,\beta, B, \Gamma, \varphi, u^{-1}, u^{0}, x^{0})$
 \end{algorithmic}
\end{algorithm}

\begin{theorem}\label{thm:convergence-L1-Iota}
Let $B$ be an $(m+1) \times n$ matrix defined by \eqref{eq:B}, let $\mathcal{C}$ be the set given by \eqref{eq:C}, let $\alpha$ and $\beta$ be two positive numbers, and let $L$ be a positive such that $L \ge \|B\|^2$, where $\|B\|$ is the largest singular value of $B$. If
\begin{equation*}\label{eq:alpha-beta}
\alpha \beta L < 1,
\end{equation*}
then for any arbitrary initial vector $(x^{-1}, x^{0}, u^{0}) \in \mathbb{R}^n \times \mathbb{R}^n \times \mathbb{R}^{m+1}$, the sequence $\{x^k: k \in \mathbb{N}\}$ generated by Algorithm \ref{Alg:PD} converges to a solution of model~\eqref{model:general}.
\end{theorem}

The proof of Theorem~\ref{thm:convergence-L1-Iota} follows immediately from Theorem 1 in \cite{Chambolle-Pock:JMIV11} or Theorem 3.5 in \cite{Li-Micchelli-Shen-Xu:IP-12}. We skip its proof here.

Both proximity operators $\mathrm{prox}_{\alpha\|\cdot\|_1 \circ \Gamma}$ and $\mathrm{prox}_{\beta \varphi^*}$ should be computed easily and efficiently in order to make the iterative scheme in Algorithm \ref{Alg:PD} numerically efficient.  Indeed, the proximity operator $\mathrm{prox}_{\alpha\|\cdot\|_1 \circ \Gamma}$ is given at $z \in \mathbb{R}^n$ as follows: for $j=1, 2, \ldots, n$
\begin{equation}\label{eq:prox-L1}
\left(\mathrm{prox}_{\alpha\|\cdot\|_1 \circ \Gamma}(z)\right)_j=\max\left\{|z_j|-\alpha \gamma_j,0\right\} \cdot \mathrm{sign}(z_j),
\end{equation}
where $\gamma_j$ is the $j$th diagonal element of $\Gamma$. Using the well-known Moreau decomposition (see, e.g. \cite{Bauschke-Combettes:11,moreau:RASPS:62})
\begin{equation}\label{identity-prox-dual}
\mathrm{prox}_{\beta \varphi^*} = I- \beta \; \mathrm{prox}_{\frac{1}{\beta} \varphi} \circ \left(\frac{1}{\beta} I\right),
\end{equation}
we can compute the proximity operator $\mathrm{prox}_{\beta \varphi^*}$ via $\mathrm{prox}_{\frac{1}{\beta} \varphi}$ which depends on a particular form of the function $\varphi$. As our purpose is to develop algorithms for the optimization problem in Algorithm~\ref{alg:matrix}, we need  to compute the proximity operator of $\iota^*_\mathcal{C}$ which is given in the following.
\begin{lemma}\label{lemma:iotaC}
If $\mathcal{C}$ is the set given by \eqref{eq:C} and $\beta$ is a positive number, then for $z \in \mathbb{R}^{m+1}$ we have that
\begin{equation}\label{eq:prox-dual-iotaC}
\mathrm{prox}_{\beta \iota^*_{\mathcal{C}}}(z)=(z_1-(z_1)_+, \ldots, z_m-(z_m)_+, z_{m+1}-\beta),
\end{equation}
where $(s)_+$ is $s$ if $s \ge 0$ and $0$ otherwise.
\end{lemma}
\begin{proof}\ \ We first give an explicit form for the proximity operator $\mathrm{prox}_{\frac{1}{\beta}\iota_{\mathcal{C}}}$. Note that  $\iota_{\mathcal{C}} = \frac{1}{\beta}\iota_{\mathcal{C}}$ for $\beta>0$ and $\iota_{\mathcal{C}}(z) = \iota_{\{1\}}(z_{m+1}) + \sum_{i=1}^m \iota_{[0,\infty)}(z_i)$, for $z \in \mathbb{R}^{m+1}$. Hence, we have that
\begin{equation}\label{eq:prox-iotaC}
\mathrm{prox}_{\frac{1}{\beta}\iota_{\mathcal{C}}}(z)=((z_1)_+, (z_2)_+, \ldots, (z_m)_+, 1),
\end{equation}
where $(s)_+$ is $s$ if $s \ge 0$ and $0$ otherwise. Here we use the facts that $\mathrm{prox}_{\iota_{[0,+\infty)}}(s)=(s)_+$  and $\mathrm{prox}_{\iota_{\{1\}}}(s)=1$ for any $s \in \mathbb{R}$.

By the Moreau decomposition \eqref{identity-prox-dual}, we have that $\mathrm{prox}_{\beta \iota^*_{\mathcal{C}}}(z)=z-\beta \mathrm{prox}_{\frac{1}{\beta}\iota_{\mathcal{C}}}(\frac{1}{\beta}z)$. This together with equation~\eqref{eq:prox-iotaC} yields \eqref{eq:prox-dual-iotaC}.
\end{proof}

Next, we comment on the diagonal matrix $\Gamma$ in model~\eqref{model:general}. When the function $\varphi$ in model~\eqref{model:general} is chosen to be $\iota_C$, then the relation $a \varphi = \varphi$ holds for any positive number $a$. Hence, by rescaling the diagonal matrix $\Gamma$ in model~\eqref{model:general} with any positive number, that does not alter the solutions of model~\eqref{model:general}. Therefore, we can assume that the largest diagonal entry of $\Gamma$ is always equal to one.

In applications of Theorem~\ref{thm:convergence-L1-Iota} as in Algorithm \ref{Alg:PD}, we should make the product of $\alpha$ and $\beta$ as close to ${1}/{\|B\|^2}$ as possible. In our numerical simulations, we always set
\begin{equation}\label{eq:alpha-beta}
\alpha = \frac{0.999}{\beta \|B\|^2}.
\end{equation}
In such the way, $\beta$ is essentially the only parameter that needs to be determined.

Prior to computing $\alpha$ for a given $\beta$ by equation~\eqref{eq:alpha-beta}, we need to know the norm of the matrix $B$. When $\min \{m,n\}$ is small, the norm of the matrix $B$ can be computed directly. When $\min \{m,n\}$ is large, an upper bound of the norm of the matrix $B$ is estimated in terms of the size of $B$ as follows.

\begin{prop}\label{prop:norm}
Let $\Phi$ be an $m \times n$ matrix with i.i.d. standard Gaussian entries and $y$ be an $m$-dimensional vector with its component being $+1$ or $-1$. We define an $(m+1)\times n$ matrix $B$ from $\Phi$ and $y$ via equation \eqref{eq:B}. Then
$$
\mathbb{E} \{\|B\|\} \le \sqrt{m+1}(\sqrt{n}+\sqrt{m}).
$$
Moreover,
$$
\|B\| \le \sqrt{m+1}(\sqrt{n}+\sqrt{m}+t)
$$
holds with probability at least $1-2 e^{-t^2/2}$ for all $t \ge 0$.
\end{prop}
\begin{proof}\ \ By the structure of the matrix $B$ in \eqref{eq:B}, we know that
$$
\|B\| \le \left\|\begin{bmatrix} \mathrm{diag}(y) \\ y^\top \end{bmatrix}\right\| \cdot \|\Phi\|.
$$
Therefore, we just need to compute the norms on the right-hand side of the above inequality. Denote by $I_m$ the $m \times m$ identity matrix and ${1}_m$ the vector with all its components being $1$. Then
$$
\begin{bmatrix} \mathrm{diag}(y) \\ y^\top \end{bmatrix} \begin{bmatrix} \mathrm{diag}(y) & y \end{bmatrix} =
\begin{bmatrix} I_m & {1}_m \\ {1}_m^\top & m \end{bmatrix},
$$
which is a special arrow-head matrix and has $m+1$ as its largest eigenvalue (see \cite{Shen-Suter:Eurasip:09}). Hence,
$$
\left\|\begin{bmatrix} \mathrm{diag}(y) \\ y^\top \end{bmatrix}\right\|=\sqrt{m+1}.
$$
Furthermore, by using random matrix theory for the matrix $\Phi$, we know that  $\mathbb{E} \{\|\Phi\|\} \le \sqrt{n}+\sqrt{m}$ and $\|\Phi\| \le \sqrt{n}+\sqrt{m}+t$ with probability at least $1-2 e^{-t^2/2}$ for all $t \ge 0$ (see, e.g., \cite{Davidson-Szarek:01}). This completes the proof of this proposition.
\end{proof}

Let us compute the norm of $B$ numerically for $100$ randomly generated matrices $\Phi$ and vectors $y$ for the pair $(m,n)$ with three different choices $(500,1000)$,  $(1000,1000)$, and $(1500,1000)$, respectively. Corresponding to these choices, the mean values of $\|B\|$ are about $815$, $1276$, and $1711$ while the upper bounds of the expected values of $\|B\|$ by Proposition~\ref{prop:norm} are about $1208$, $2001$, and $2726$, respectively. We can see that the norm of $B$ varies with its size and turns to be a big number when the value of $\min \{m,n\}$ is relatively large. As a consequence, the parameter $\alpha$ or $\beta$ must be very small relative to the other by equation~\eqref{eq:alpha-beta}. Therefore, in what follows, the used matrix $B$ in model~\eqref{model:general} is considered to have been rescaled in the following way:
\begin{equation}\label{eq:mormalize-B}
\frac{B}{\|B\|} \quad \mbox{or} \quad \frac{B}{\sqrt{m+1}(\sqrt{n}+\sqrt{m})}
\end{equation}
when the norm of $B$ can be computed easily or not.

%
%

The complete procedure for model~\eqref{model:matrix} and how the PD-subroutine is employed are summarized in Algorithm~\ref{alg:matrix-final}.
\begin{algorithm}\caption{(Iterative scheme for model~\eqref{model:matrix})}\label{alg:matrix-final}
 \begin{algorithmic}[htb]
   \State \textbf{Input}: the $(m+1) \times n$ matrix $B$ formed by an $m \times n$ matrix $\Phi$ and an $m$-dimensional vector $y$ via \eqref{eq:B}; the set $\mathcal{C}$ given by \eqref{eq:C}; $\epsilon \in (0,1)$, and $\tau>0$; $\alpha_{\max}$ and $\epsilon_{\min}$ be two real numbers; the maximum iteration number $k_{\max}$.

   \State \textbf{Initialization}: normalizing $B$ according to \eqref{eq:mormalize-B}; $\Gamma$ being the $n \times n$ identity matrix;  an initial guess $(u^{old_{0}}, u^{cur_0}, x^{(0)}) \in \mathbb{R}^{m+1} \times \mathbb{R}^{m+1} \times \mathbb{R}^{n}$; and initial parameters $\beta$ and $\alpha=0.999/\beta$.
   \While{$k < k_{\max}$}

    \State Step 1: Compute
    \begin{eqnarray*}
    &&(u^{old_{k+1}}, u^{cur_{k+1}}, x^{(k+1)})\\
    &=&\mathrm{PD}(\alpha,\beta, B, \Gamma, \iota_{\mathcal{C}}, u^{old_{k}}, u^{cur_k}, x^{(k)})
    \end{eqnarray*}

    \State Step 2: Update $\Gamma$ as the scaled matrix $\mathrm{diag}(\nabla F_\epsilon(x^{(k+1)}))$ such that the largest diagonal element of $\Gamma$ is one.

    \State Step 3: If $\alpha < \alpha_{\max}$, update $\alpha \leftarrow 2\alpha, \quad \beta \leftarrow \beta/2$;
    if $\epsilon > \epsilon_{\min}$, update $\epsilon \leftarrow \tau \epsilon$;

    \State Step 4: Update $k \leftarrow k+1$.
\EndWhile
    \State \textbf{Output}: $x^{(k_{\max})}$
 \end{algorithmic}
\end{algorithm}

\section{Numerical Simulations}\label{sec:experiments}
In this section, we demonstrate the performance of Algorithm~\ref{alg:matrix-final} for 1-bit compressive sampling reconstruction in terms of accuracy and consistency and compare it with the BIHT algorithm.

Through this section, all random $m \times n$ matrices $\Phi$ and length-$n$, $s$-sparse vectors $x$ are generated based on the following assumption: entries of $\Phi$ and $x$ on its support are i.i.d. Gaussian random variables with zero mean and unit variances. The locations of the nonzero entries (i.e., the support) of $x$ are randomly permuted. We then generate the 1-bit observation vector $y$ by equation \eqref{model:1bit}. We obtain reconstruction of $x^\star$ from $y$ by using the BIHT and Algorithm~\ref{alg:matrix-final}. The quality of the reconstructed $x^\star$ is measured in terms of the signal-to-noise ratio (SNR) in dB
$$
\mathrm{SNR}(x,x^\star) = 20 \log_{10} \left(\left\|\frac{x}{\|x\|}\right\|_2/\left\|\frac{x}{\|x\|}-\frac{x^\star}{\|x^\star\|}\right\|_2\right).
$$
The accuracy of the BIHT and Algorithm~\ref{alg:matrix-final} is measured by the average of SNR values over 100 trials unless otherwise noted. For all figures in this section, results by  the BIHT and Algorithm~\ref{alg:matrix-final} with the Mangasarian function \eqref{eq:Mangasarain} and the Log-Det function \eqref{eq:L} are marked by the symbols ``$\triangledown$'', ``$\circ$'',  and ``$\star$'', respectively.

\subsection{Effects of using inaccurate sparsity on the BIHT}

The BIHT requires the availability of the sparsity of the underlying signals. This requirement is, however, not known in practical applications. In this subsection, we demonstrate through numerical experiments that the mismatched sparsity for a signal will degenerate the performance of the BIHT.

To this end, we fix $n=1000$ and $s=10$ and consider two cases of $m$ being $500$ and $1000$. For each case, we vary the sparsity input for the BIHT from $8$ to $12$ in which $10$ is the only right choice. Therefore, there are total ten configurations. For each configuration, we record the SNR values and the numbers of sign constraints not being satisfied of the reconstructed signals by the BIHT.

Figure~\ref{fig:VaryingM-BIHT} depicts the SNR values of the experiments. The plots in the left column of Figure~\ref{fig:VaryingM-BIHT} are for the case $m=500$ while the plots in the right column are for the case $m=1000$. The marks in each plot represent the pairs of the SNR values with a mismatched sparsity input (i.e., $s=8$, $s=9$, $s=11$, or $s=12$ corresponding to the row 1, 2, 3, or 4) and with the correct sparsity input (i.e., $s=10$). A mark below the red line indicates that the BIHT with the correct sparsity input works better than the one with an incorrect sparsity input. A mark that is far away from the red line indicates the BIHT with the correct sparsity input works much better than the one with an incorrect sparsity input or vice versa.  Except the second plot in the left column, we can see that the BIHT with the correct sparsity input performs better than the one with an inaccurate sparsity input. In particular, when an underestimated sparsity input to the BIHT is used, the performance of the BIHT will be significantly reduced (see the plots in the first two columns of Figure~\ref{fig:VaryingM-BIHT}). When an overestimated sparsity input to the BIHT is used, majority marks are under the red lines and are relatively closer to the red lines than those from the BIHT with underestimated sparsity input. We further report that the average SNR values for the sparsity input $s=8$, $9$, $10$, $11$, and $12$ for $m=500$ are $21.89$dB, $24.18$dB, 	 $23.25$dB, $22.10$dB, and $21.00$dB, respectively. Similarly, for $m=1000$, the average SNR values for the sparsity input $s=8$, $9$, $10$, $11$, and $12$ are $19.77$dB, $26.37$dB, 	$34.74$dB,	$31.12$dB, and $29.46$dB, respectively.

\begin{figure}[h]
\centering
\begin{tabular}{cc}
\scalebox{0.36}{\includegraphics{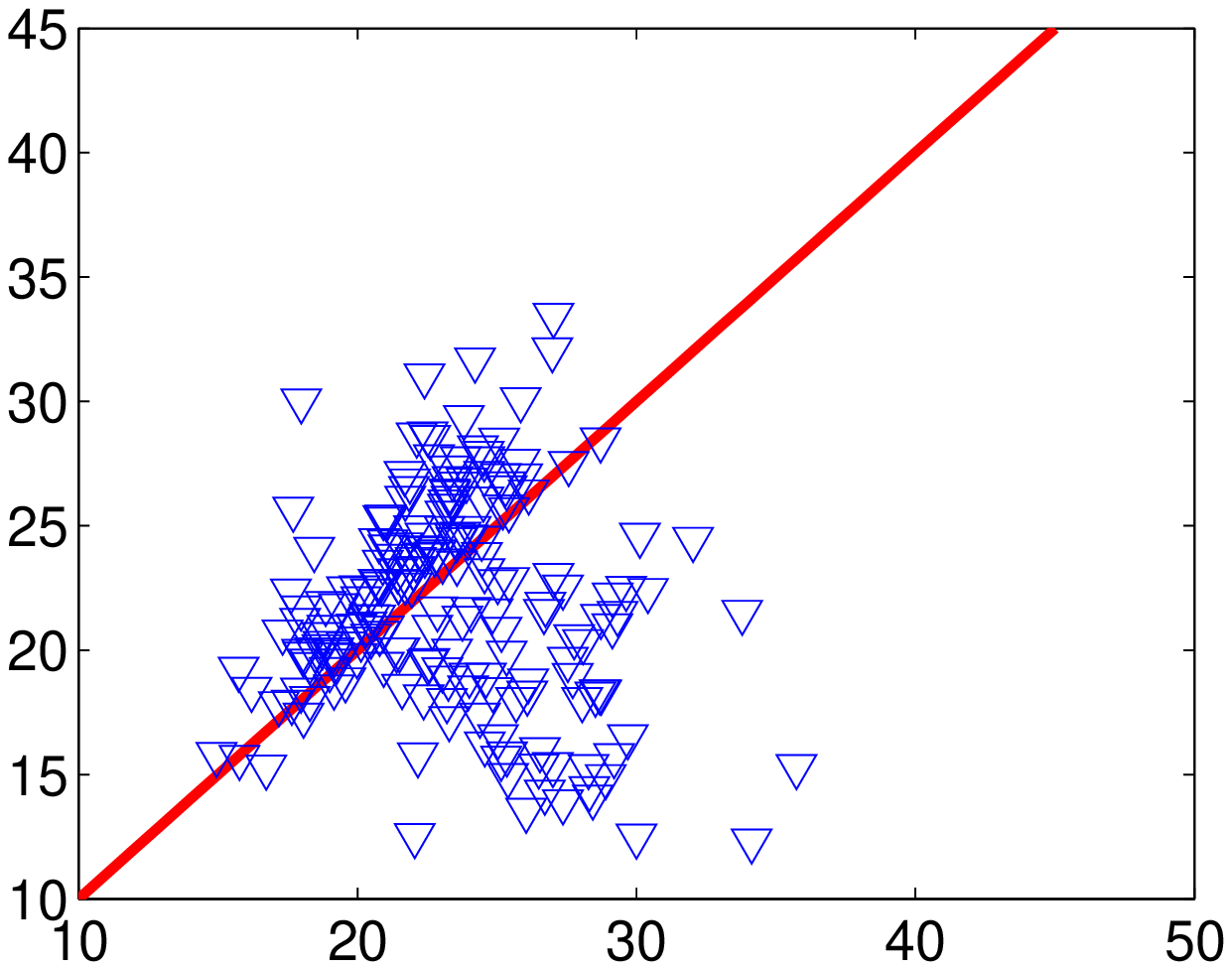}}&\scalebox{0.36}{\includegraphics{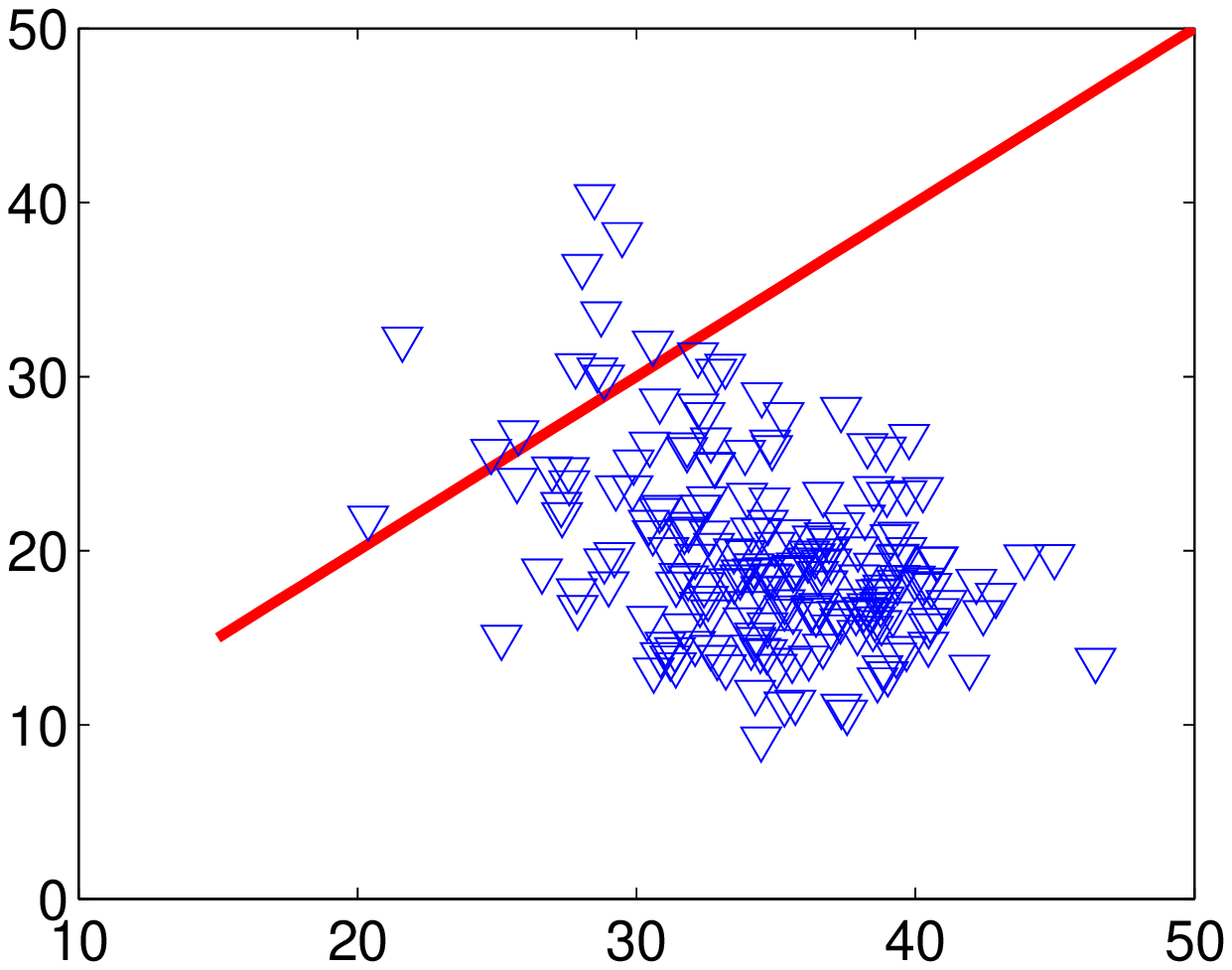}}\\
\scalebox{0.36}{\includegraphics{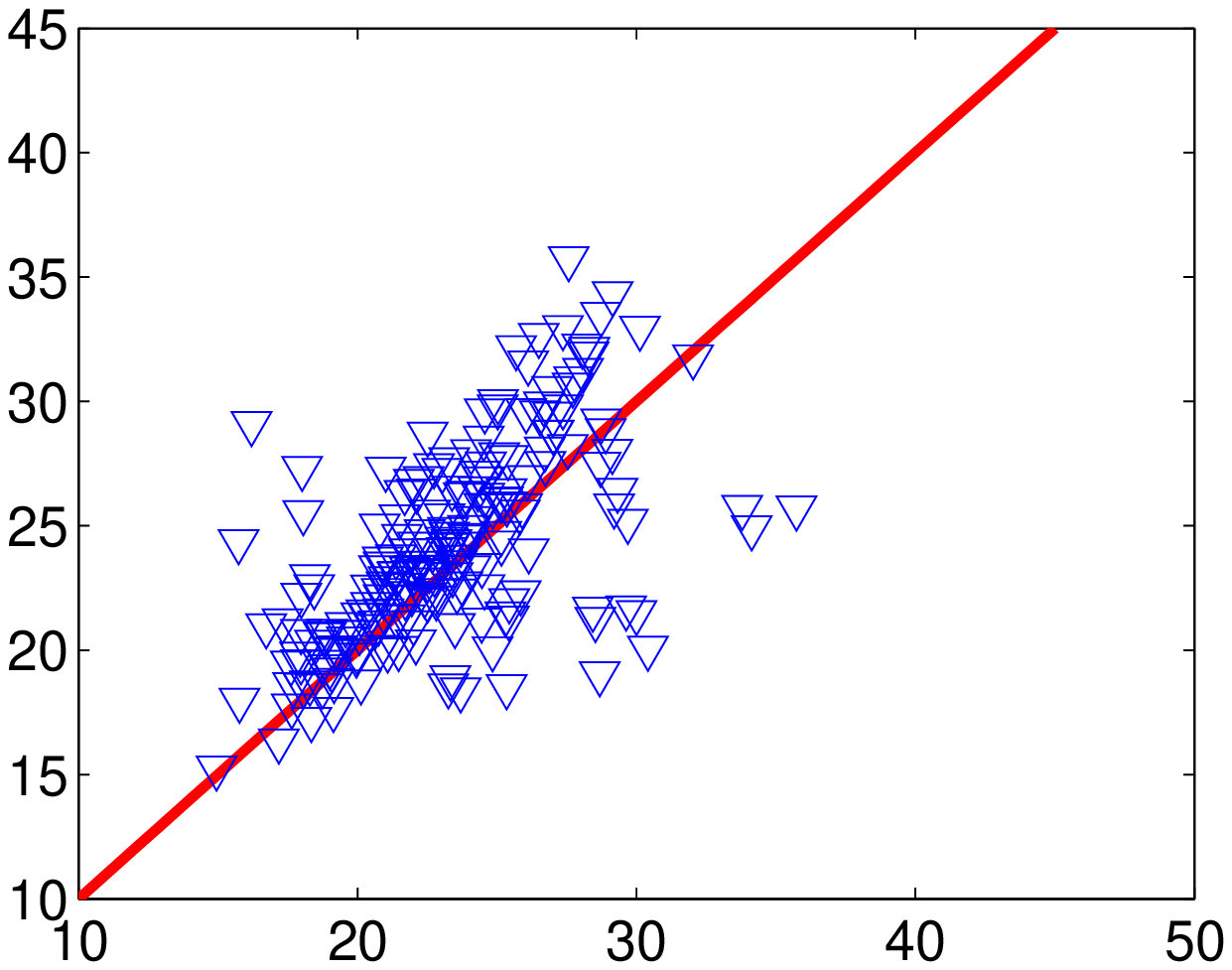}}&\scalebox{0.36}{\includegraphics{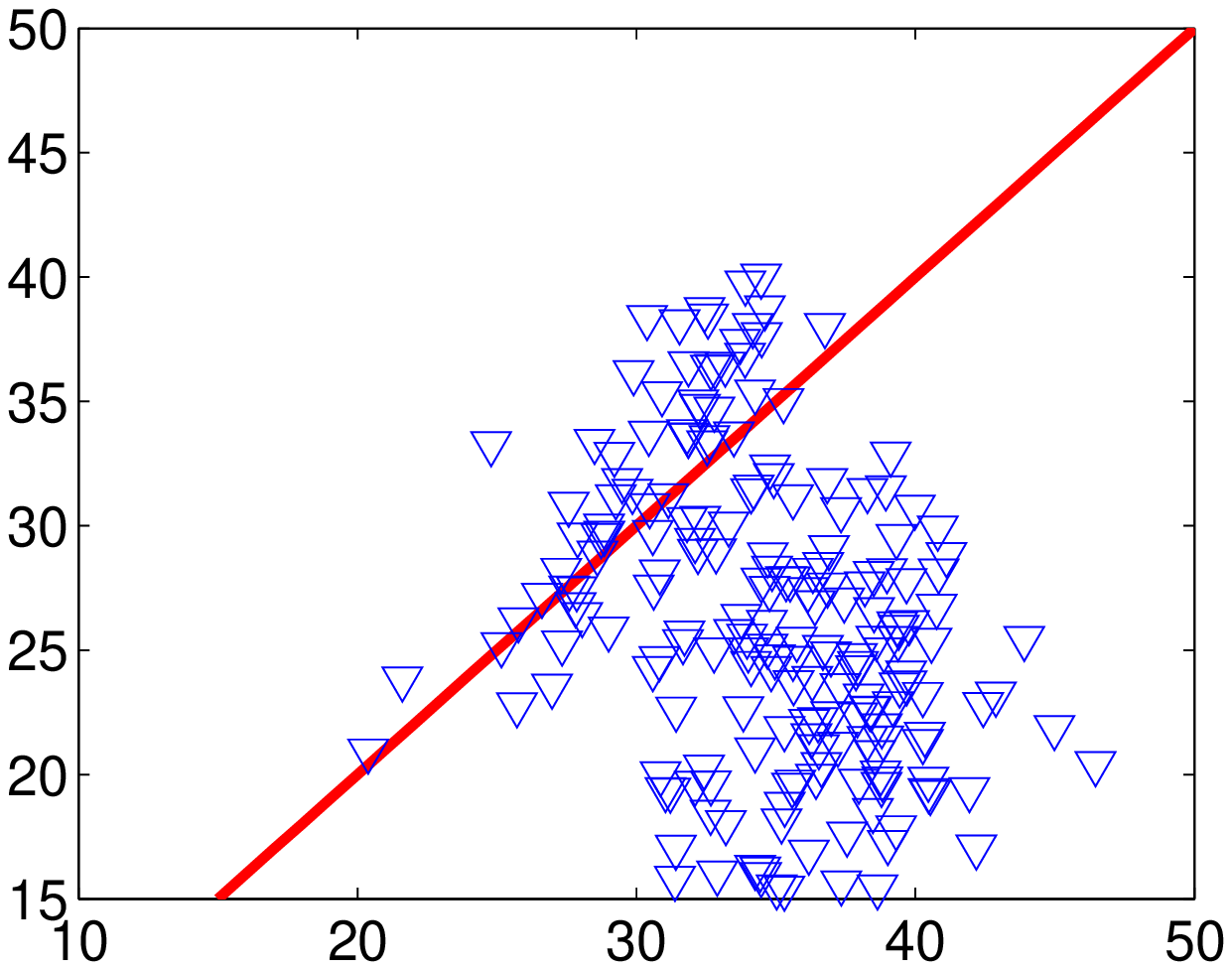}}\\
\scalebox{0.36}{\includegraphics{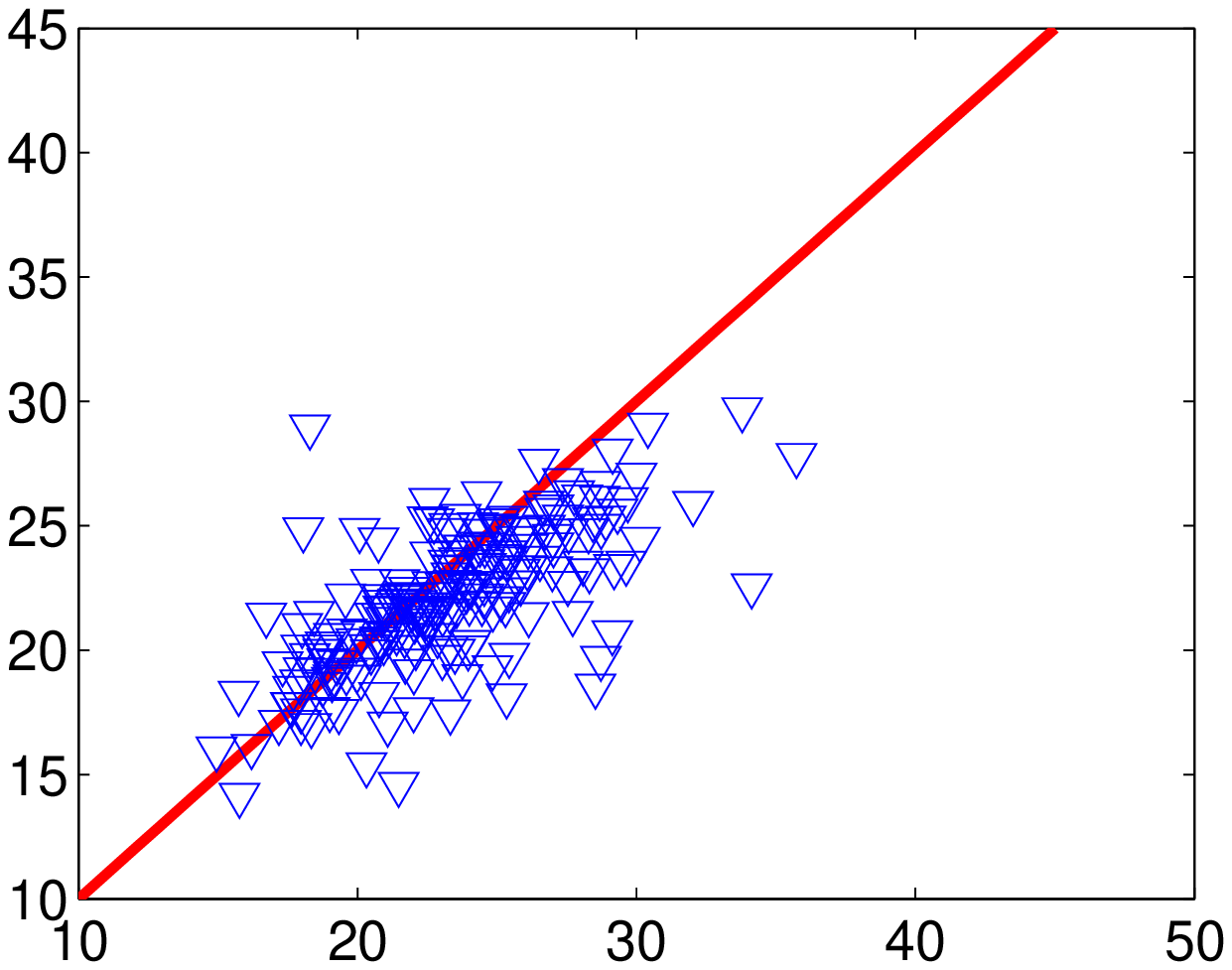}}&\scalebox{0.36}{\includegraphics{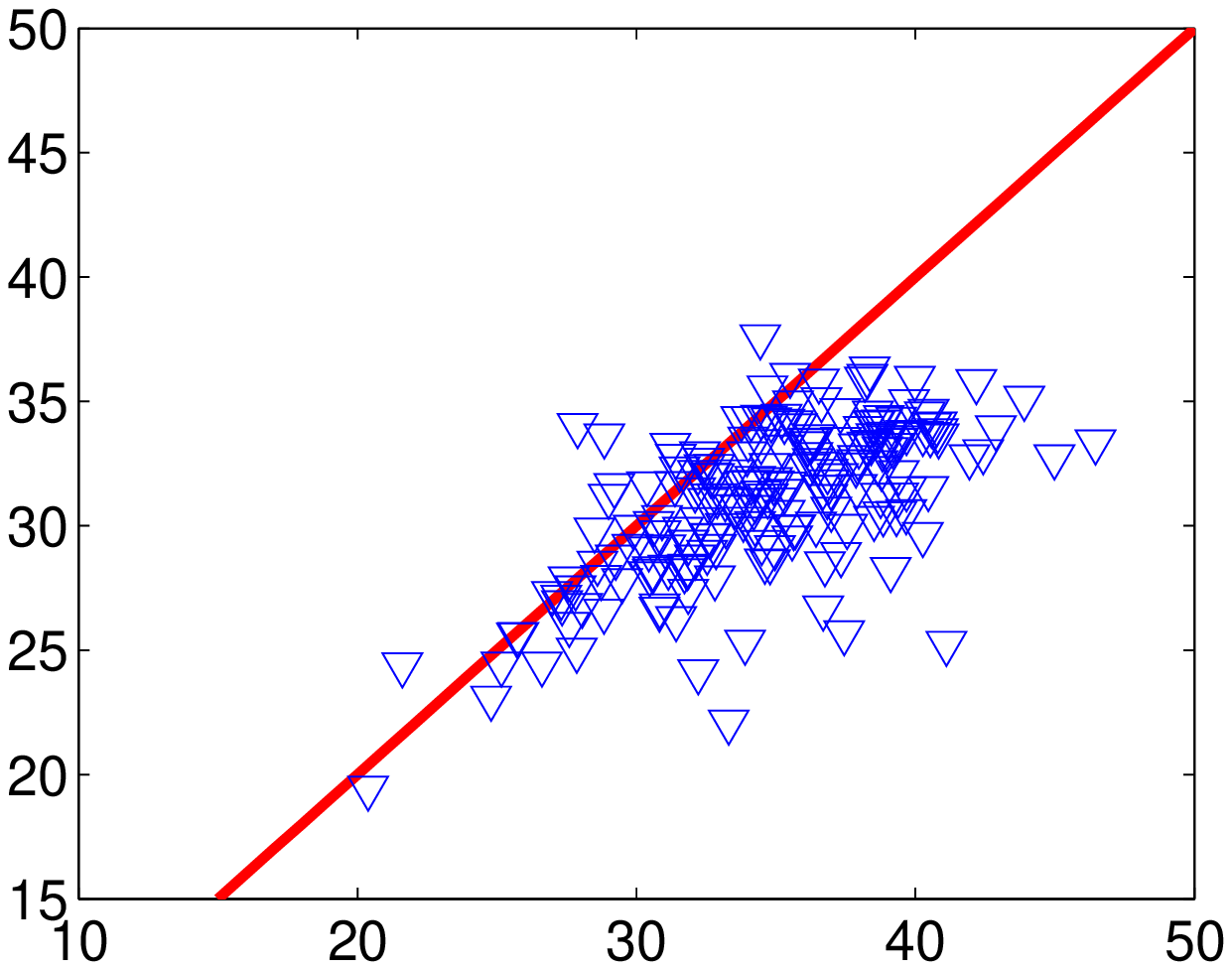}}\\
\scalebox{0.36}{\includegraphics{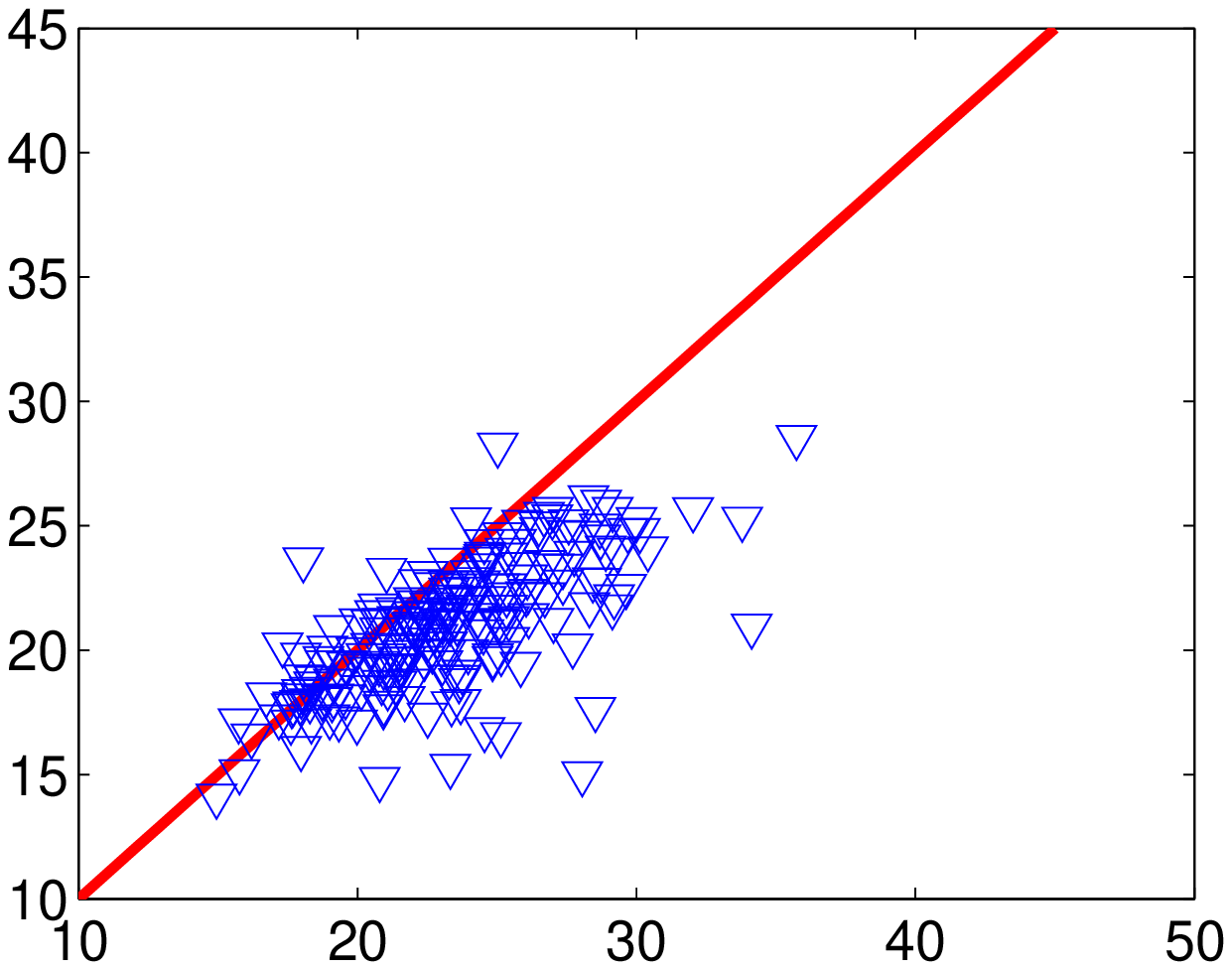}}&\scalebox{0.36}{\includegraphics{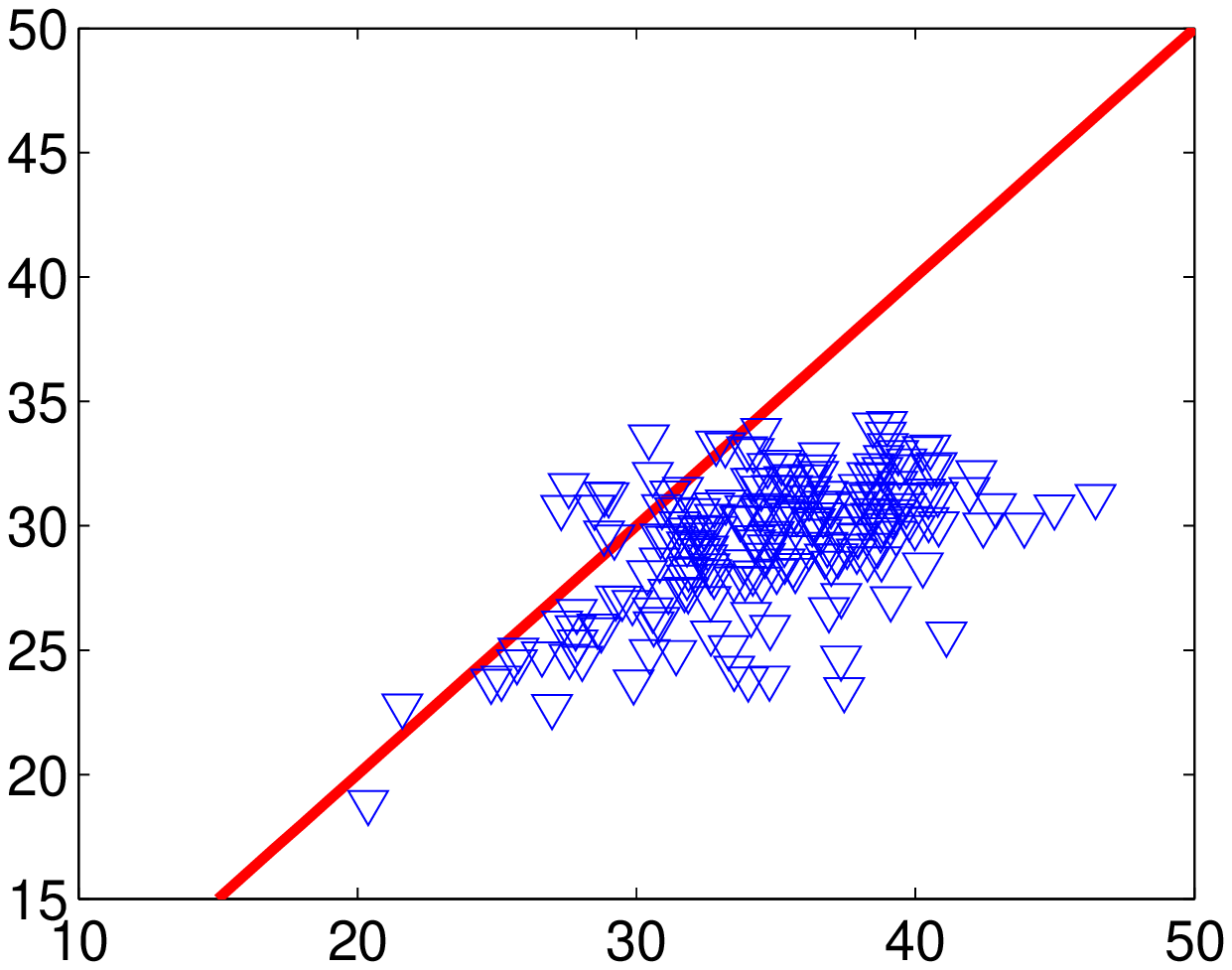}}\\\
(a) $m=500$  & (b) $m=1000$
\end{tabular}
\caption{The marks in the plots (from the top row to the bottom row) represent the pairs of the SNR values from the BIHT with the correct sparsity input (i.e., $s=10$) and incorrect sparsity inputs 8, 9, 11, and 12, respectively. We fix $n=1000$.}
\label{fig:VaryingM-BIHT}
\end{figure}

Figure~\ref{fig:VaryingM-BIHT-Cons} (a) and (b) illustrate the histograms of the numbers of unsatisfied consistency conditions over 200 trials for $m=500$ and $1000$, respectively. We can see from Figure~\ref{fig:VaryingM-BIHT-Cons} (a) that the use of an underestimated sparsity constraint ($s=8$ or $9$) will tend to yield, on average, a solution with a large amount of sign constraints unsatisfied, in other words, under the current setting the solution to model~\eqref{model:BIHT} via the BIHT does not satisfy equation \eqref{model:1bit}. As expected, when an overestimated sparsity constraint ($s=11$ or $12$) is used, the sign constraints are usually satisfied.

In summary, we conclude that a proper chosen sparsity constraint is critical for the success of the BIHT.

\begin{figure}[h]
\centering
\begin{tabular}{cc}
\scalebox{0.46}{\includegraphics{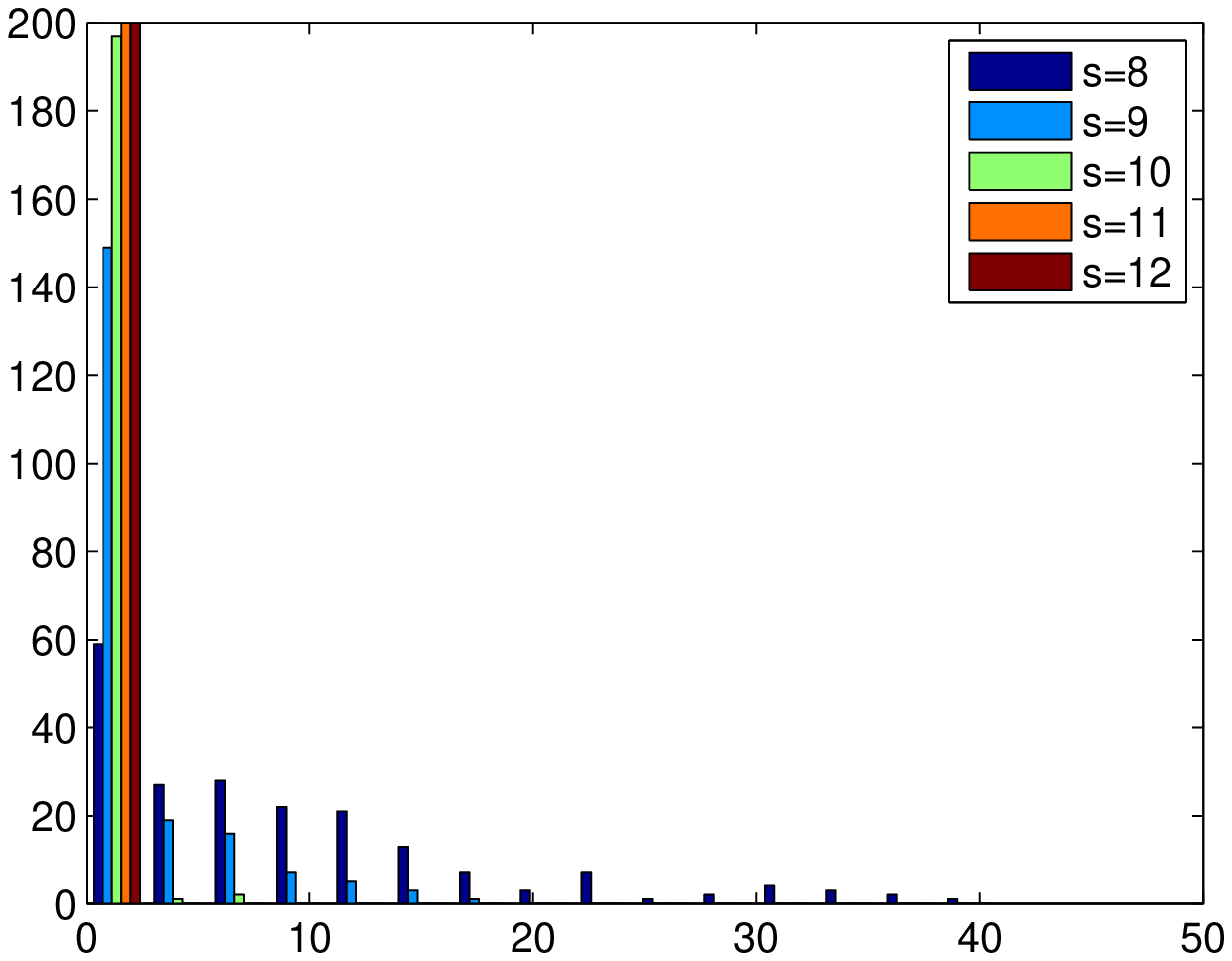}}&
\scalebox{0.46}{\includegraphics{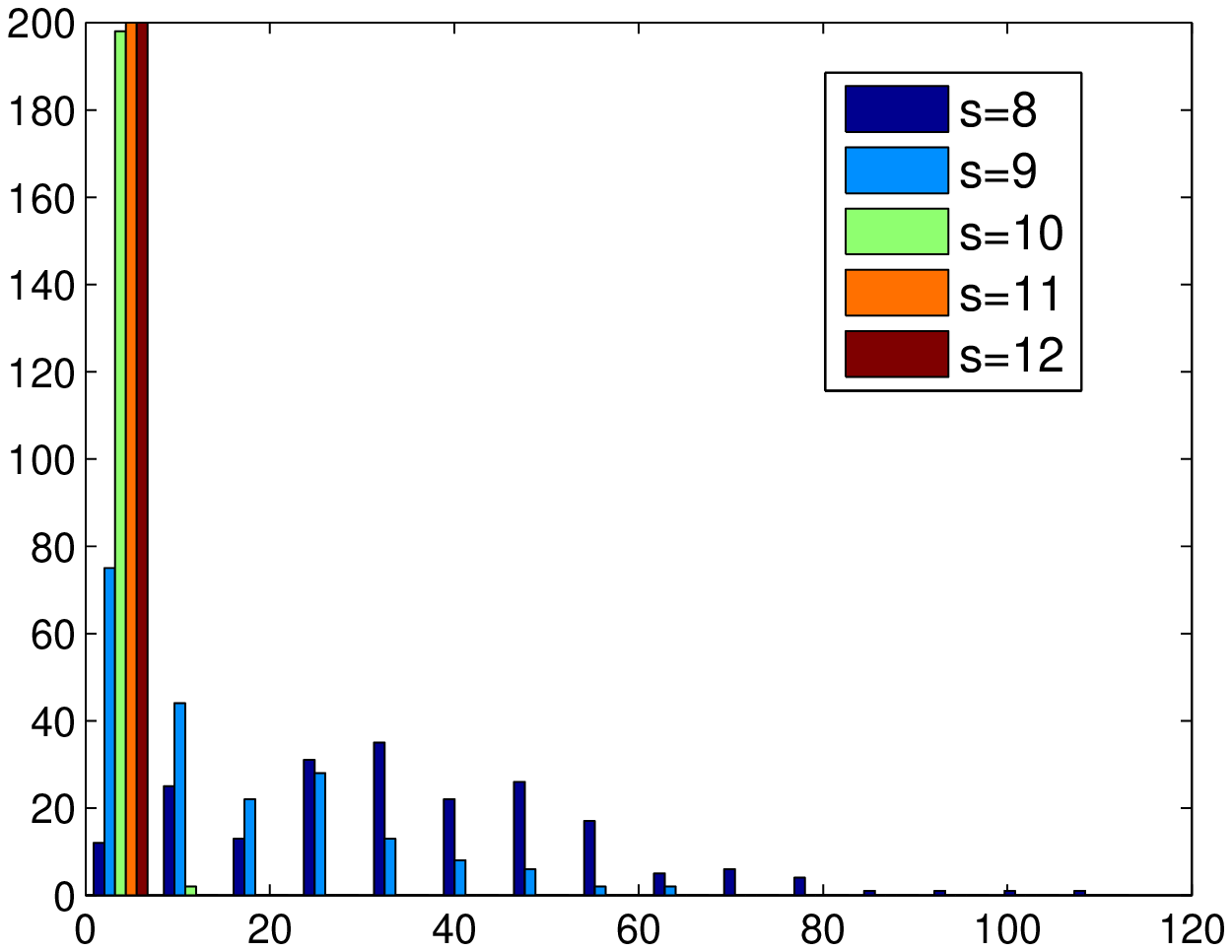}}\\
(a) &  (b)

\end{tabular}
\caption{The histograms of the numbers of unsatisfied consistency conditions over 200 trials with (a) $(m,n)=(500, 1000)$  and (b) $(m,n)=(1000, 1000)$.}
\label{fig:VaryingM-BIHT-Cons}
\end{figure}

\subsection{Plan-Vershynin's model for 1-bit reconstruction}
Both our model~\eqref{model:SSZ} and Plan-Vershynin's model~\eqref{model:plan} use the same constraint conditions. Their objective functions are different. Our model uses the $\ell_0$-norm while Plan-Vershynin's model uses the $\ell_1$-norm. As suggested in \cite{Plan-Vershynin:11}, linear programming can be applied for the Plan-Vershynin model. We report here some numerical results for this model.

In our simulations, we fix $n=1000$, $m=1000$, and $s=10$. All simulations were performed 100 trials. Figure~\ref{fig:PlanVershynin} illustrates the sparsity of the reconstructions of all trials which are clearly greater than 10 (indicated by the solid red line in the figure). The average sparsity  of the reconstructions over 100 trials is $23.42$. Recall that the average SNR values of all reconstructions by the BIHT is $34.74$dB.

\begin{figure}[h]
\centering
\begin{tabular}{c}
\scalebox{0.33}{\includegraphics{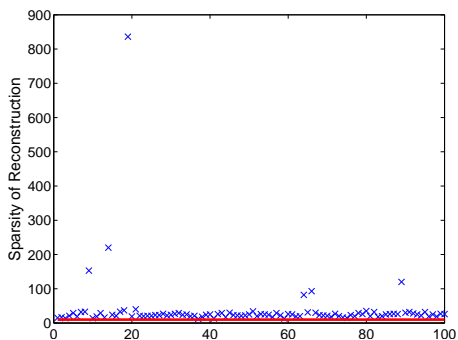}}
\end{tabular}
\caption{Results for Plan-Vershynin's model using Linear programming over 100 trials. }
\label{fig:PlanVershynin}
\end{figure}



\subsection{Performance of Algorithm~\ref{alg:matrix-final}}

Prior to applying Algorithm~\ref{alg:matrix-final} for 1-bit compressive sampling problem, parameters $k_{\max}$, $\tau$, $\alpha_{\max}$, $\epsilon_{\min}$, $\alpha$, and $\epsilon$ in Algorithm~\ref{alg:matrix-final} need to be determined.  Under the aforementioned setting for the random matrix $\Phi$ and sparse signal $x$, we fix $k_{\max}=17$, $\tau=\frac{1}{2}$, $\alpha_{\max} = 8000$, $\epsilon_{\min}=10^{-5}$. For the functions $F_\epsilon$  defined by \eqref{eq:Mangasarain}  and  \eqref{eq:L}, we set the pair of initial parameters $(\alpha,\epsilon)$ as $(500, 0.25)$ and $(250, 0.125)$, respectively. The iterative process in the PD-subroutine is forced to stop if the corresponding number of iteration exceeds $300$. These parameters are used in all simulations performed by Algorithm~\ref{alg:matrix-final} in the rest of this section.

To evaluate the performance of Algorithm~\ref{alg:matrix-final} in terms of SNR values at various scenarios, we consider three configurations for the size of the random matrix $\Phi$ and the sparsity of the vector $x$. In the first configuration, we fix $n=1000$ and $s=10$ and vary $m$ such that the ratio $m/n$ is between $0.1$ and $2$. In the second configuration, we fix $m=1000$ and $n=1000$ and vary the sparsity of $x$ from $1$ to $20$. In the third configuration, we fix $m=1000$ and $s=10$ and vary $n$ from $500$ to $1400$.

For every cases in each configuration, we compare the accuracy of Algorithm~\ref{alg:matrix-final} with the BIHT by computing the average of SNR values over 100 trials. For the given parameters and stopping criteria adopted by Algorithm~\ref{alg:matrix-final}, the estimate $x^{(k_{\max})}$ may not satisfy the consistency condition~\eqref{model:1bit-matrix}, that is the signs of measurements of the estimate $x^{(k_{\max})}$ are not completely consistent with that of the original measurements. Thus, for a fair comparison, we only compute the average of SNR values for those trials that both reconstructions from the BIHT and Algorithm~\ref{alg:matrix-final} satisfy the consistency condition~\eqref{model:1bit-matrix} and we say the corresponding trials are valid.

For the first configuration, the SNR values in decibels of the average reconstruction errors by both the BIHT and Algorithm~\ref{alg:matrix-final} are depicted in Figure~\ref{fig:VaryingM-Alg}. The plots demonstrate that our proposed algorithm performs as equally good as the BIHT, in particular, when $m/n$ is greater than $1$, even thought our algorithm does not require to know the exact sparsity of the original signal. We can see that Algorithm~\ref{alg:matrix-final} with the Log-Det function \eqref{eq:L} (Figure~\ref{fig:VaryingM-Alg}(b)) performs slightly better than with the Mangasarian function \eqref{eq:Mangasarain} (Figure~\ref{fig:VaryingM-Alg}(a)) .

\begin{figure}[ht]
\centering
\begin{tabular}{cc}
\scalebox{0.32}{\includegraphics{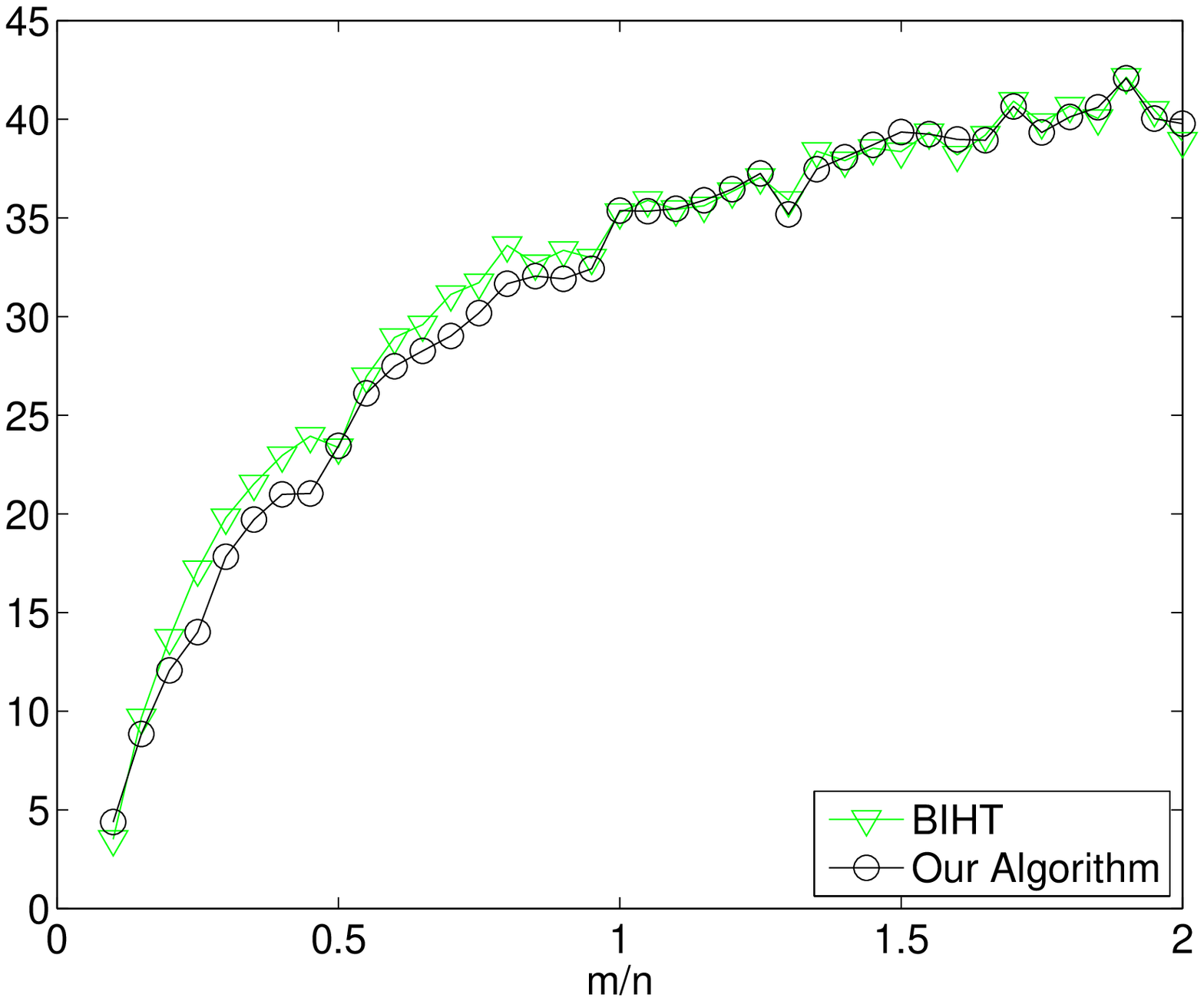}}&
\scalebox{0.32}{\includegraphics{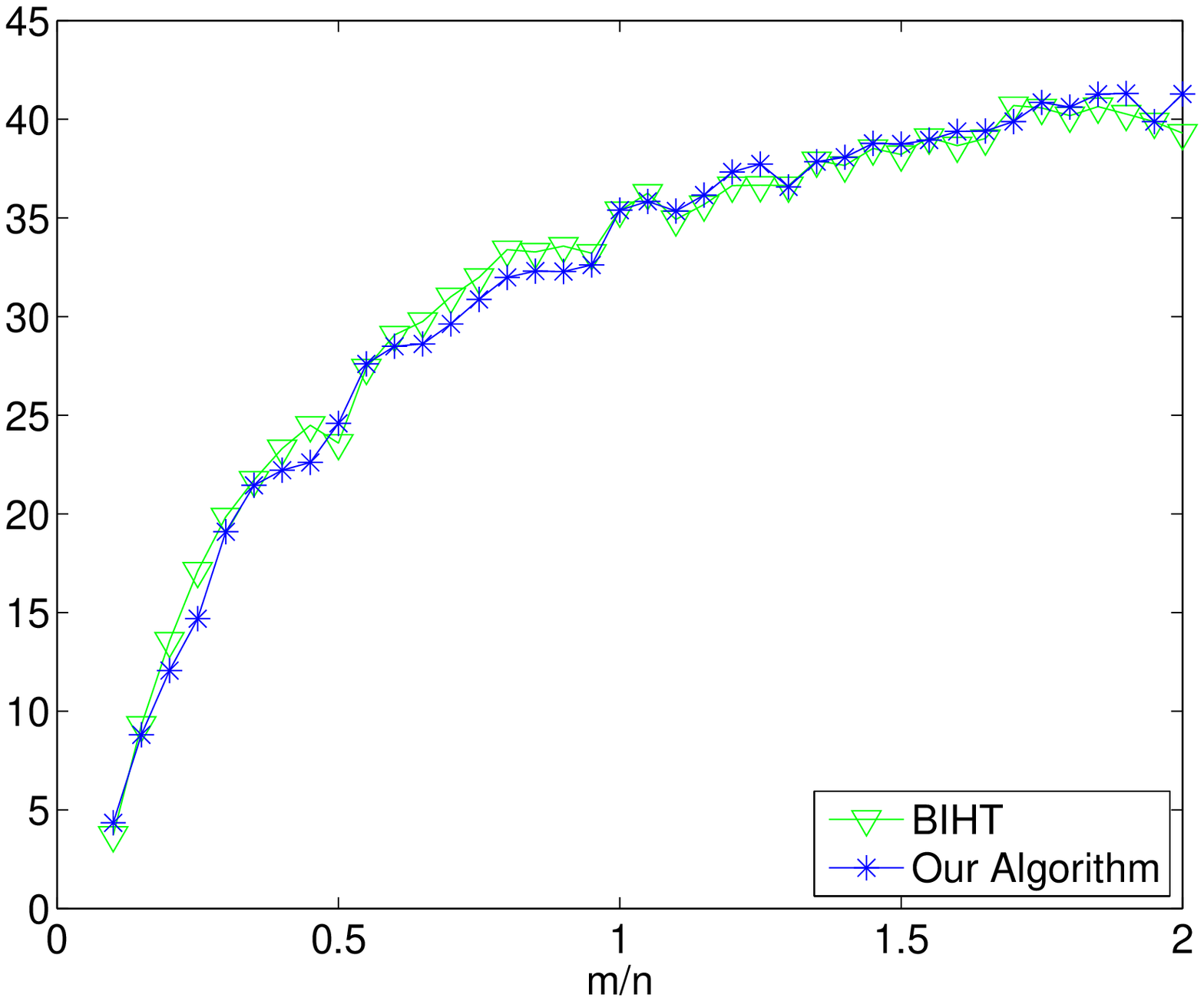}}\\
(a) Mangasarian &(b) Log-Det
\end{tabular}
\caption{Average SNR values vs. $m/n$ for fixed $n=1000$ and $s=10$. }
\label{fig:VaryingM-Alg}
\end{figure}

\begin{figure}[htb]
\centering
\begin{tabular}{cc}
\scalebox{0.31}{\includegraphics{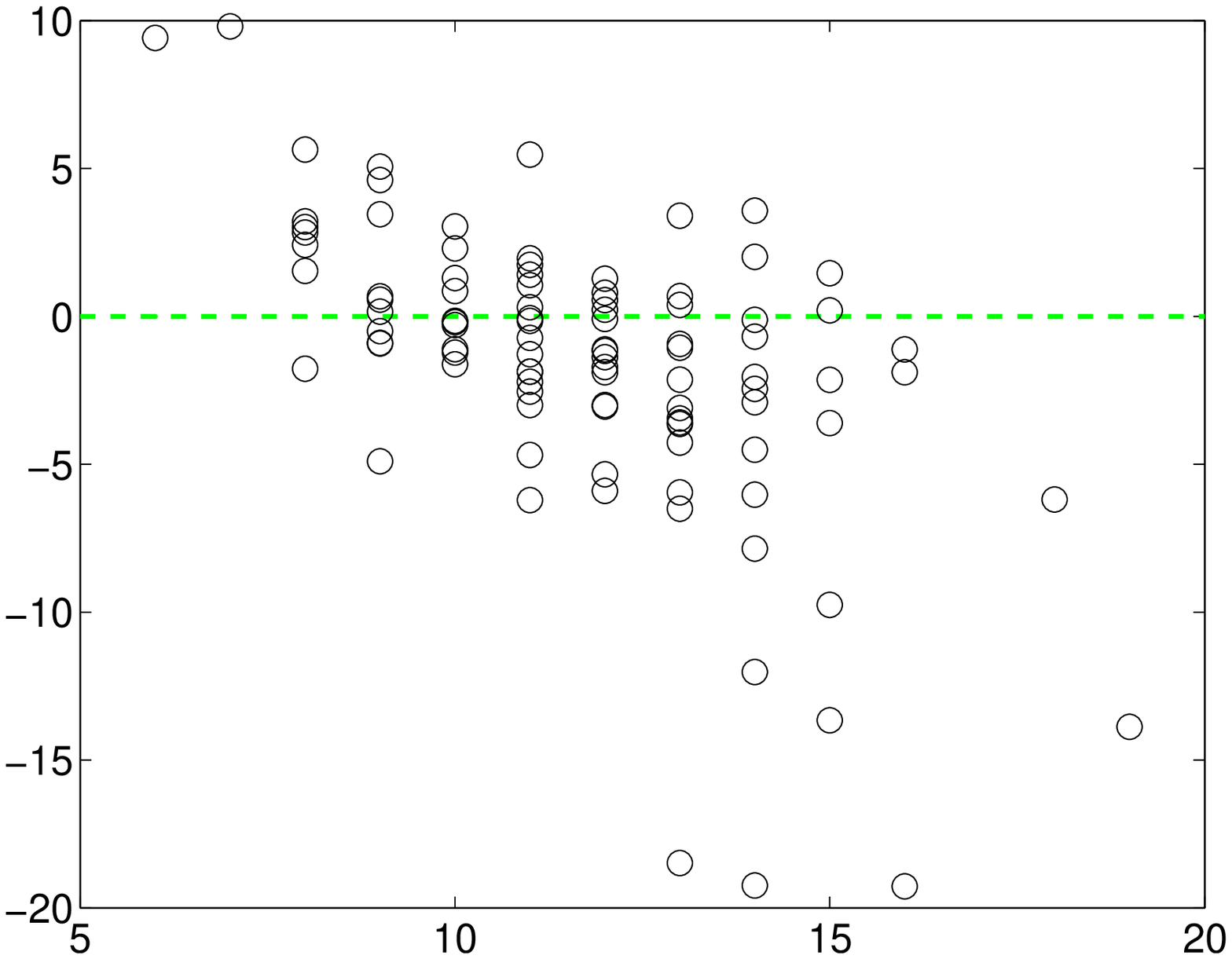}}&\scalebox{0.31}{\includegraphics{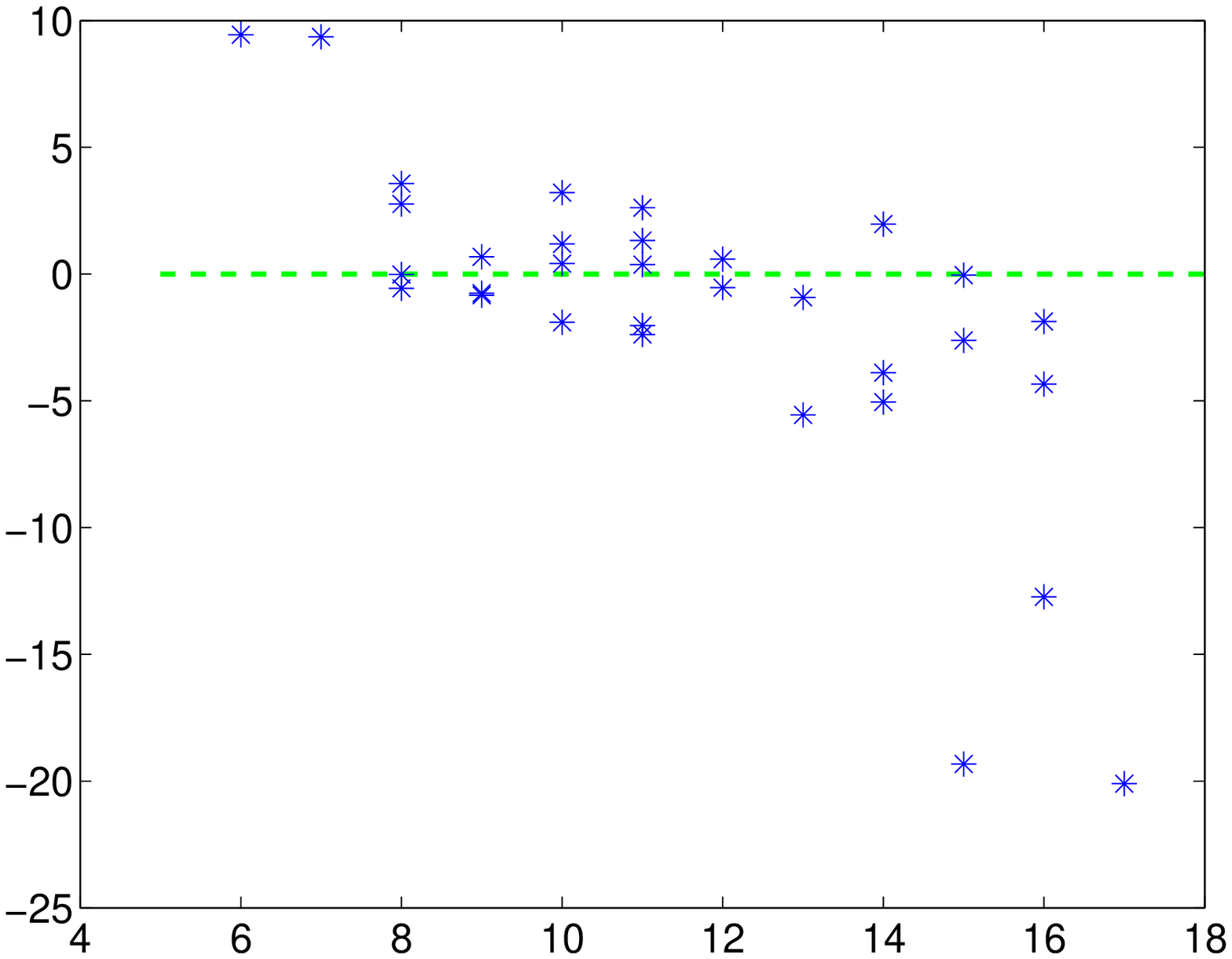}}\\
\scalebox{0.31}{\includegraphics{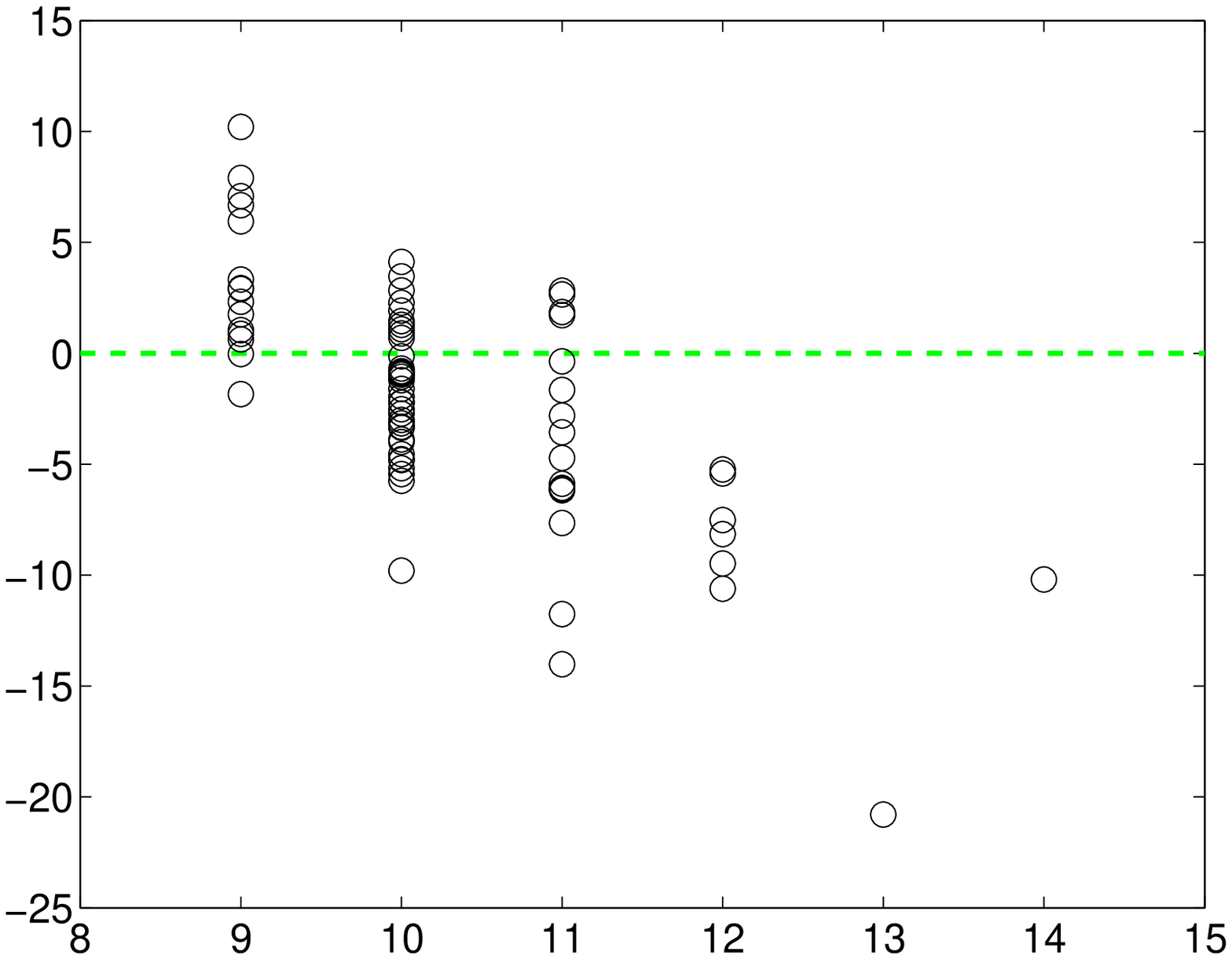}}&\scalebox{0.31}{\includegraphics{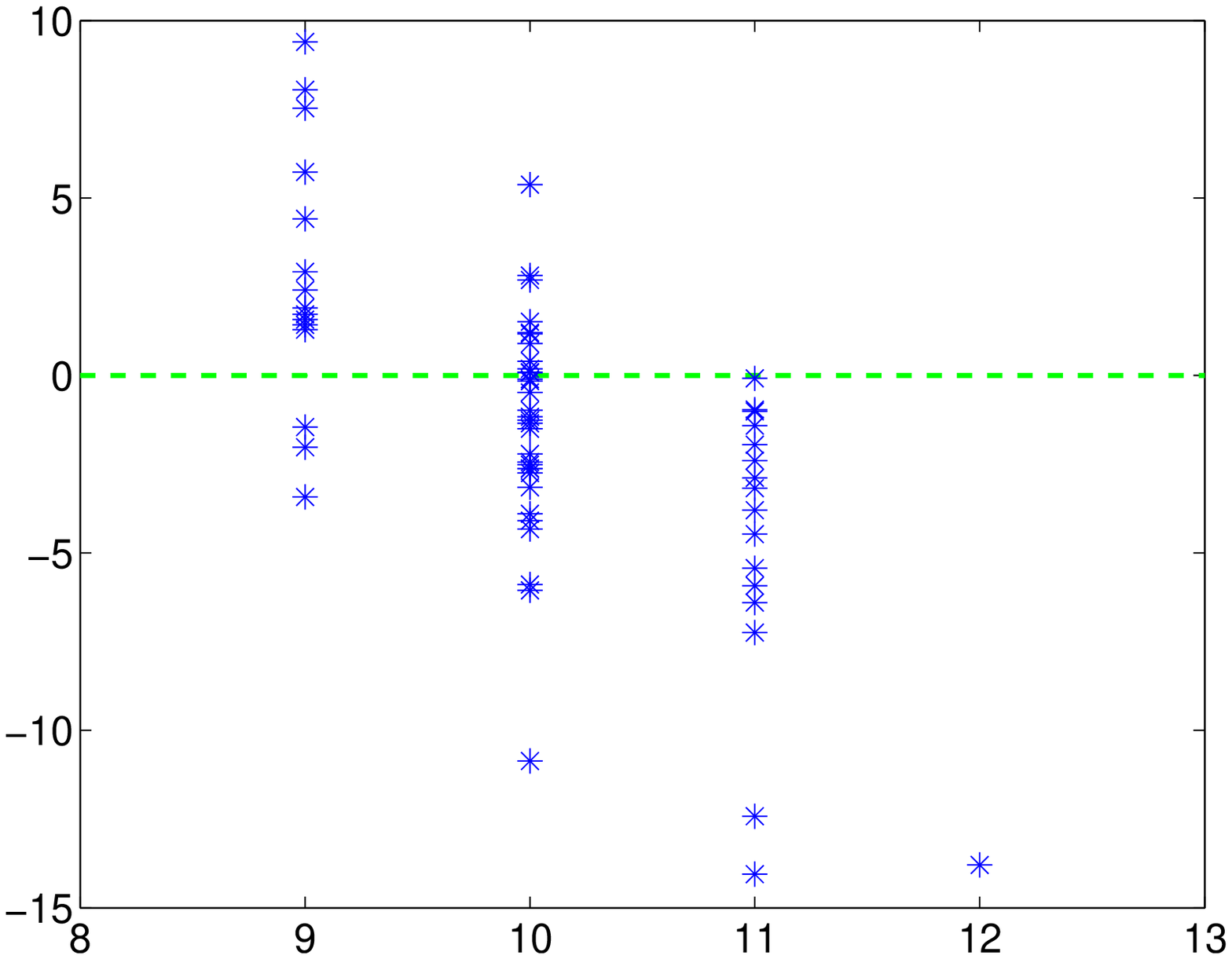}}\\
\scalebox{0.31}{\includegraphics{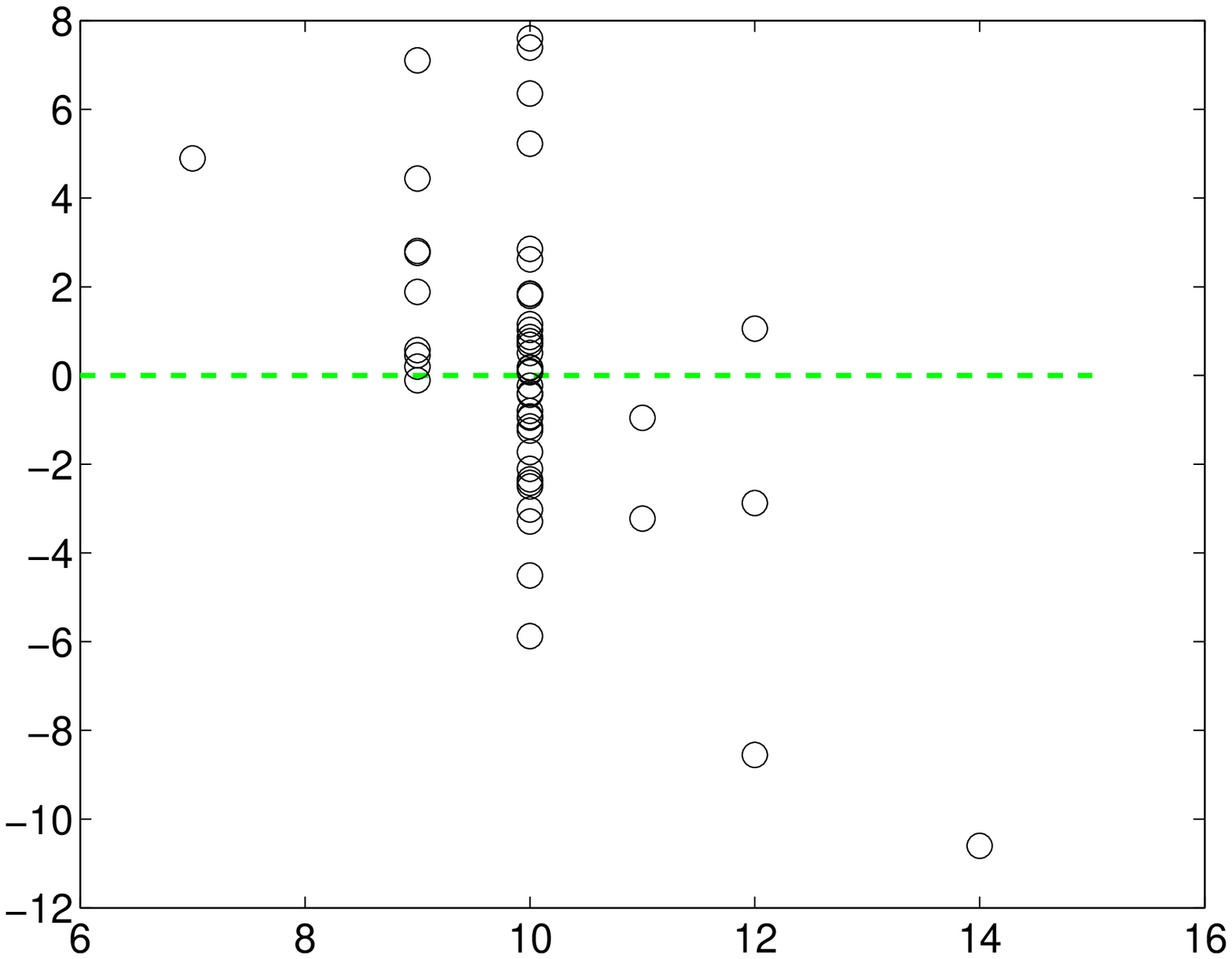}}&\scalebox{0.31}{\includegraphics{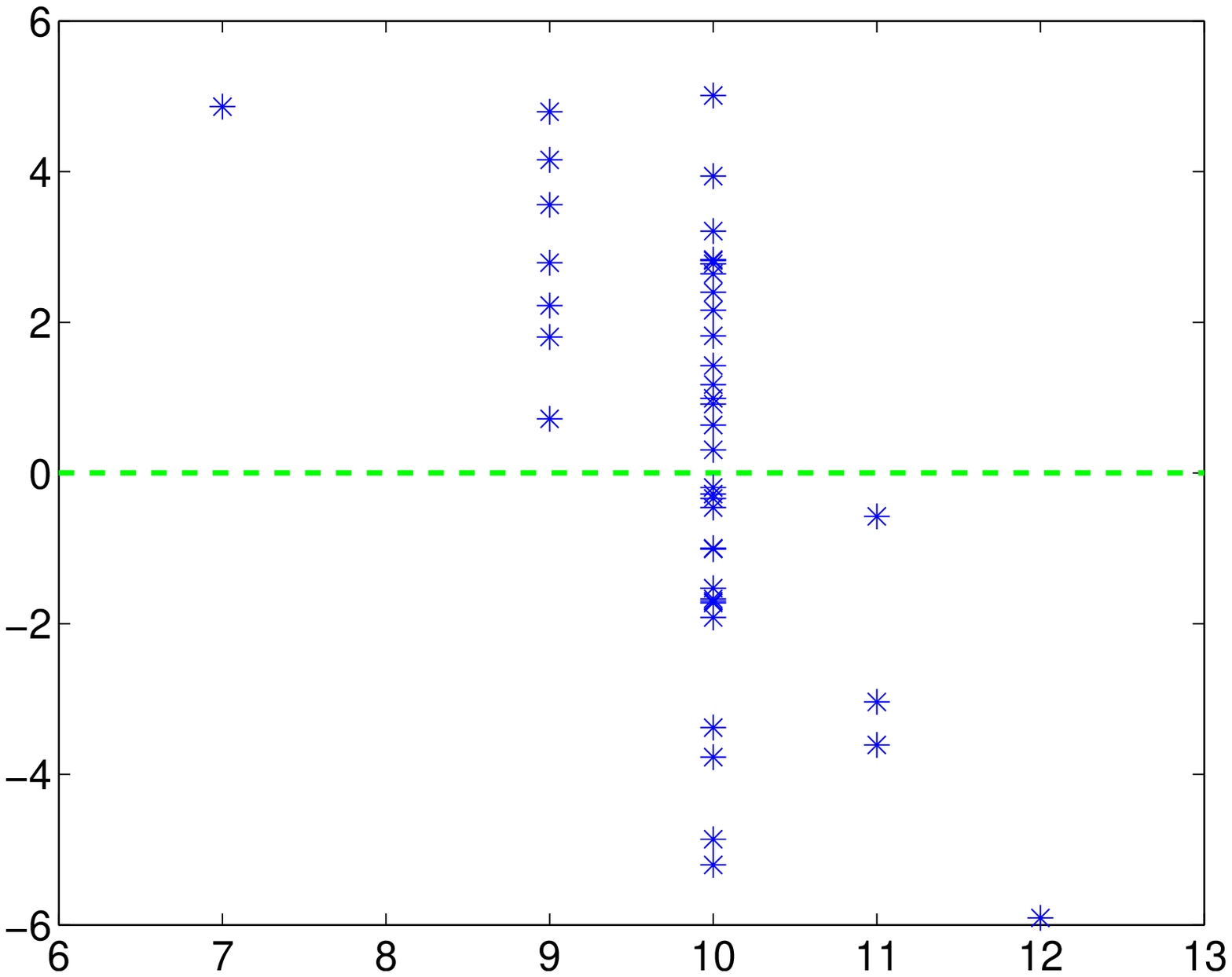}}\\
\scalebox{0.31}{\includegraphics{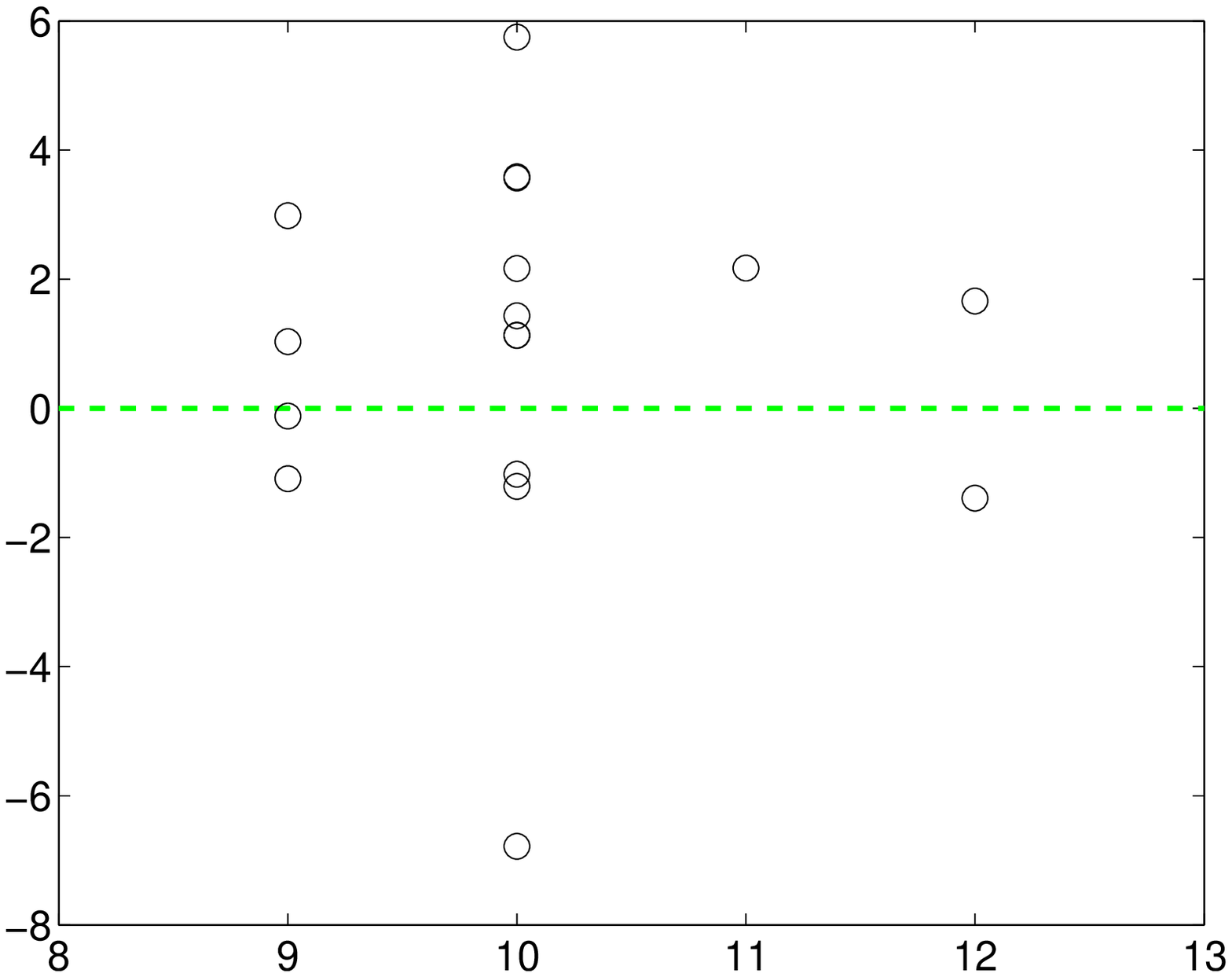}}&\scalebox{0.31}{\includegraphics{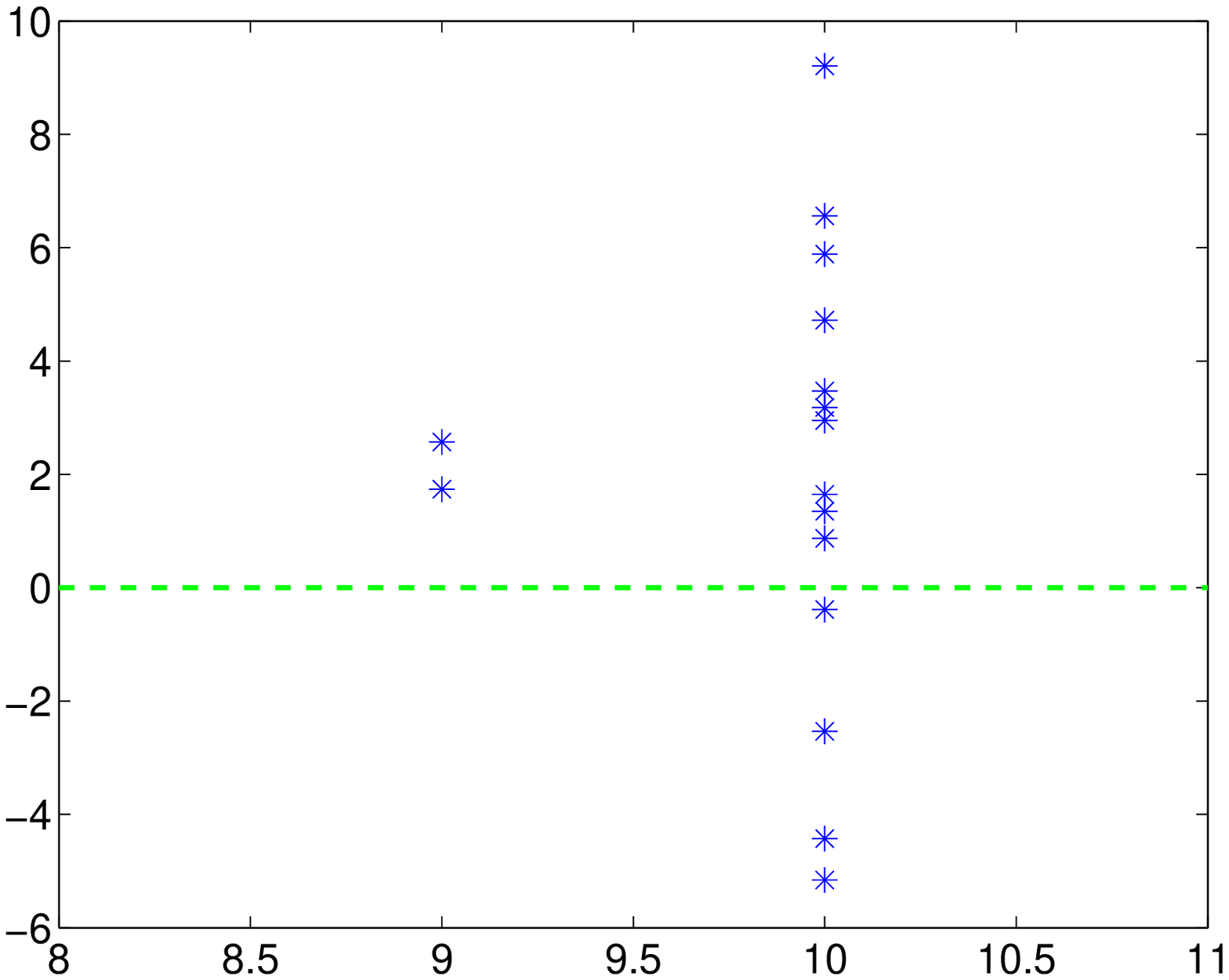}}\\
(a) Mangasarian &(b) Log-Det
\end{tabular}
\caption{The difference of SNR values of estimates by Algorithm~\ref{alg:matrix-final} and the BIHT vs. sparsity of reconstructed estimates by Algorithm~\ref{alg:matrix-final}. We fix $n=1000$ and $s=10$. Row 1 to Row 4 are corresponding to $m$ being 200, 800, 1400, and 2000, respectively.}
\label{fig:VaryingM-Alg-Cases}
\end{figure}

Detailed descriptions for valid trials for $m=200$, $800$, $1400$, and $2000$ are displayed in the rows (from top to bottom) of Figure~\ref{fig:VaryingM-Alg-Cases}, respectively. The horizontal axis of each plot represents the sparsity of the reconstructed signals by Algorithm~\ref{alg:matrix-final} while the vertical axis represents the difference of the SNR values of the reconstructions between Algorithm~\ref{alg:matrix-final} and the BIHT. Therefore, the marks (``$\circ$'' and ``$\star$'') above the dashed horizontal lines indicate that Algorithm~\ref{alg:matrix-final} performs better than the BIHT for the corresponding trials. Since all ideal signals in our simulations are 10-sparse, the marks whose horizontal axis are bigger than, exactly equal to, or smaller than $10$ imply that the $\ell_0$-norm of the reconstructions by Algorithm~\ref{alg:matrix-final} are bigger than, exactly equal to, or smaller than $10$, respectively. Thus, the $\ell_0$-norm of a reconstruction over $10$ indicates that the reconstruction is not a global minimizer of model~\eqref{model:SSZ},  the $\ell_0$-norm of a reconstruction being $10$ indicates that the sparsity of the reconstruction is consistent with the one of the original test signal, the $\ell_0$-norm of a reconstruction below $10$ indicates that the reconstruction is potentially a global minimizer of model~\eqref{model:SSZ} and the original test signal is not a solution to model~\eqref{model:SSZ}. We can conclude from Figure~\ref{fig:VaryingM-Alg-Cases} that (i) the reconstructions by Algorithm~\ref{alg:matrix-final} with sparsity higher (res. lower) than 10 usually have lower (res. higher) SNR values than that by the BIHT; (ii) Increasing $m$ (number of measurements) tends to reduce the sparsity of the reconstructions. For example, average sparsity of the reconstructions for $m=200$, $800$, $1400$, and $2000$ are, respectively, 11.74, 10.26, 10, and 10.06 for Algorithm~\ref{alg:matrix-final} with the Mangasarian function,  and 11.53, 10.05, 9.88, 9.88 for Algorithm~\ref{alg:matrix-final} with the Log-Det function.

For the second configuration, the SNR values in decibels of the average reconstruction errors by both the BIHT and Algorithm~\ref{alg:matrix-final} are compared in Figure~\ref{fig:VaryingK-Alg} for varying sparsity of original signals. The plots demonstrate that our proposed algorithm performs better than the BIHT for sparsity $s$ being $2$ and $6$ to $10$. We emphasize again that unlike the BIHT the exact sparsity of the original signal is not required in advance by Algorithm~\ref{alg:matrix-final}. We remark that when $s=1$ both the BIHT and Algorithm~\ref{alg:matrix-final} find an exact solution to model~\eqref{model:SSZ}. This phenomenon was also reported in \cite{Laska-Wen-Yin-Baraniuk:IEEESP:11}. Detailed descriptions for valid trials for $s=2$, $8$, $14$, and $20$ are displayed in the rows (from top to bottom) of Figure~\ref{fig:VaryingK-Alg-Cases}, respectively. The marks in each plot of Figure~\ref{fig:VaryingK-Alg-Cases} have the same meaning as that in Figure~\ref{fig:VaryingM-Alg-Cases}. For fixed $m=1000$ and $n=1000$  we can draw conclusions from Figure~\ref{fig:VaryingK-Alg-Cases} that (i) Algorithm~\ref{alg:matrix-final} tends to produce an estimate whose sparsity is consistent with the ideal sparse signal; (ii) Algorithm~\ref{alg:matrix-final} can give an estimate whose sparsity is smaller than that of the ideal sparse signal, in particular, when the sparsity of an original signal is relative large.

\begin{figure}[htb]
\centering
\begin{tabular}{cc}
\scalebox{0.32}{\includegraphics{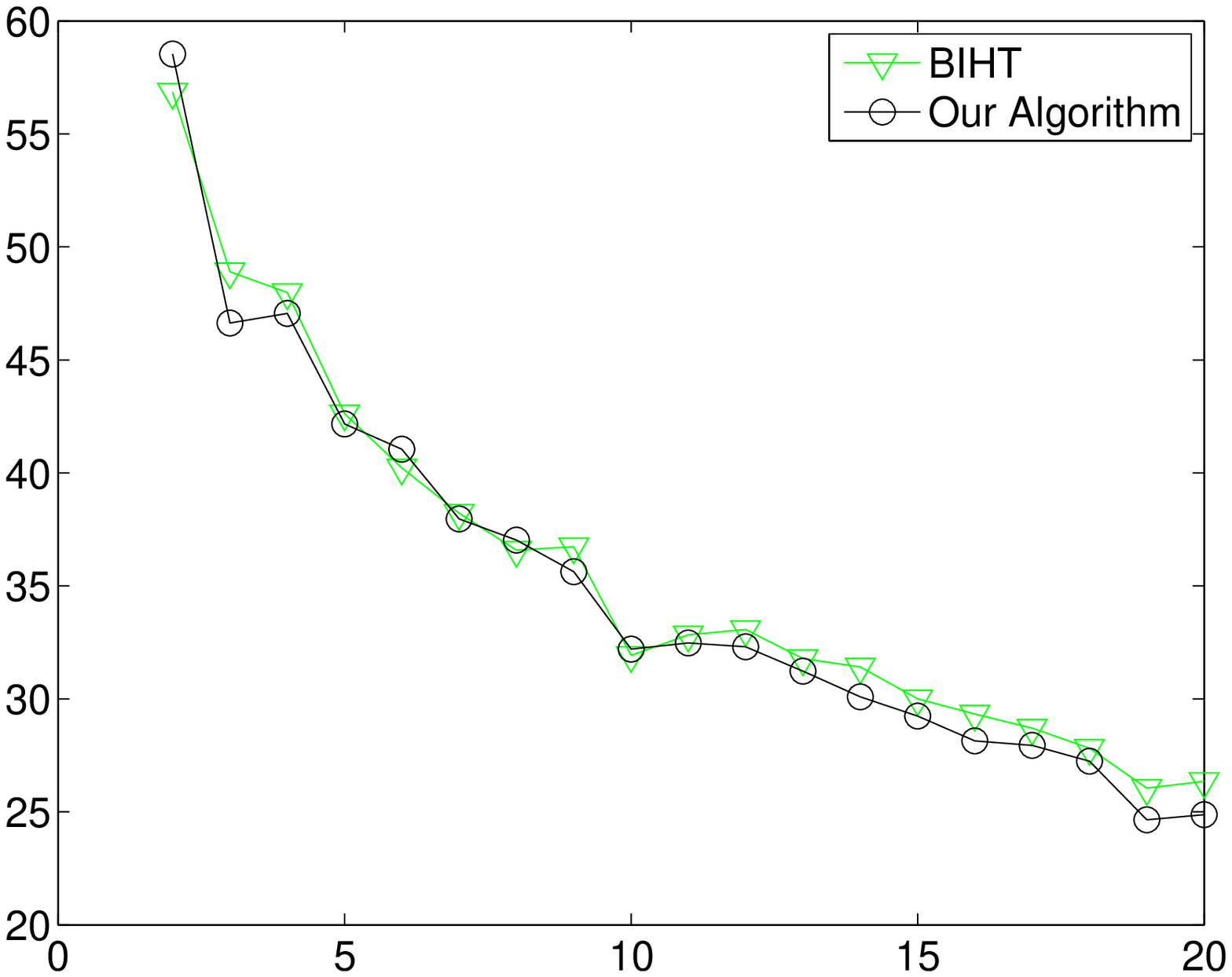}}&
\scalebox{0.32}{\includegraphics{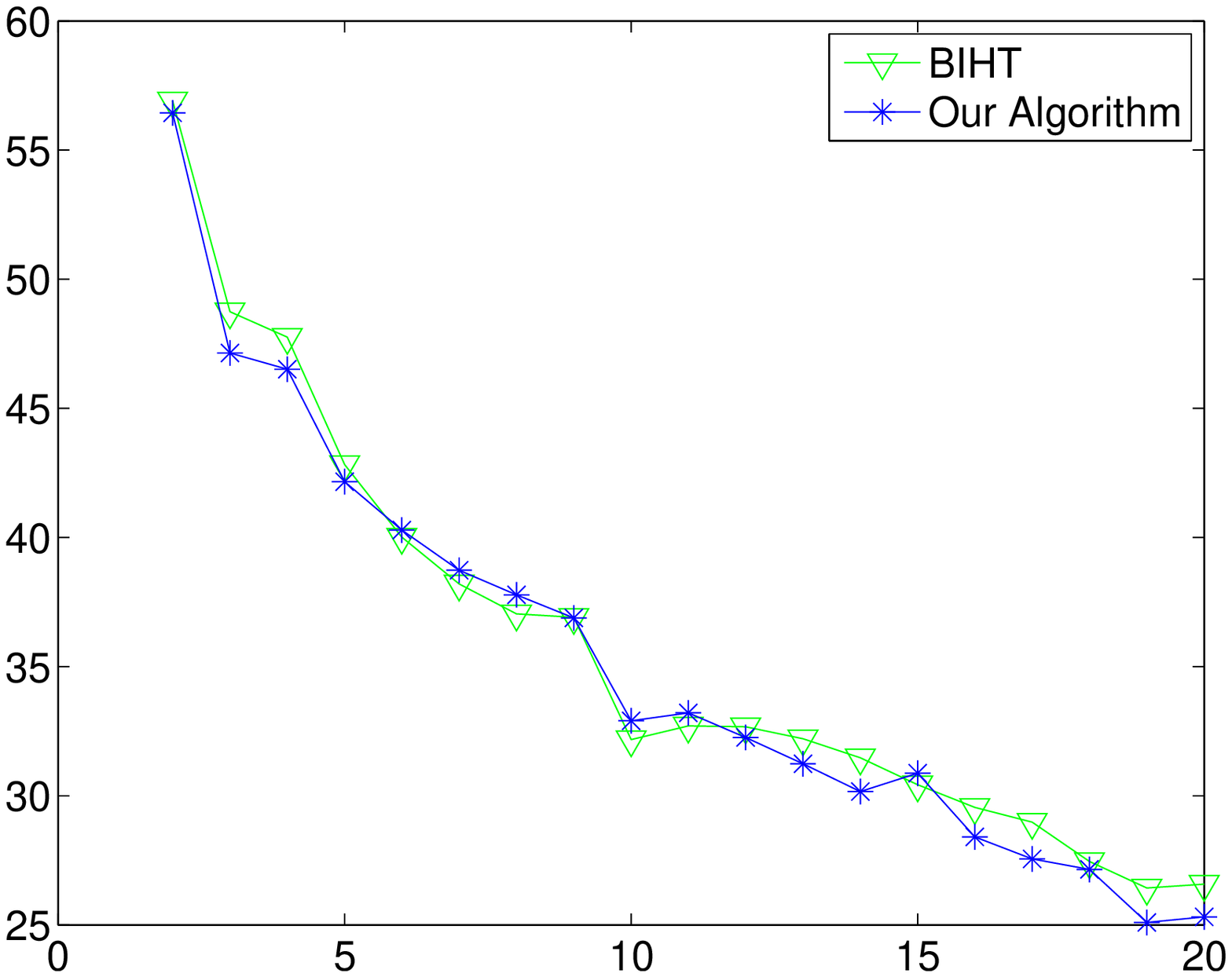}}\\
(a) Mangasarian &(b)  Log-Det
\end{tabular}
\caption{Average SNR values vs. sparsity of original signals for fixed $n=1000$ and $m=1000$. }
\label{fig:VaryingK-Alg}
\end{figure}

\begin{figure}[htb]
\centering

\begin{tabular}{cc}
\scalebox{0.31}{\includegraphics{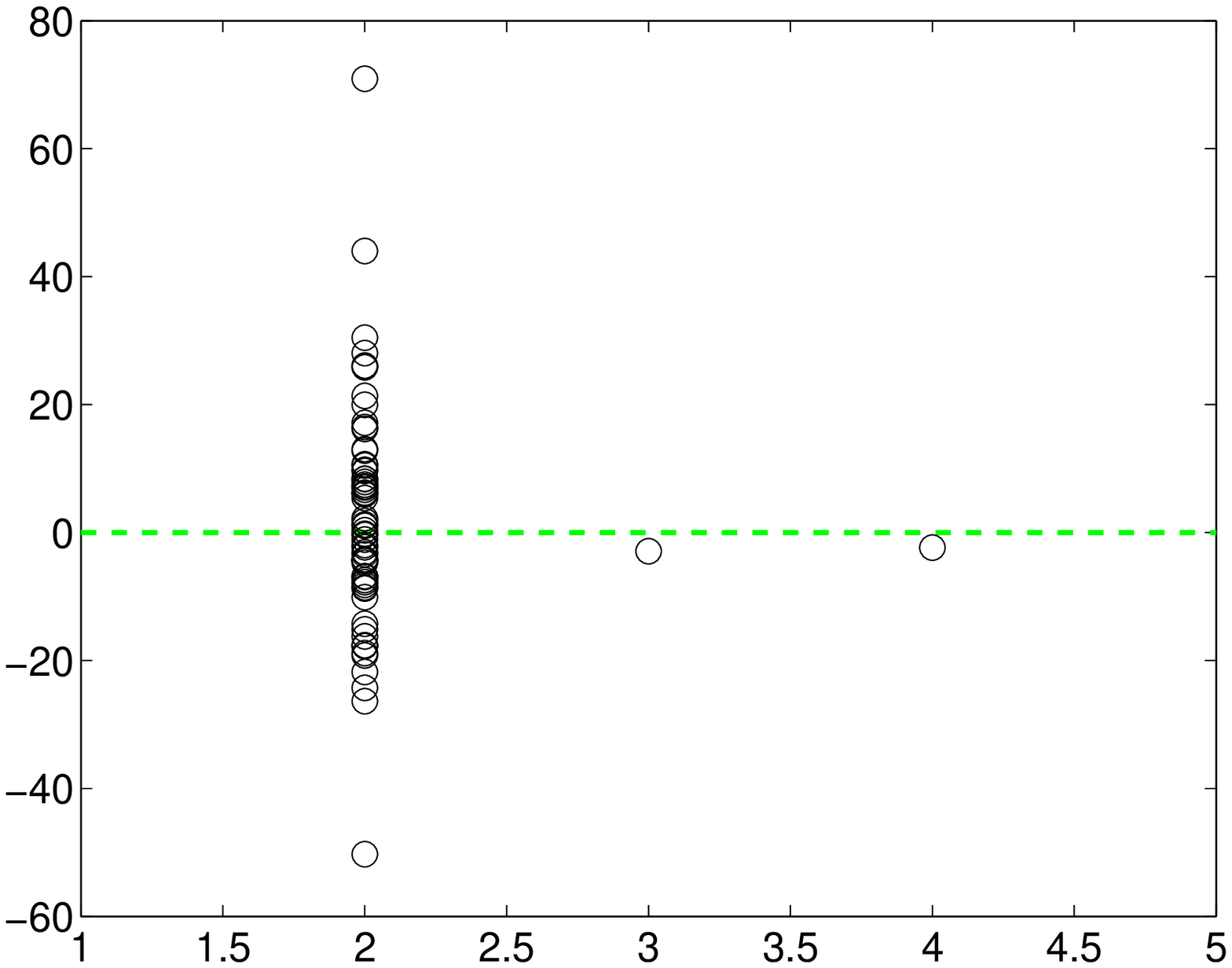}}&\scalebox{0.31}{\includegraphics{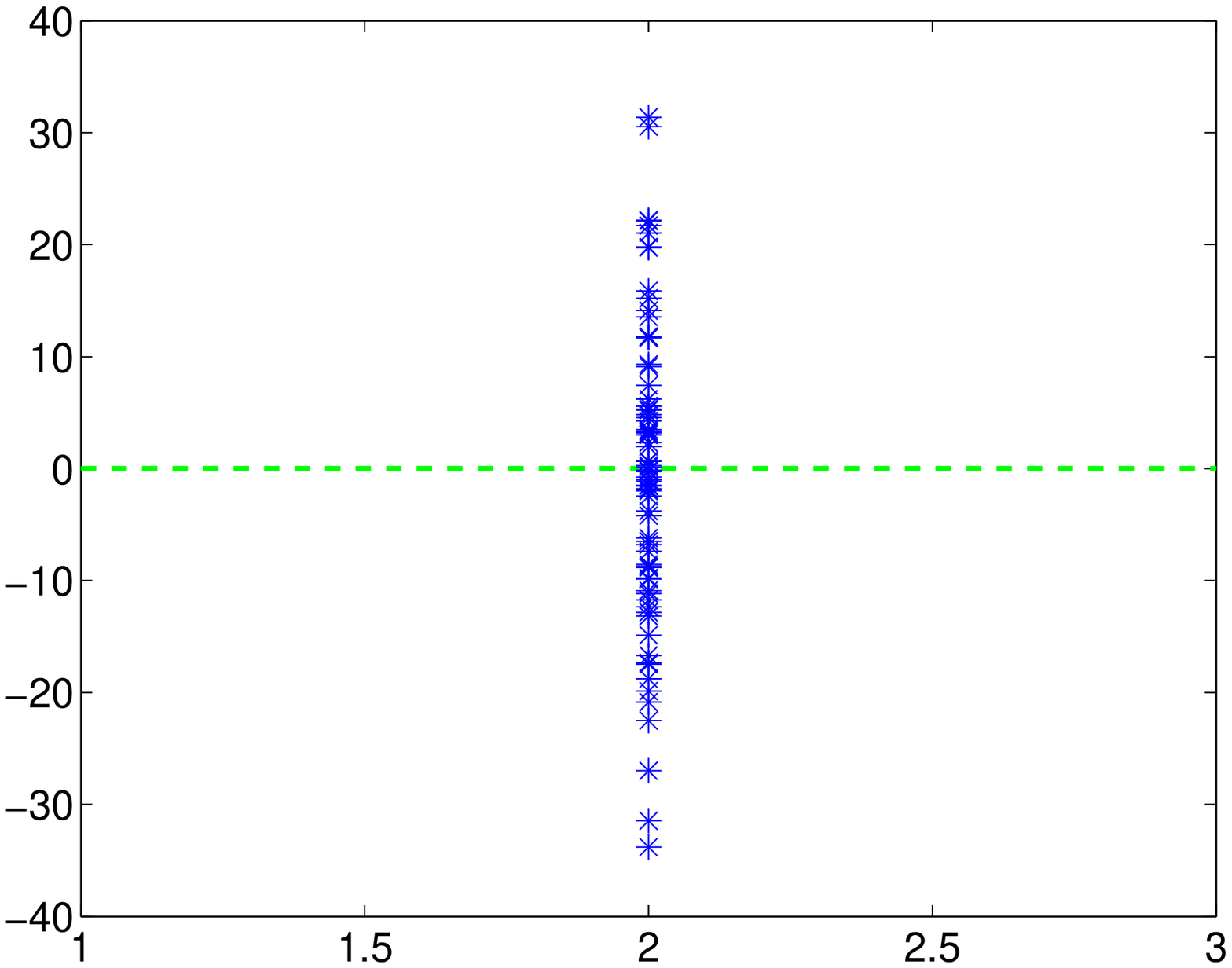}}\\
\scalebox{0.31}{\includegraphics{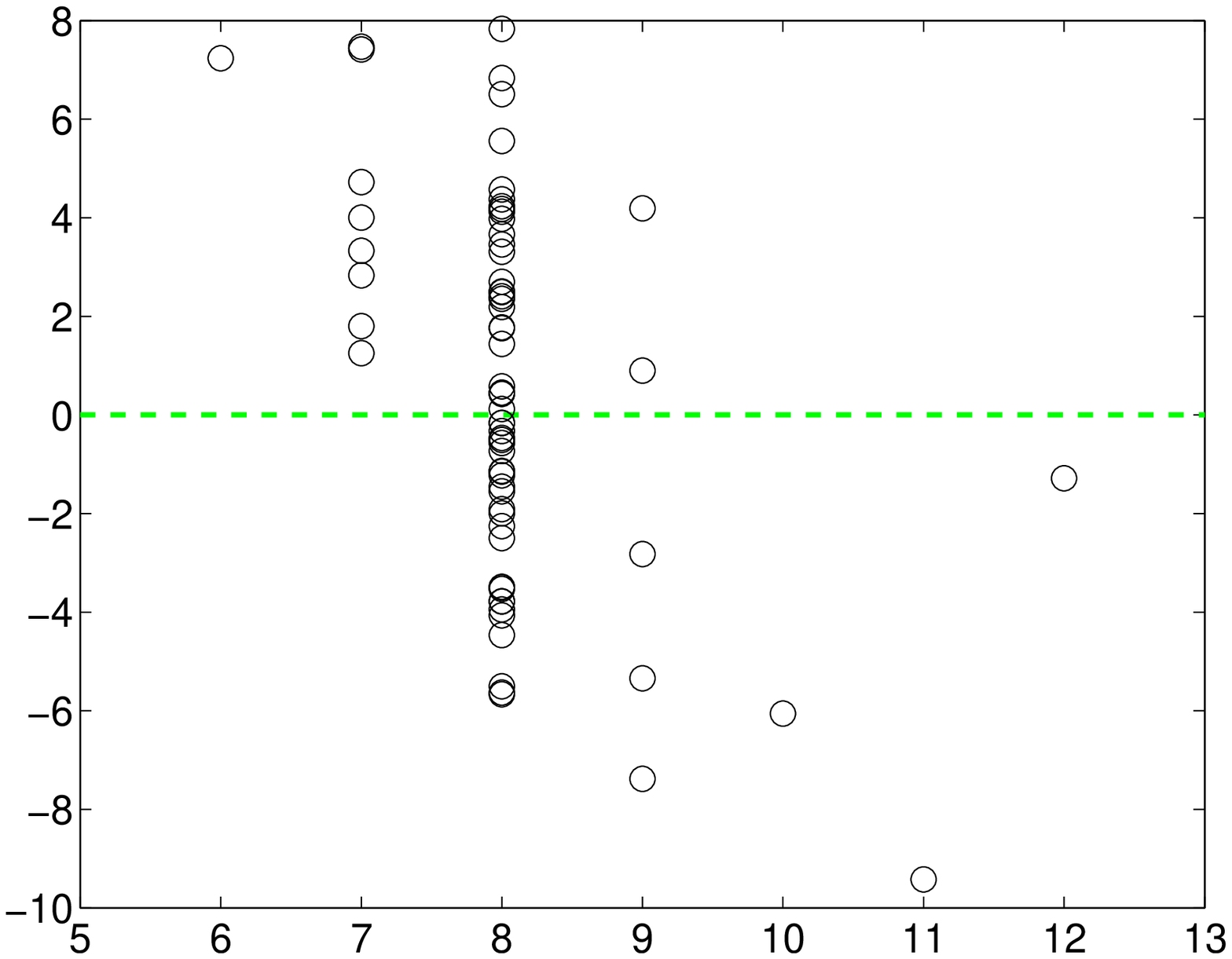}}&\scalebox{0.31}{\includegraphics{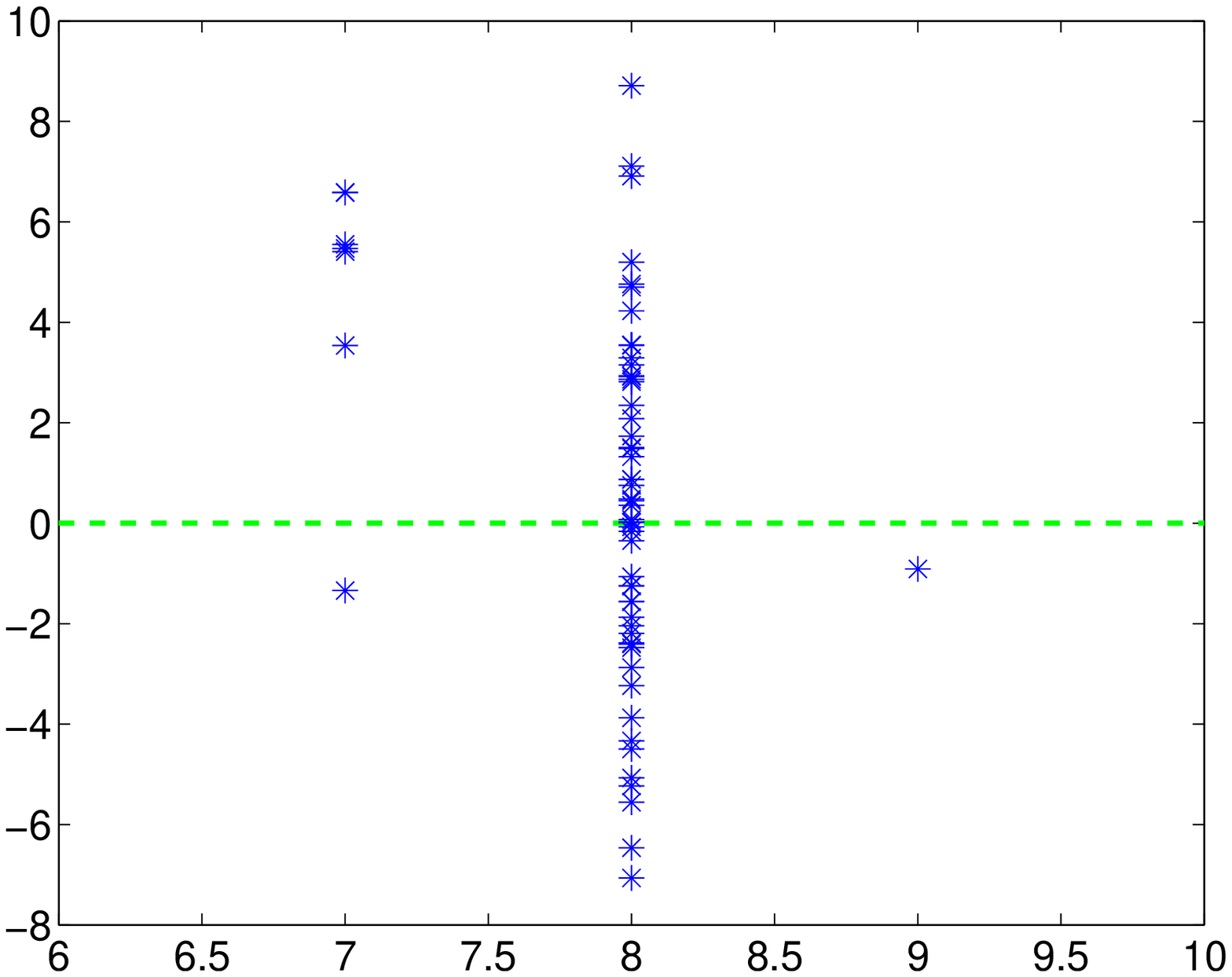}}\\
\scalebox{0.31}{\includegraphics{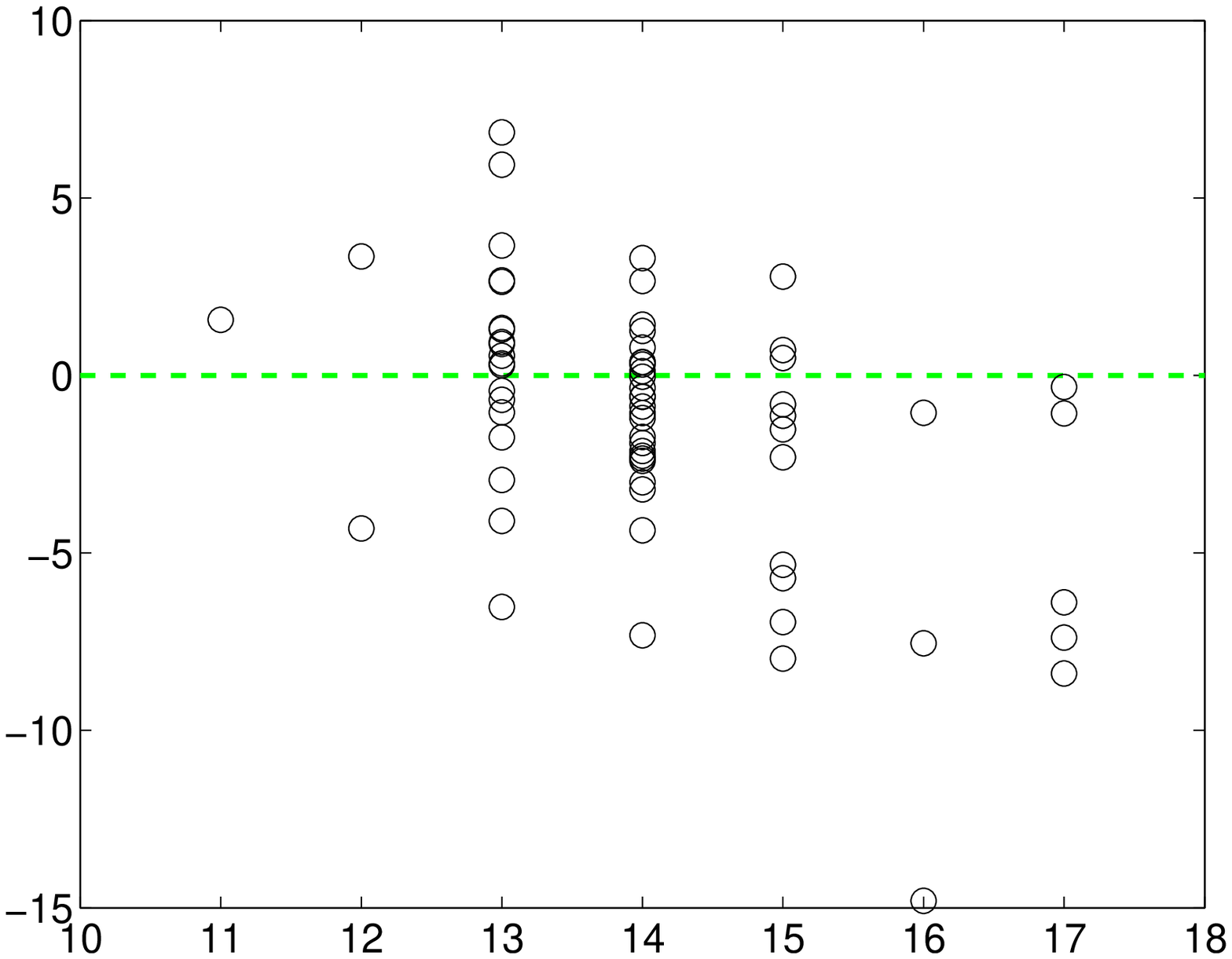}}&\scalebox{0.31}{\includegraphics{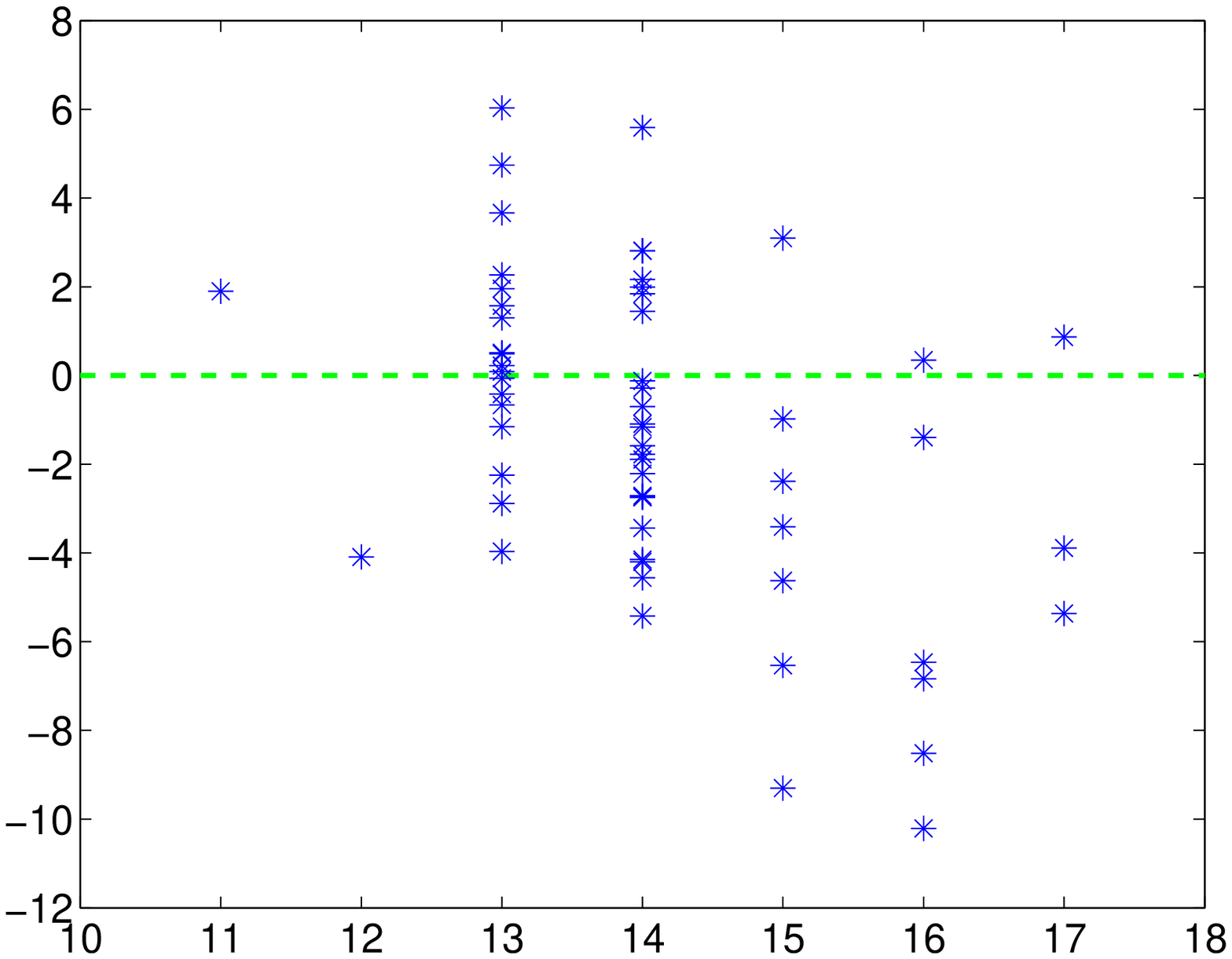}}\\
\scalebox{0.31}{\includegraphics{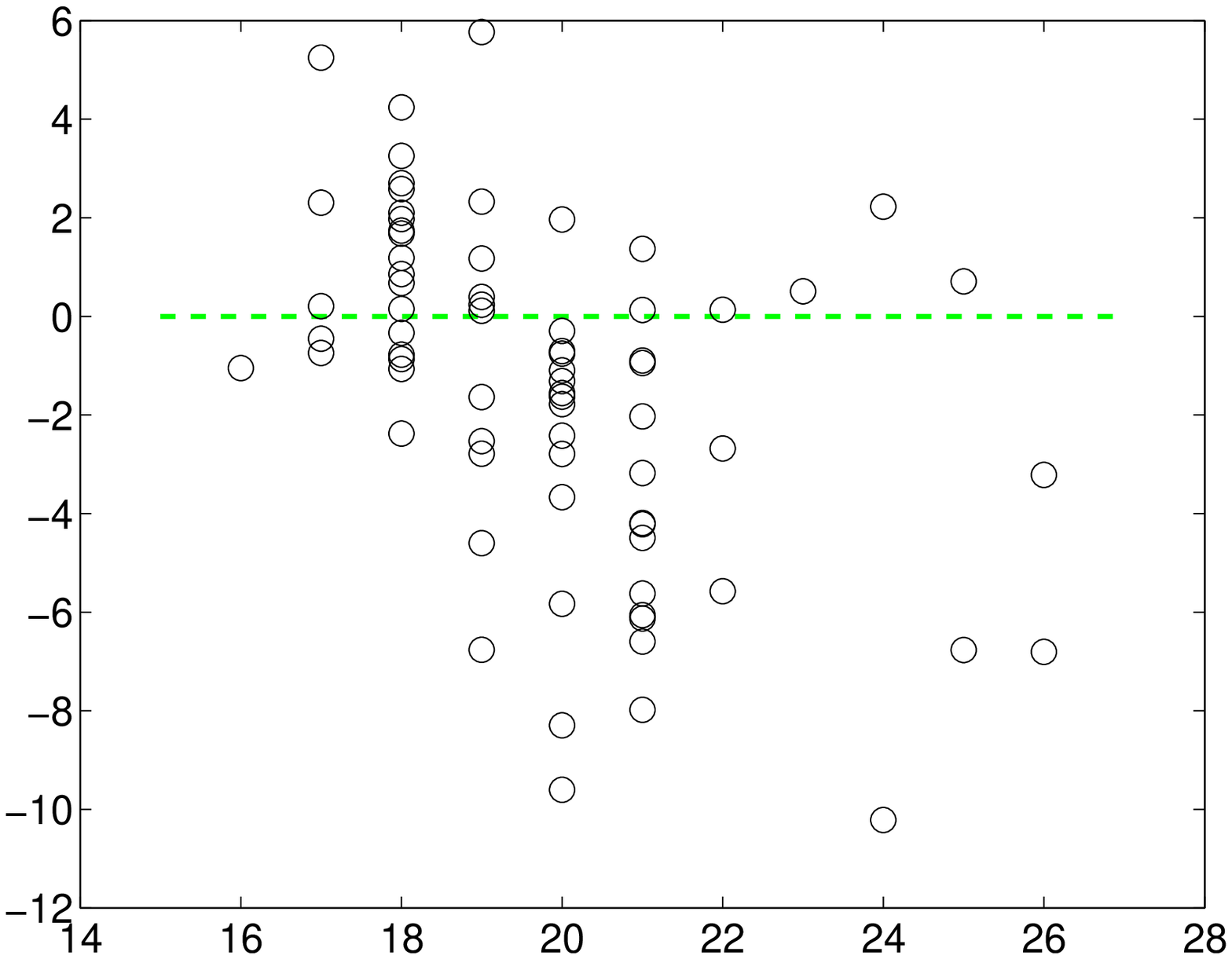}}&\scalebox{0.31}{\includegraphics{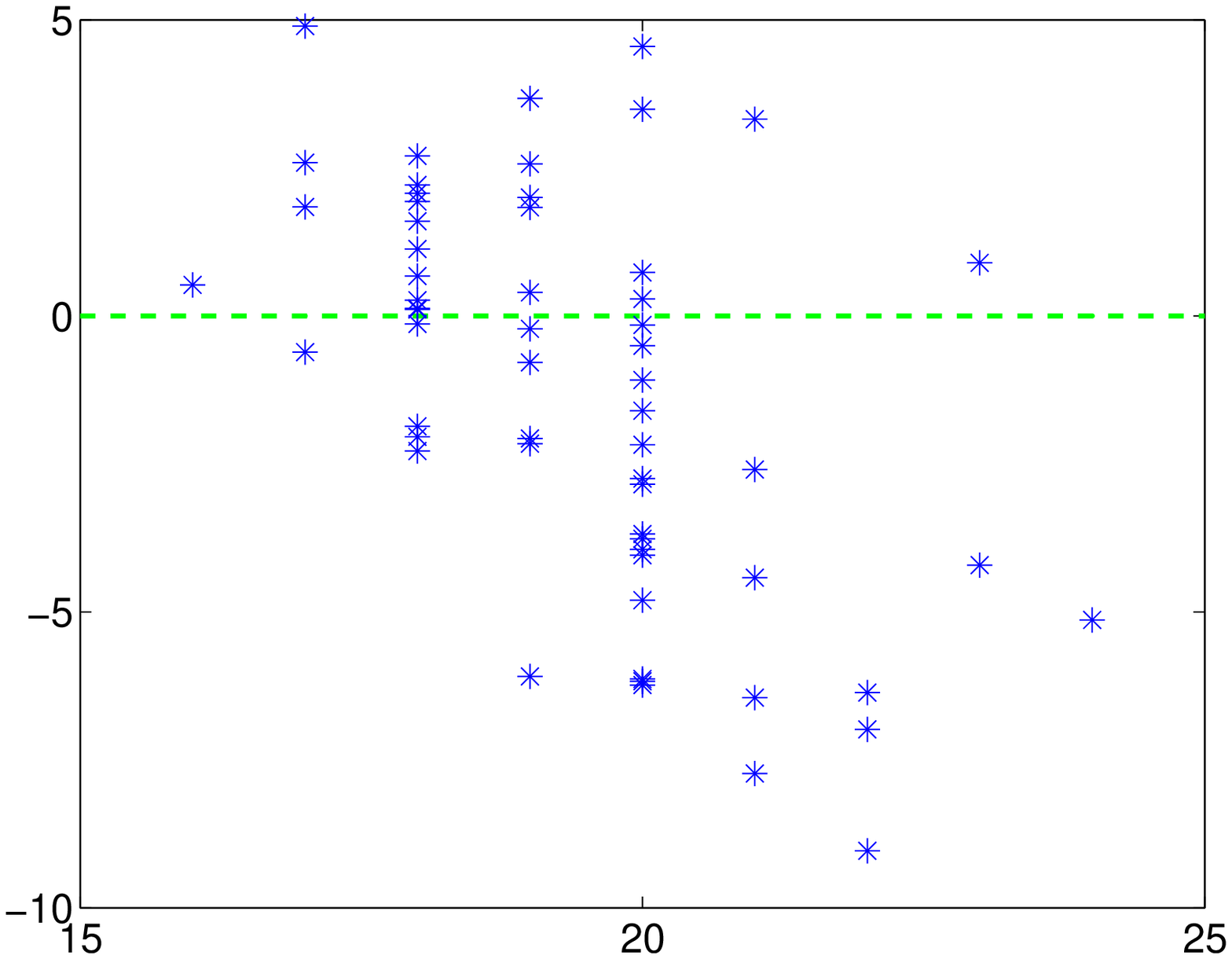}}\\
(a) Mangasarian &(b) Log-Det
\end{tabular}
\caption{The difference of SNR values of estimates by Algorithm~\ref{alg:matrix-final} and the BIHT vs. sparsity of reconstructed estimates by Algorithm~\ref{alg:matrix-final}. We fix $n=1000$ and $m=1000$.  Row 1 to Row 4 are corresponding to $s$ being 2, 8, 14, and 20, respectively.}
\label{fig:VaryingK-Alg-Cases}
\end{figure}


For the third configuration, the SNR values in decibels of the average reconstruction errors by both the BIHT and Algorithm~\ref{alg:matrix-final} are compared in Figure~\ref{fig:VaryingN-Alg} for fixed $m=1000$ and $s=10$ and varying dimensions of original signals. The plots in Figure~\ref{fig:VaryingN-Alg} show that the average SNR values for reconstructions by Algorithm~\ref{alg:matrix-final} are lower than that by the BIHT in most cases. This is due to the fact that the BIHT explores an unattainable additional information on the sparsity of the original signal. Another reason which we can see from Figure~\ref{fig:VaryingN-Alg-Cases} is that reconstructions with their sparsity larger than $10$ by Algorithm~\ref{alg:matrix-final} usually have lower SNR values than by the BIHT. The marks in each plot of Figure~\ref{fig:VaryingN-Alg-Cases} have the same meaning as that in Figures~\ref{fig:VaryingM-Alg-Cases} and \ref{fig:VaryingK-Alg-Cases}. For fixed $m=1000$ and $s=10$ we can draw conclusions from Figure~\ref{fig:VaryingN-Alg-Cases} that (i) Algorithm~\ref{alg:matrix-final} can give an estimate whose sparsity is smaller than that of the ideal sparse signal and (ii) Algorithm~\ref{alg:matrix-final}  with the Log-Det function works better than that with the Mangasarian function.

\begin{figure}[ht]
\centering
\begin{tabular}{cc}
\scalebox{0.32}{\includegraphics{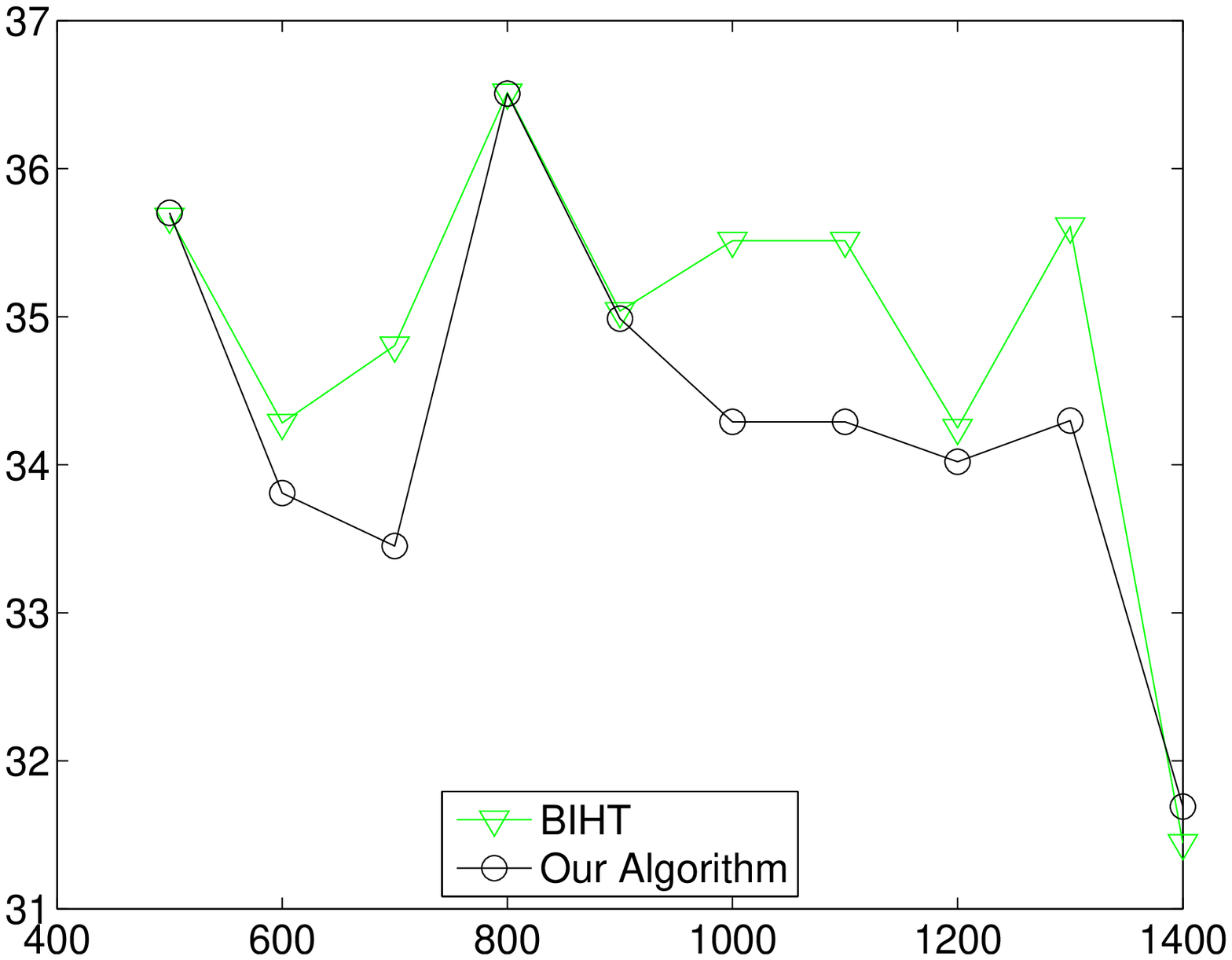}}&
\scalebox{0.32}{\includegraphics{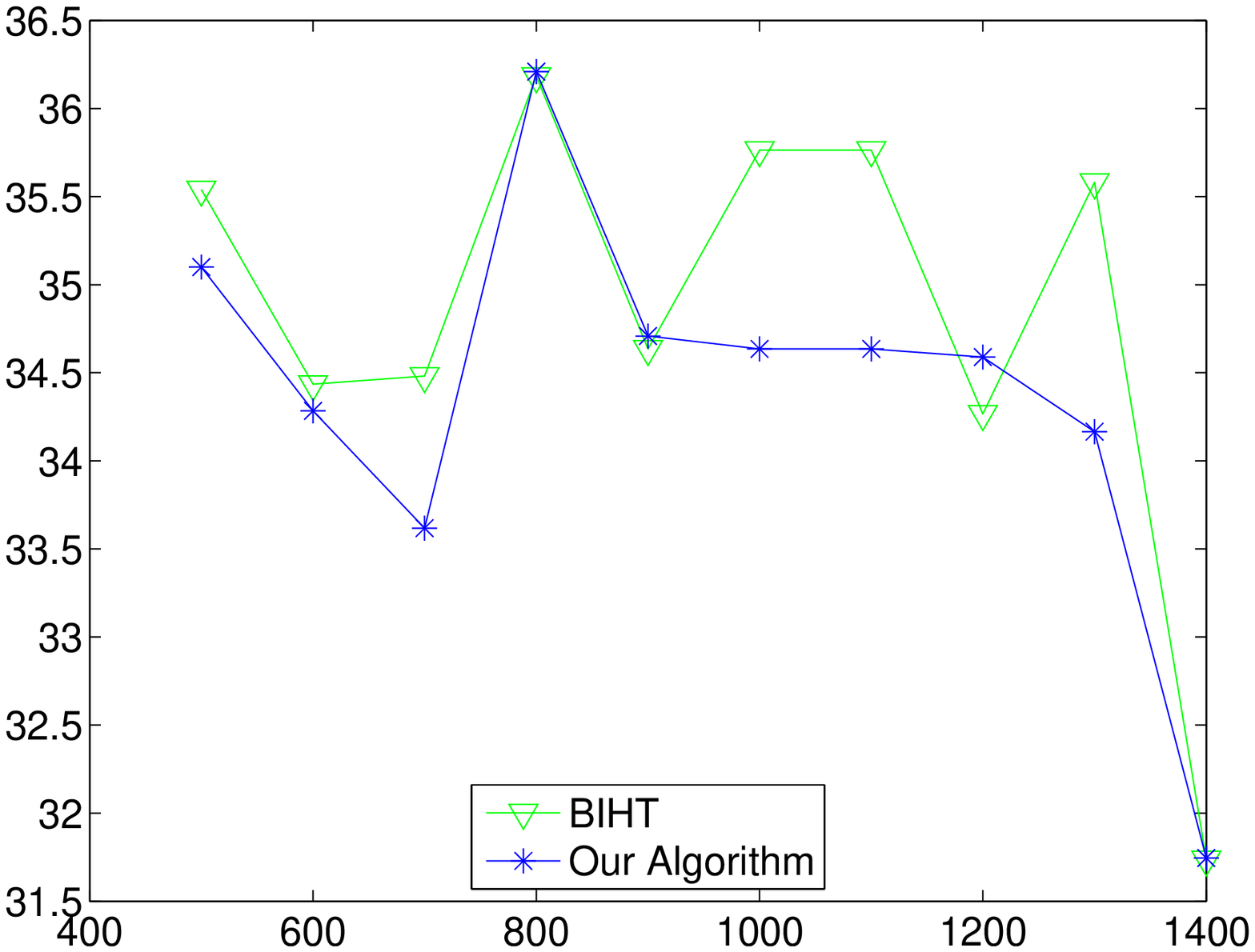}}\\
(a) Mangasarian &(b) Log-Det
\end{tabular}
\caption{Average SNR values of estimates vs. the signal size $n$ for fixed $m=1000$ and $s=10$. }
\label{fig:VaryingN-Alg}
\end{figure}


\begin{figure}[htb]
\centering

\begin{tabular}{cc}
\scalebox{0.31}{\includegraphics{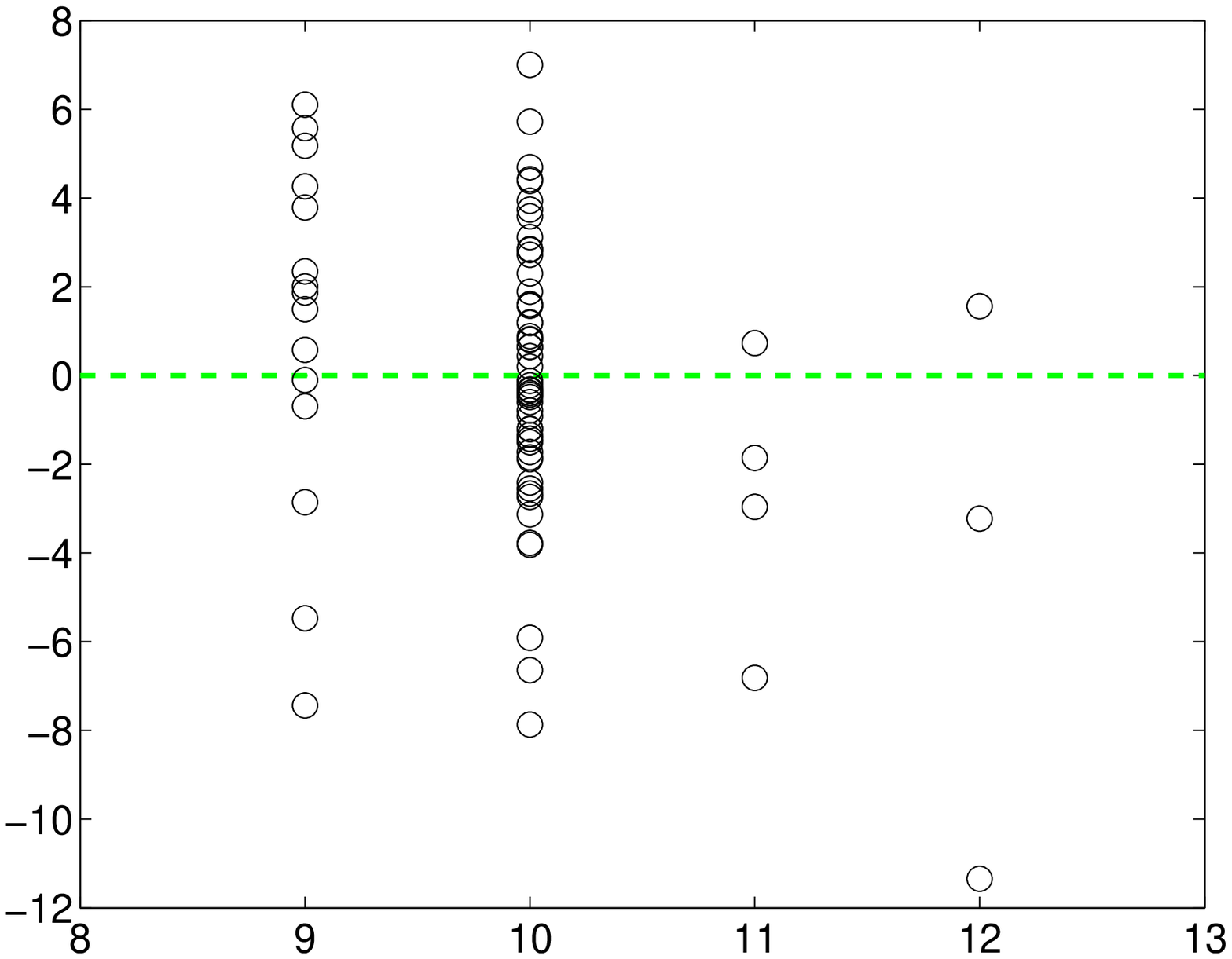}}&\scalebox{0.31}{\includegraphics{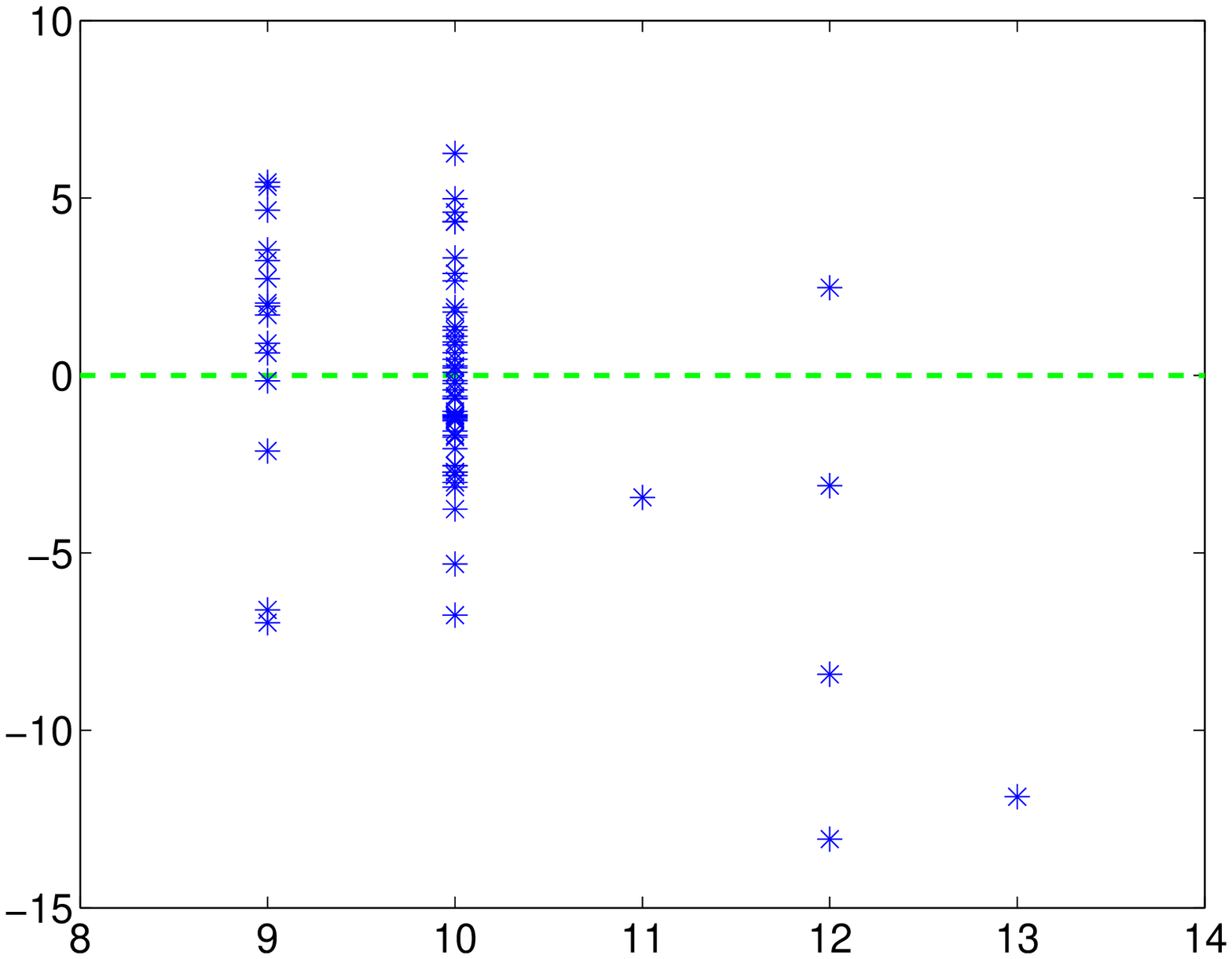}}\\
\scalebox{0.31}{\includegraphics{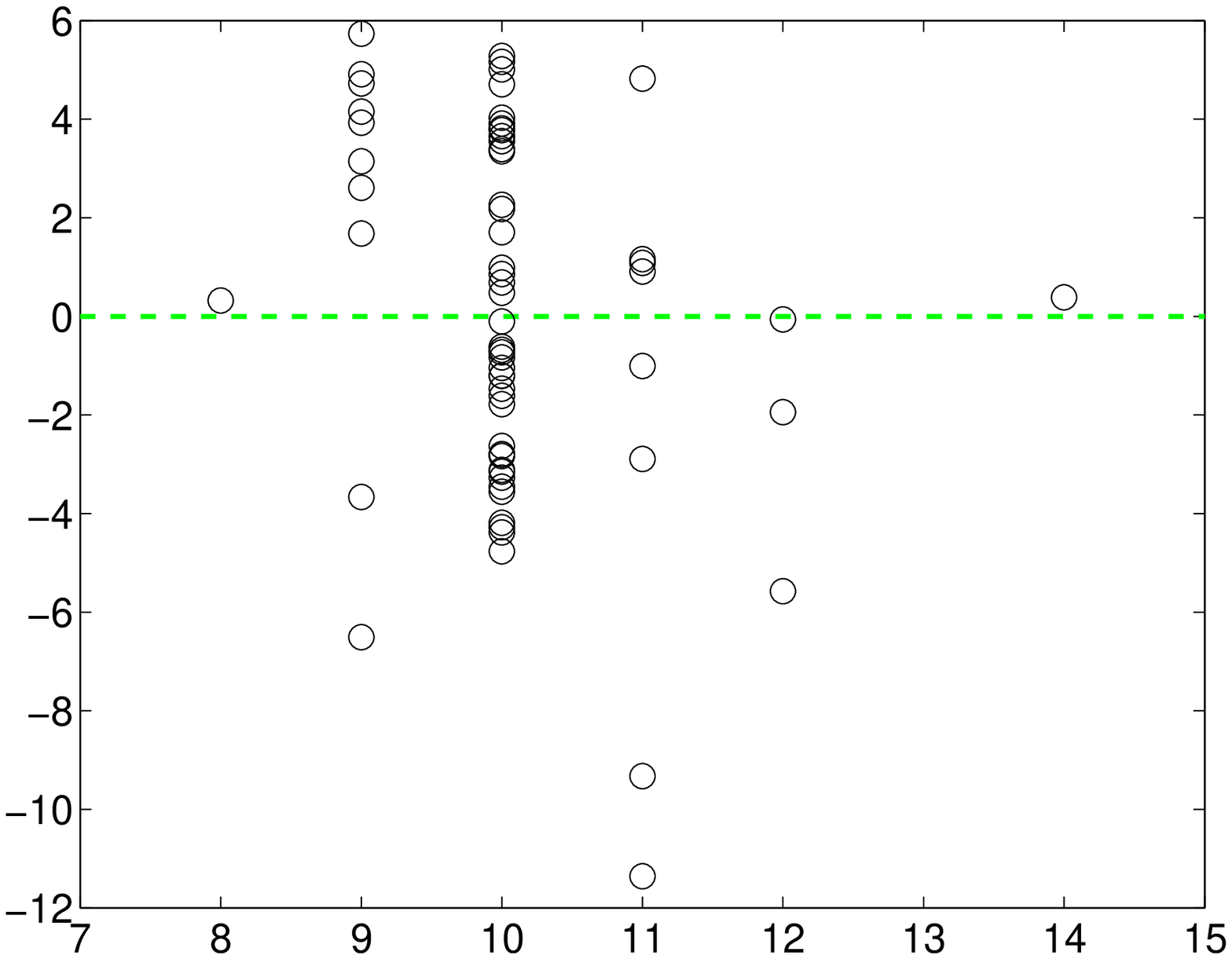}}&\scalebox{0.31}{\includegraphics{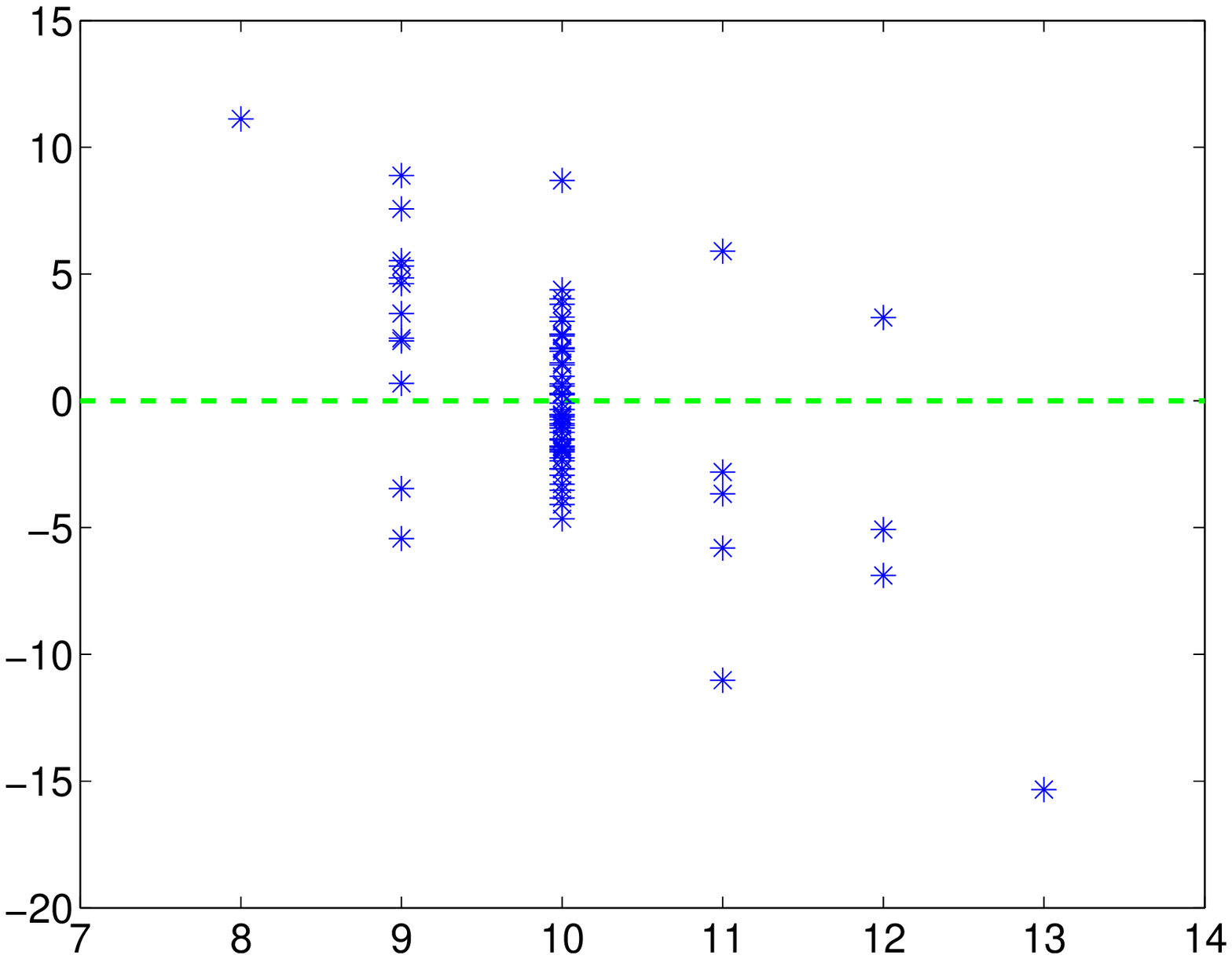}}\\
\scalebox{0.31}{\includegraphics{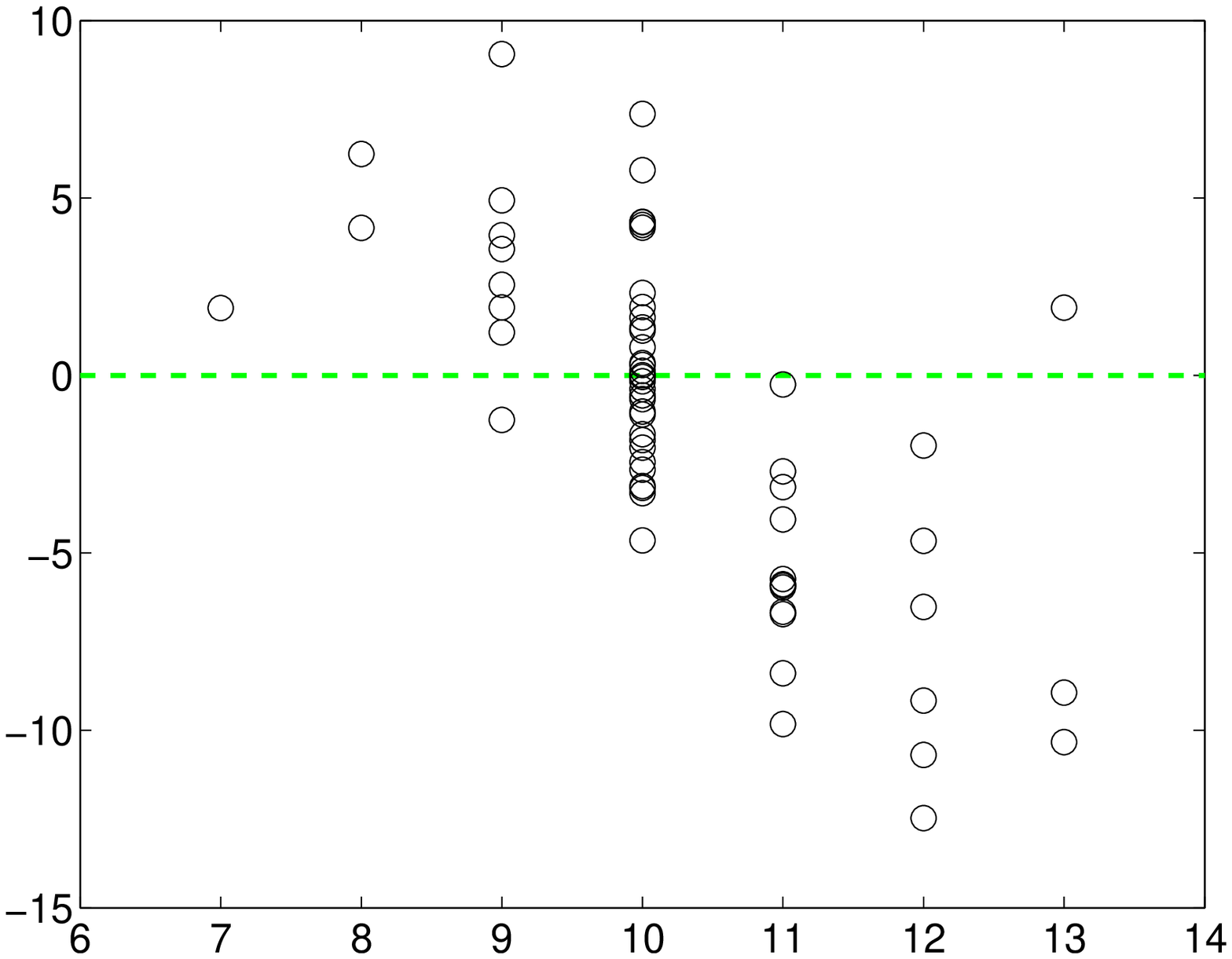}}&\scalebox{0.31}{\includegraphics{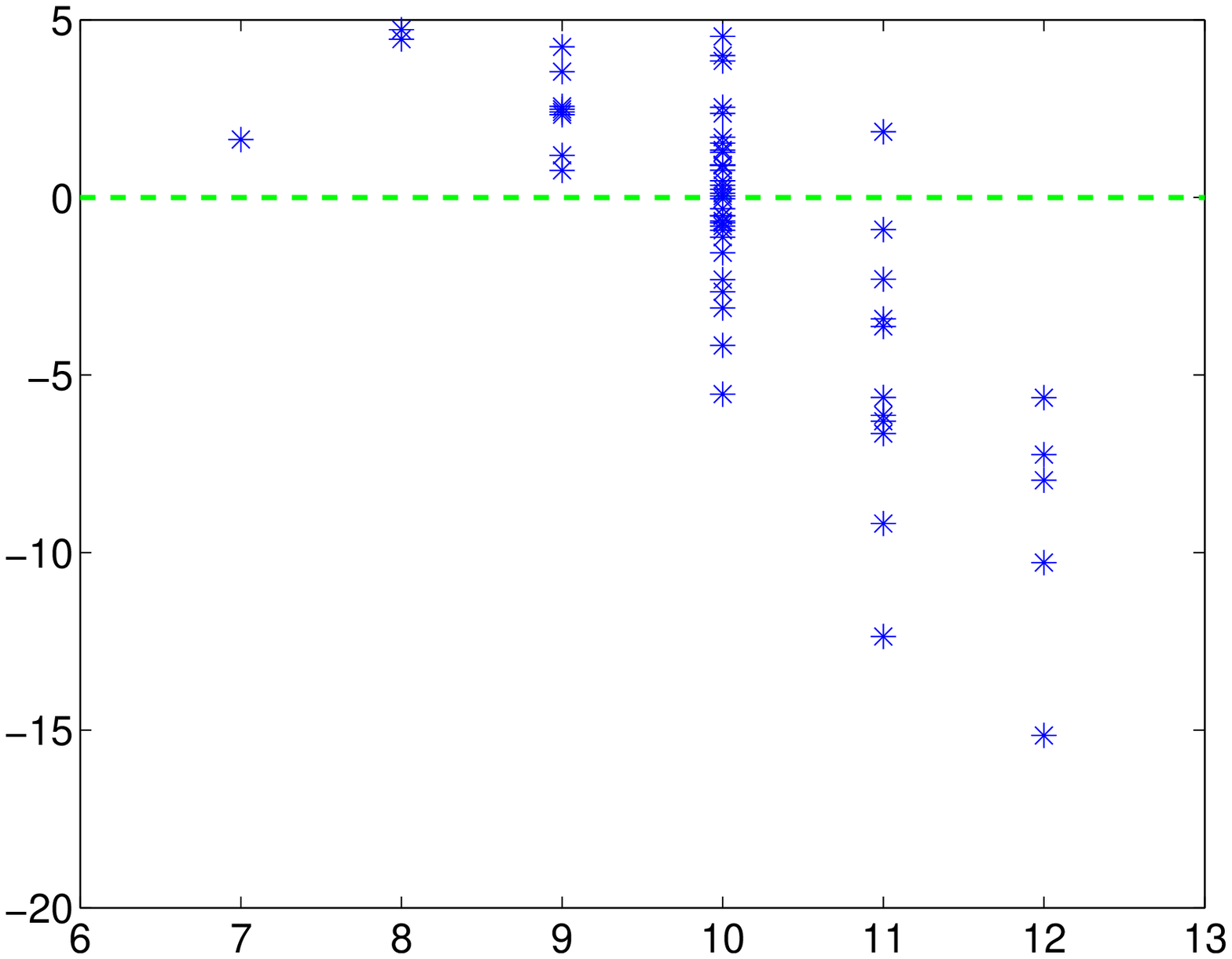}}\\
\scalebox{0.31}{\includegraphics{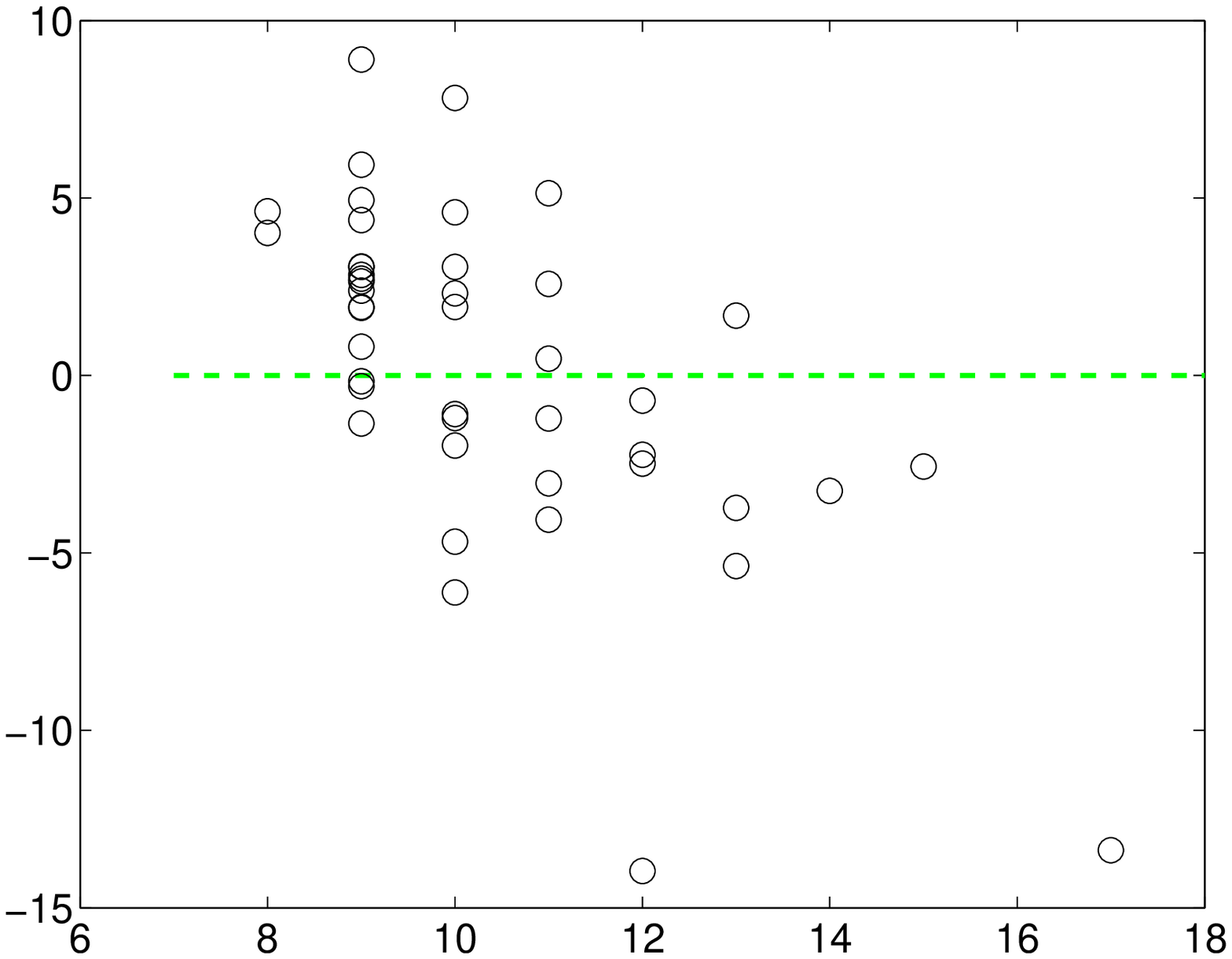}}&\scalebox{0.31}{\includegraphics{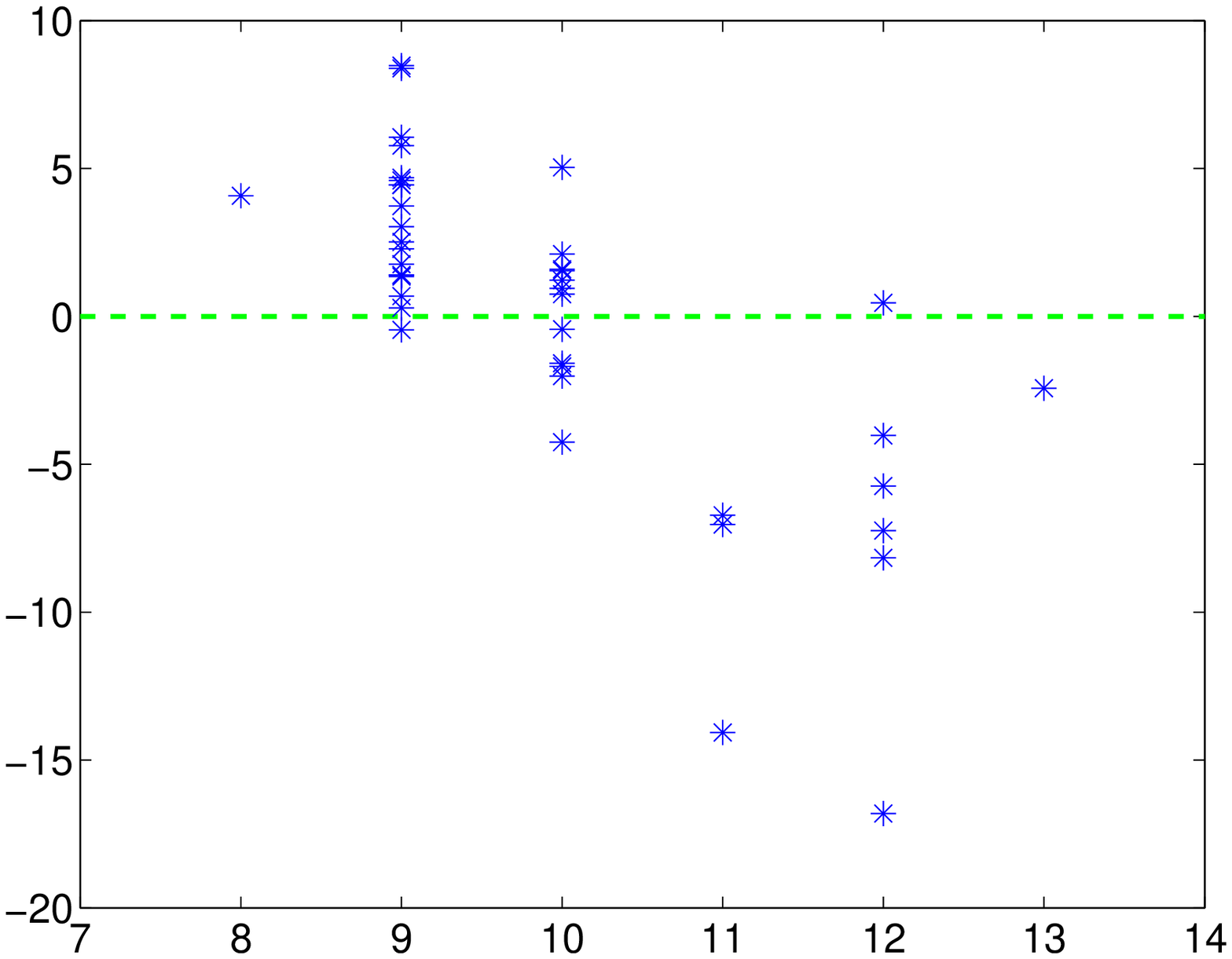}}\\
(a) Mangasarian &(b) Log-Det
\end{tabular}
\caption{The difference of SNR values of estimates by Algorithm~\ref{alg:matrix-final} and the BIHT vs. sparsity of reconstructed estimates by Algorithm~\ref{alg:matrix-final}. We fix $m=1000$ and $s=10$.  Row 1 to Row 4 are corresponding to $n$ being 500, 800, 1100, and 1400, respectively.}
\label{fig:VaryingN-Alg-Cases}
\end{figure}


\section{Summary and Conclusion}\label{sec:conclusion}
In this paper we proposed a new model and algorithm for 1-bit compressive sensing. Unlike the state-of-the-art BIHT method, our model does not need to know the sparsity of the signal of interest. We demonstrated the performance of our proposed algorithm for reconstruction from 1-bit measurements.

It would be of interest to study the convergence of Algorithm~\ref{alg:matrix-final} with the Mangasarian function in the future. It would be highly needed to adaptively update all parameters in Algorithm~\ref{alg:matrix-final} so that consistent reconstruction can be always achieved with improved accuracy of the solution.


\section*{Acknowledgements}
The authors would like to thank Ms. Na Zhang for her valuable comments and insightful suggestions which have
brought improvements to several aspects of this manuscript.

The views and conclusions contained herein are those of
the authors and should not be interpreted as necessarily representing
the official policies or endorsement, either expressed
or implied, of the Air Force Research Laboratory or the U.S.
Government.
\bibliographystyle{siam}


\end{document}